\documentclass[12pt,letterpaper]{article}
\usepackage{amsfonts}
\usepackage{amssymb}
\usepackage{amsmath}
\usepackage{amsthm}
\usepackage[round,authoryear]{natbib}
\usepackage[margin=1in]{geometry}
\usepackage{booktabs}
\usepackage{dcolumn}
\newcolumntype{d}[1]{D{.}{.}{#1}} 
\usepackage{graphicx,epstopdf}
\usepackage[T1]{fontenc}
\usepackage{lmodern}
\usepackage{enumitem}
\setlist[enumerate,1]{itemsep=1pt, topsep=4pt, partopsep=0pt}
\setlist[enumerate,2]{nosep}
\setlist[itemize,1]{itemsep=1pt, topsep=4pt, partopsep=0pt}
\setlist[itemize,2]{nosep}
\usepackage{subfig}
\usepackage{bm}
\usepackage[colorlinks,linkcolor=red,citecolor=blue,pdftex,breaklinks]{hyperref}
\hypersetup{bookmarksopen=true}
\usepackage[nameinlink]{cleveref}
\usepackage{setspace}

\theoremstyle{plain}
\newtheorem{theorem}{Theorem}
\newtheorem{corollary}{Corollary}
\newtheorem{lemma}{Lemma}[section]
\theoremstyle{definition}
\newtheorem{assumptionx}{Assumption}
\newenvironment{assumption}{\pushQED{\qed}\assumptionx}{\popQED\endassumptionx}
\newtheorem{remarkx}{Remark}
\newenvironment{remark}{\pushQED{\qed}\remarkx}{\popQED\endremarkx}

\newtheorem{algorithm}{Algorithm}

\Crefname{assumptionx}{Assumption}{Assumptions} 
\Crefname{examplex}{Example}{Examples} 
\Crefname{remarkx}{Remark}{Remarks} 

\makeatletter
\renewcommand*{\eqref}[1]{\hyperref[{#1}]{\textup{\tagform@{\ref*{#1}}}}}
\makeatother

\makeatletter
\DeclareRobustCommand\citepos
  {\begingroup\def\NAT@nmfmt##1{{\NAT@up##1's}}%
   \NAT@swafalse\let\NAT@ctype\z@\NAT@partrue
   \@ifstar{\NAT@fulltrue\NAT@citetp}{\NAT@fullfalse\NAT@citetp}}
\makeatother

\expandafter
\def \expandafter \normalsize \expandafter{\normalsize \setlength \abovedisplayskip{10pt plus 2pt minus 7pt}}
\expandafter
\def \expandafter \normalsize \expandafter{\normalsize \setlength \abovedisplayshortskip{0pt plus 2pt}}
\expandafter
\def \expandafter \normalsize \expandafter{\normalsize \setlength \belowdisplayskip{10pt plus 2pt minus 7pt}}
\expandafter
\def \expandafter \normalsize \expandafter{\normalsize \setlength \belowdisplayshortskip{5pt plus 2pt minus 3pt}}

\def\biA{{\bm{A}}}

\def\bib{{\bm{b}}}
\def\bbeta{{\bm\beta}}
\def\biX{{\bm{X}}}
\def\biZ{{\bm{Z}}}

\def\biM{{\bm{M}}}

\def\biH{{\bm{H}}}

\def\bis{{\bm{s}}}
\def\biy{{\bm{y}}}
\def\bix{{\bm{x}}}

\def\biu{{\bm{u}}}

\def\bSigma{{\bm{\Sigma}}}

\def\biQ{{\bm{Q}}}
\def\bXi{{\bm{\Xi}}}

\def\btheta{{\bm{\theta}}}
\def\bdelta{{\bm{\delta}}}

\def\E{{\rm E}}
\def\U{{\rm U}}
\def\IF{{\mathbb I}}
\def\N{{\rm N}}
\def\dto{\overset{d}  \longrightarrow}
\def\Pto{\overset{P}  \longrightarrow}

\def\plimN{\underset{N \to \infty}{\textrm{plim}}}
\def\limN{\underset{N \to \infty}{\lim}}
\def\tn{\kern -.08333em}
\def\tk{\kern 0.08333em}
\def\tkk{\kern 0.04167em}
\def\bzero{{\bm{0}}}
\def\bfI{{\bf I}}

\def\th#1{$#1^{\tk{\rm th}}$}
\def\H#1{{\textrm{H$_{#1}$}}}
\def\longto{\longrightarrow}

\DeclareMathOperator{\diag}{diag}
\DeclareMathOperator{\vech}{vech}
\DeclareMathOperator{\Var}{Var}
\DeclareMathOperator{\var}{Var}

\def\vec{{\textrm{vec}}}

\begin{document}


\title{Testing for the appropriate level of clustering\\in linear 
regression models\thanks{We are grateful to the editor, Serena Ng, an 
anonymous associate editor, three anonymous referees, Yevgeniy Feyman,
and participants at the 2018 Canadian Economics Association
conference, 2018 Canadian Econometric Study Group conference, 2020
Econometric Society World Congress, 2021 WEA conference, 2021 APPAM
conference, 2021 IAAE Annual Conference, Universit\'e de Montr\'eal,
University of Exeter, UCLA, UC Santa Barbara, Lakehead University,
Michigan State University, and Copenhagen Business School for
comments. MacKinnon and Webb thank the Social Sciences and Humanities
Research Council of Canada (SSHRC grants 435-2016-0871 and
435-2021-0396) for financial support. Nielsen thanks the Danish
National Research Foundation for financial support (DNRF Chair grant
number DNRF154). Computer code, including a \texttt{Stata} ado file, 
for performing the testing procedures proposed here may be found at
\url{http://qed.econ.queensu.ca/pub/faculty/mackinnon/svtest/}.
}}


\author{James G. MacKinnon\thanks{Corresponding author. Address: 
Department of Economics, 94 University Avenue, Queen's University, 
Kingston, Ontario K7L 3N6, Canada. Email:\ \texttt{mackinno@queensu.ca}.
Tel.\ 613-533-2293. Fax 613-533-6668.}\\Queen's University\\
\texttt{mackinno@queensu.ca}  \and
Morten \O rregaard Nielsen\\Aarhus University\\
\texttt{mon@econ.au.dk} \and
Matthew D. Webb\\Carleton University\\
\texttt{matt.webb@carleton.ca}}

\maketitle

\begin{abstract}
The overwhelming majority of empirical research that uses 
cluster-robust inference assumes that the clustering structure
is known, even though there are often several possible ways in which a 
dataset could be clustered. We propose two tests for the correct level
of clustering in regression models. One test focuses on inference about
a single coefficient, and the other on inference about two or more 
coefficients. We provide both asymptotic and wild bootstrap 
implementations. The proposed tests work for a null hypothesis of 
either no clustering or ``fine'' clustering against alternatives of 
``coarser'' clustering. We also propose a sequential testing procedure 
to determine the appropriate level of clustering. Simulations suggest 
that the bootstrap tests perform very well under the null hypothesis 
and can have excellent power. An empirical example suggests that using 
the tests leads to sensible inferences.

\medskip \noindent \textbf{Keywords:} CRVE, grouped data, clustered
data, cluster-robust variance estimator, robust inference, wild 
bootstrap, wild cluster bootstrap.

\medskip \noindent \textbf{JEL Codes:} C12, C15, C21, C23.

\end{abstract}

\clearpage
\onehalfspacing
\section{Introduction}
\label{sec:intro}

Modern empirical econometrics often allows for correlation within
clusters of observations, and this can have serious consequences for
statistical inference. Theoretical work on cluster-robust inference
almost always assumes that the structure of the clusters is known,
even though the form of the correlation within clusters is arbitrary.
Unless it is obvious that clustering must be at a certain level,
however, this can leave empirical researchers in a difficult
situation. They must generally rely on rules of thumb, their own
intuition, or referees' suggestions to decide how the observations
should be clustered. To make this process easier, we propose tests for
any given level of clustering (including no clustering as a special
case) against an alternative within which it is nested. When two or
more levels of clustering are possible, we propose a sequence of such
tests.


There has been a great deal of research on cluster-robust inference in
the past two decades. \citet*{CM_2015} cover much of the literature up
to a few years ago. \citet{Esarey_2019} and \citet{MW-survey} provide
more recent surveys. \citet*{CGH_2018} deal with a broader class of
methods for various types of dependent data. \citet*{MNW-guide}
provide a thorough and detailed guide to empirical practice. Areas
that have received particular attention include: asymptotic theory for
cluster-robust inference \citep*{DMN_2019,HansenLee_2019}; bootstrap
methods with clustered data \citep*{CGM_2008, DMN_2019, RMNW, 
MNW-bootknife}; and inference with unbalanced clusters 
\citep*{Imbens_2016, CSS_2017, MW-JAE, DMN_2019, MNW-influence}.

Almost all of this literature assumes that the way in which 
observations are allocated to clusters is known to the econometrician.
This is quite a strong assumption. Imagine that a dataset has many
observations taken from individuals in different geographical
locations. In order to utilize a cluster-robust variance estimator 
(CRVE), the researcher needs to specify at what level the clustering 
occurs. For example, there could possibly be clustering at the 
zip\tkk-code, city, county, state, or country level. Even in this 
relatively simple setting, there are many possible ways in which a 
researcher could `cluster' the standard errors.

A few rules of thumb have emerged to cover some common cases. For
instance, in the case of nested clusters, such as cities within
states, \citet*{CM_2015} advocate clustering at the larger, more
aggregate level. In the case of randomized experiments, 
\citet*{Athey_2017} recommend clustering at the level of
randomization. In the case of experiments where treatment is assigned
to groups in pairs, with one group treated and one not treated,
\citet{Chaisemartin_2022} recommend clustering at the pair level
rather than the group level. While these rules of thumb can sometimes
be very helpful, they may or may not lead to the appropriate
clustering level in any particular case.

Getting the level of clustering correct is extremely important. 
Simulation results in several papers have shown that ignoring
clustering in a single dimension can result in rejection frequencies
for tests at the 5\% level that are actually well over 50\%
\citep*{BDM_2004, CGM_2008} and confidence intervals that are too
narrow by a factor of five or more \citep{JGM-CJE}. On the other hand,
clustering at too coarse a level (say state\tkk-level clustering when
there is actually city-level clustering) can lead to the problems
associated with having few treated clusters, which can be severe
\citep{MW-JAE,MW-EJ}, and can also reduce power \citep{MW-survey}.

In \Cref{sec:tests}, we propose two tests for the cluster structure of
the error variance matrix in a linear regression model. They test the
null hypothesis of a fine level of clustering (or of no clustering at
all) against an alternative hypothesis with a coarser level of
clustering. The tests are based on the difference between two
functions of the scores for the parameter(s) of interest. These
functions are essentially the filling in the sandwich for two
different cluster-robust variance estimators, one associated with the
null level of clustering and one associated with the alternative
level. Since the functions estimate the variance of the scores under
two different clustering assumptions, we refer to the tests as
score\tkk-variance, or SV, tests. A procedure for sequential testing,
described in \Cref{subsec:level}, allows for determination of the
appropriate level of clustering without inflating the family-wise
error rate when there are several possible levels of clustering.

Tests for the appropriate level of clustering have also been proposed
by \citet{Ibragimov_2016} and recently by \citet{Cai_2022}. These
tests are very different from our tests and very different from each
other. We discuss them briefly in \Cref{subsec:other}.

The model of interest is discussed in \Cref{sec:model}. Our 
score\tkk-variance tests are described in \Cref{sec:tests}, including 
the bootstrap implementation, the sequential testing procedure, and
the use of our tests as pre\tkk-tests for inference about regression
coefficients. \Cref{sec:theory} provides asymptotic theory for the two
test statistics, the bootstrap tests, and the sequential testing
procedure. In \Cref{sec:partial}, we consider the common situation in
which the regressors that are not of primary interest have been
partialed out prior to performing the test. The size and power of the
proposed tests are analyzed by Monte Carlo simulations in
\Cref{sec:simulations}. An empirical example that deals with
clustering by classroom or school using the STAR dataset
\citep{Finn_1990,Mosteller_1995} is discussed in \Cref{sec:example}.
Finally, \Cref{sec:conclusion} concludes and offers some guidance for
empirical researchers. All mathematical proofs are given in
\Cref{sec:proofs}.

\section{The Regression Model with Clustering}
\label{sec:model}

We focus on the linear regression model
\begin{equation}
\label{model} 
\biy = \biX \bbeta + \biu,
\end{equation}
where $\biy$ and $\biu$ are, respectively, $N \times 1$ vectors of
observations and disturbances (or error terms), and $\biX$ is an $N
\times k$ matrix of regressors (or covariates). The $k \times 1$
parameter vector $\bbeta$ contains the coefficients on the regressors.

Suppose that the data are divided into $G$ clusters, indexed by $g$,
where the \th{g} cluster has $N_g$ observations, so that $N =
\sum_{g=1}^G N_g$. Thus, there are $G$ vectors $\biy_g$ and $\biu_g$ of
size $N_g$, along with $G$ matrices $\biX_g$, each with $N_g$ rows
and $k$ columns. Using this notation, the ordinary least squares (OLS) 
estimator of $\bbeta$ is
\begin{equation}
\label{betahat}
\hat\bbeta = (\biX^\top\biX)^{-1}\biX^\top\biy
= \bbeta_0 +  (\biX^\top\biX)^{-1}\biX^\top\biu
=  \bbeta_0 +  (\biX^\top\biX)^{-1}\sum_{g=1}^G \biX_g^\top\biu_g,
\end{equation}
where $\bbeta_0$ denotes the true value of $\bbeta$. Now define the 
$k \times 1$ score vectors $\bis_g = \biX_g^\top \biu_g$. We assume that 
these score vectors satisfy $\E(\bis_g)=\bzero$ for all $g$ and
\begin{equation}
\label{def Sigma}
\E(\bis_g\bis_{g'}^\top)= \IF (g=g') \bSigma_g, \quad g,g'=1,\ldots,G,
\end{equation}
where $\IF (\cdot)$ denotes the indicator function and $\bSigma_g$ is
a $k \times k$ variance matrix. Although the properties of the
$\bSigma_g$ depend on the properties of the variance matrix of $\biu$,
we do not explicitly make any assumptions about the latter because our
tests are concerned solely with the variances of the score vectors.

It is clear from \eqref{betahat} that the asymptotic distribution of
$\hat\bbeta$ depends on the asymptotic distribution of the score
vectors. An estimator of the variance matrix of $\hat\bbeta$ is given
by the sandwich formula
\begin{equation}
\label{covbeta}
\widehat\var (\hat\bbeta ) = 
(\biX^\top\biX)^{-1} \hat\bSigma (\biX^\top\biX)^{-1}\tn,
\end{equation}
where $\hat\bSigma$ is an estimator of the variance matrix of the sum 
of scores, $\bSigma = \E (\biX^\top\biu\biu^\top\!\biX )$. The condition 
\eqref{def Sigma} implies that $\E (\bis_g \bis_{g'}^\top) = \bzero$ 
whenever $g \neq g'$. In this case $\bSigma =\sum_{g=1}^G \bSigma_g$, so
that the usual estimator for $\bSigma$ under condition \eqref{def Sigma} is
\begin{equation}
\label{Sighat}
\hat\bSigma_{\rm c} = m_{\rm c} \sum_{g=1}^G \biX_g^\top\hat\biu_g
\hat\biu_g^\top\biX_g = m_{\rm c} \sum_{g=1}^G \hat\bis_g\hat\bis^\top_g,
\end{equation}
where $\hat\biu_g$ contains the residuals for cluster $g$ and
$\hat\bis_g = \biX_g^\top \hat\biu_g$ is the $k \times 1$ vector of 
empirical scores for cluster~$g$. The scalar factor $m_{\rm c}$ is a 
finite\tkk-sample correction, the most commonly employed factor being 
$m_{\rm c}=G/(G-1)\times (N-1)/(N-k)$, which is designed to account
for degrees of freedom. Using $\hat\bSigma = \hat\bSigma_{\rm c}$ in
\eqref{covbeta} yields CV$_1$, the most widely-used CRVE
for~$\hat\bbeta$. Asymptotic inference on regression coefficients
using CV$_1$ is studied by \citet{DMN_2019} and
\citet{HansenLee_2019}.

\begin{remark}
\label{rem:hetmodel}
In the special case in which each cluster has $N_g =1$ observation, we 
can use
\begin{equation}
\label{HCmid}
\hat\bSigma_{\rm het} = \sum_{i=1}^N \hat{u}^2_i\tk \biX_i^\top\biX_i
   = \biX^\top\tn \diag (\hat{u}^2_1,\dots,\hat{u}^2_N) \biX\tn,
\end{equation}
where $\biX_i$ is the \th{i} row of the $\biX$ matrix and $\hat{u}_i$
is the \th{i} residual. The variance matrix obtained by setting 
$\hat\bSigma = \hat\bSigma_{\rm het}$ in \eqref{covbeta} is the famous
heteroskedasticity-consistent variance matrix estimator (HCCME) of
\citet{Eicker_1963} and \citet{White_1980}. Of course, the matrix
$\hat\bSigma_{\rm het}$ can be modified in various ways to improve its
finite\tkk-sample properties \citep{MW_1985,JGM_2013}. The simplest is
to multiply it by \mbox{$m_{\rm het}=N/(N-k)$}, so that
\eqref{covbeta} becomes what is usually called HC$_1$.
\end{remark}

\begin{remark}
\label{rem:AAIW} As \citet*{AAIW_2023} point out, when the object of
interest is the average treatment effect in a finite population,
cluster-robust standard errors based on \eqref{Sighat} can be
``unnecessarily conservative.'' Consequently, they develop an approach
to inference that depends both on how the data were sampled and on how
treatment was assigned. In this paper, however, we follow most of the
literature on cluster-robust inference and rely on the traditional
approach in which every sample is treated as a random outcome from a
data-generating process (DGP). The objective is to draw inferences
about the parameters of the DGP, which may be interpreted as features
of an infinitely large population; see \citet{MNW-guide} for
additional details.
\end{remark}

\section{The Testing Procedure}
\label{sec:tests}

The fundamental idea of our testing procedure is to compare two
estimates of the variance of the coefficient(s) that we want to
estimate. We test the null hypothesis that a CRVE based on a ``fine''
clustering structure is valid against the alternative that the CRVE
needs to be based on a ``coarser'' clustering structure. Since it is
only the filling in the sandwich \eqref{covbeta} that differs across
different clustering structures, we are actually comparing two
estimates of the variance matrix of the sum of scores. Our procedure
is somewhat like the specification test of \citet{Hausman_1978}. The
``fine'' CRVE is efficient when there actually is fine clustering, but
it is invalid when there is coarse clustering. In contrast, the
``coarse'' CRVE is inefficient when there actually is fine clustering,
but it is valid in both cases.

To make our testing procedure operational, we formulate it in terms of
the parameters of the model. To this end, we first define some
notation. There are $G$ coarse clusters indexed by $g=1,\ldots,G$.
Within coarse cluster~$g$, there are $M_g$ fine clusters indexed by
$h=1,\ldots,M_g$.  In total there are $G_{\rm f} = \sum_{g=1}^G M_g$
fine clusters. Fine cluster $h$ in coarse cluster $g$ contains
$N_{gh}$ observations indexed by $i=1,\ldots,N_{gh}$. Coarse cluster
$g$ therefore contains $N_g = \sum_{h=1}^{M_g} N_{gh}$ observations,
and the entire sample contains $N=\sum_{g=1}^G N_g =\sum_{g=1}^G
\sum_{h=1}^{M_g} N_{gh}$ observations. We let $\biX_{ghi}$ and
$u_{ghi}$ denote the regressors and disturbance for observation $i$
within fine cluster $h$ in coarse cluster~$g$. We then define the
corresponding score as $\bis_{ghi}= \biX_{ghi}^\top u_{ghi}$, the
score for fine cluster $h$ in coarse cluster $g$ as
$\bis_{gh}=\sum_{i=1}^{N_{gh}} \bis_{ghi}$, and the score for coarse
cluster~$g$ as $\bis_g = \sum_{h=1}^{M_g}\bis_{gh}$.

Under the coarse clustering structure, the $\bis_g$ satisfy 
\eqref{def Sigma}, so that in particular they are uncorrelated
across~$g$. Under the fine clustering structure, the $\bis_{gh}$ are
themselves uncorrelated across~$h$, for each~$g$. That is, for all
$g=1,\ldots ,G$,
\begin{equation}
\label{def Sigma gh}
\E (\bis_{gh}\bis_{gh'}^\top ) = \IF (h=h') \bSigma_{gh},
\quad h,h' = 1,\ldots ,M_g ,
\end{equation}
where each of the $\bSigma_{gh}$ is a $k \times k$ matrix. Thus, 
\eqref{def Sigma} and \eqref{def Sigma gh} embody the assumption that 
the fine clustering structure is nested within the coarse one. Another
possible design would have a one\tkk-way clustering structure nested
within a two\tkk-way one; see \Cref{rem:two-way}.

Now let $\bSigma_{\rm c}$ and $\bSigma_{\rm f}$ denote the matrix 
$\bSigma$ under the coarse and fine clustering structures,
respectively.  From \eqref{def Sigma} and \eqref{def Sigma gh}, these
matrices are
\begin{equation}
\label{Sigmas}
\bSigma_{\rm c} = \sum_{g=1}^G \bSigma_g
\quad\textrm{and}\quad
\bSigma_{\rm f} = \sum_{g=1}^G \sum_{h=1}^{M_g}\bSigma_{gh}.
\end{equation}
We consider the null and alternative hypotheses
\begin{equation}
\label{hypotheses}
\H{0}\!: \limN \tk \bSigma_{\rm f} \tk \bSigma^{-1}_{\rm c} = \bfI 
\quad \textrm{and} \quad
\H{1}\!: \limN \tk \bSigma_{\rm f} \tk \bSigma^{-1}_{\rm c} \neq \bfI.
\end{equation}
The hypotheses are expressed in this way, rather than in terms of the
difference between the limits of normalized versions of $\bSigma_{\rm f}$ 
and $\bSigma_{\rm c}$, because the appropriate normalizing factors will,
in general, be unknown; see \citet{DMN_2019}.

\begin{remark}
\label{rem:hyp}
In \eqref{hypotheses}, we are not directly testing the fine clustering
condition in \eqref{def Sigma gh}. Instead, we are testing an
important implication of the clustering structure. Specifically, we
test whether $\bSigma_{\rm c}=\bSigma_{\rm f}$, which implies that
a valid CRVE for $\hat\bbeta$ is given by \eqref{covbeta} with
$\hat\bSigma = \hat\bSigma_{\rm f}$.
\end{remark}

\begin{remark}
\label{rem:hetnull}
An important null hypothesis is that no CRVE is needed because the
HCCME considered in \Cref{rem:hetmodel}, obtained by combining
\eqref{covbeta} and \eqref{HCmid}, is valid. In this case, each fine
cluster has just one observation, so that $M_g=N_g$ and $N_{gh}=1$ for
all $g$ and~$h$.
\end{remark}

\begin{remark}
\label{rem:partial}
In practical applications, the number of coefficients in regression
models, and hence the size of the CRVE matrices, is often large, so
that comparing these matrices directly can be impractical.
Furthermore, it is usually only one coefficient, or a small subset of
them, that is actually of interest. Many coefficients typically
correspond to fixed effects and other conditioning variables that are
not of primary interest. By partialing out the latter, it is possible
to reduce the dimensionality of the test and focus on the parameter(s)
of interest. We discuss this issue in \Cref{sec:partial}.
\end{remark}

\subsection{Test Statistics}
\label{sec:stats}

Our score\tkk-variance, or SV, test statistics are based on comparing
estimates $\hat\bSigma_{\rm f}$ and $\hat\bSigma_{\rm c}$ obtained
under fine and coarse clustering, respectively. There are many ways in
which one could compare these $k \times k$ matrices. We focus on two
quantities of particular interest, which define two test statistics.
The first is obtained for $k=1$. This could be after all regressors 
except one have been partialed out (\Cref{sec:partial}), so that 
interest is focused on a particular coefficient that we are trying to 
make inferences about. This leads to a test statistic with the form of
a $t$-statistic. The second is obtained for $k > 1$, in which case 
our test statistic is a quadratic form involving all the unique
elements of $\hat\bSigma_{\rm f}$ and $\hat\bSigma_{\rm c}$, as in
\citepos{White_1980} ``direct test'' for heteroskedasticity. The first
test is of course a special case of the second, but we treat it
separately because it is particularly simple to compute and may often
be of primary interest.

In order to derive the test statistics, we write $\hat\bSigma_{\rm c}$ 
and $\hat\bSigma_{\rm f}$ using common notation. Let $\hat\bis_{ghi}$
denote the empirical score for observation $i$ within fine cluster $h$ 
in coarse cluster $g$, and let $\hat\bis_{gh} = \sum_{i=1}^{N_{gh}}
\hat\bis_{ghi}$ denote the empirical score for fine cluster $h$ in 
coarse cluster $g$, such that $\hat\bis_g = \sum_{h=1}^{M_g}
\hat\bis_{gh}$. Under coarse clustering, the estimated $\bSigma_{\rm c}$
matrix in \eqref{Sighat} is
\begin{equation}
\label{Sigmac}
\hat\bSigma_{\rm c} 
= m_{\rm c} \sum_{g=1}^G \hat\bis_g\hat\bis^\top_g
= m_{\rm c} \sum_{g=1}^G \left( \sum_{h=1}^{M_g} 
\hat\bis_{gh} \right)\!
\left( \sum_{h=1}^{M_g} \hat\bis_{gh} \right)^{\!\!\!\top}\!.
\end{equation}
Similarly, we can write, c.f.\ \eqref{def Sigma gh} and \eqref{Sigmas},
\begin{equation}
\label{Sigmaf}
\hat\bSigma_{\rm f} = m_{\rm f} \sum_{g=1}^G \sum_{h=1}^{M_g} 
\hat\bis_{gh} \hat\bis_{gh}^\top ,
\end{equation}
where $m_{\rm f} = G_{\rm f}/(G_{\rm f}-1) \times (N-1)/(N-k)$.

When interest focuses on just one coefficient, so that $k=1$,
the matrix $\biX$ becomes the vector $\bix$, and the empirical scores
are scalars. Specifically, $\hat s_{ghi} = x_{ghi} \hat u_{ghi}$
and $\hat s_{gh}=\sum_{i=1}^{N_{gh}} \hat s_{ghi}$ denote the
empirical scores for observation~$i$ and fine cluster~$h$,
respectively. Then the matrices \eqref{Sigmac} and \eqref{Sigmaf}
reduce to the scalars
\begin{equation}
\label{Sigmascalar}
\hat\sigma^2_{\rm c} = m_{\rm c} \sum_{g=1}^G 
\left( \sum_{h=1}^{M_g} \hat s_{gh} \right)^{\!\!2}
\quad\textrm{and}\quad
\hat\sigma^2_{\rm f} = m_{\rm f} \sum_{g=1}^G \sum_{h=1}^{M_g} 
\hat s_{gh}^2 .
\end{equation}
The quantities given in \eqref{Sigmac}, \eqref{Sigmaf}, and
\eqref{Sigmascalar} are all defined in essentially the same way. They
simply amount to different choices of empirical scores. If
$N_{gh}=1$, then $\hat\sigma^2_{\rm f}$ simplifies to
\begin{equation}
\label{sigmahet}
\hat\sigma^2_{\rm het} =
\sum_{g=1}^G\sum_{h=1}^{M_g}\sum_{i=1}^{N_{gh}} \hat s^2_{ghi},
\end{equation}
which is just the sum of the squared empirical scores over all the
observations.

Our first test is based on the difference between the two scalars in 
\eqref{Sigmascalar}, namely,
\begin{equation}
\label{thetascalar}
\hat\theta = \hat\sigma^2_{\rm c} - \hat\sigma^2_{\rm f}.
\end{equation}
Our second test is based on the difference between the $k\times k$
matrices $\hat\bSigma_{\rm c}$ and~$\hat\bSigma_{\rm f}$. For this
test, we consider the vector of contrasts,
\begin{equation}
\label{thetaSigma}
\hat\btheta = \vech (\hat\bSigma_{\rm c} - \hat\bSigma_{\rm f}),
\end{equation}
where the operator $\vech (\cdot)$ returns a vector, of dimension 
$k(k+1)/2$ in this case, with all the supra-diagonal elements of the
symmetric $k\times k$ matrix argument removed.

In order to obtain test statistics with asymptotic distributions that 
are free of nuisance parameters, we need to derive the asymptotic
means and variances of $\hat\theta$ and $\hat\btheta$, so that we can
studentize the statistics in \eqref{thetascalar} and
\eqref{thetaSigma}. To this end, suppose that we observe the (scalar)
scores $s_{gh}$ for fine cluster $h$ in coarse cluster~$g$. Then the 
analog of $\hat\theta$ is the contrast
\begin{equation}
\label{thetaknown}
\theta = \sum_{g=1}^G\sum_{h_1=1}^{M_g}\sum_{h_2\neq h_1}^{M_g} 
 s_{gh_1} s_{gh_2}.
\end{equation}
This is simply the sum of all the cross\tkk-products of scores that
are in the same coarse cluster but different fine clusters. Under the
null hypothesis, $\theta$ clearly has mean zero 
by~\eqref{def Sigma gh}.

The variance of $\theta$ in \eqref{thetaknown} is, under the null
hypothesis,
\begin{equation}
\label{var1}
\var (\theta ) = \sum_{g=1}^G\sum_{h_1=1}^{M_g} \sum_{\ell_1=1}^{M_g}
\sum_{h_2\neq h_1}^{M_g}\sum_{\ell_2\neq \ell_1}^{M_g} 
\E ( s_{gh_1}s_{g\ell_1}s_{gh_2}s_{g\ell_2}).
\end{equation}
The expectation of any product of scores can only be nonzero, under
the null and \eqref{def Sigma gh}, when their indices are the same in 
pairs. This implies that either $h_1=\ell_1 \neq h_2=\ell_2$ or 
$h_1=\ell_2 \neq h_2=\ell_1$. These cases are symmetric, and hence 
\eqref{var1} simplifies to
\begin{equation}
\label{tauvar}
\var(\theta) = 2 \sum_{g=1}^G \sum_{h_1=1}^{M_g}\sum_{h_2\neq h_1}^{M_g}
\sigma^2_{gh_1}\sigma^2_{gh_2} ,
\end{equation}
where $\sigma^2_{gh} = \var ( s_{gh})$ is used to denote $\bSigma_{gh}$
in the scalar case.

The sample analog of the right-hand side of \eqref{tauvar} is $2
\sum_{g=1}^G\sum_{h_1=1}^{M_g} \sum_{h_2\neq h_1}^{M_g}\hat s_{gh_1}^2
\hat s_{gh_2}^2$; see~\eqref{def Sigma gh}. This suggests the variance
estimator
\begin{equation}
\label{varfast}
\widehat\var (\hat\theta) =
2 \sum_{g=1}^G\left(\sum_{h=1}^{M_g} \hat s^2_{gh}\!\right)^{\!\!2} 
- 2 \sum_{g=1}^G \sum_{h=1}^{M_g} \hat s^4_{gh}.
\end{equation}
This equation avoids the triple summation in \eqref{tauvar} by
squaring the sums of squared empirical scores, which then requires
that the second term be subtracted. In deriving \eqref{varfast}, we 
have ignored the factors $m_{\rm c}$ and $m_{\rm f}$, which are 
asymptotically irrelevant. If instead we had retained them, there would 
be no cancellation when subtracting $\hat\sigma^2_{\rm f}$ from 
$\hat\sigma_{\rm c}$, leading to a much more complicated (and 
computationally burdensome) expression for $\widehat\var (\hat\theta)$. 
Combining \eqref{thetascalar} and \eqref{varfast} yields the studentized
test statistic
\begin{equation}
\label{eq:taus}
\tau_\sigma = \frac{\hat\theta}
{\sqrt{\widehat\var (\hat\theta})}.
\end{equation}
In \Cref{sec:theory}, we show that $\tau_\sigma$ is asymptotically 
distributed as $\N(0,1)$.

\begin{remark}
\label{rem:onesided}
The statistic defined in \eqref{eq:taus} yields either a one\tkk-sided
or a two\tkk-sided test. Upper-tail tests may often be of primary
interest, because we expect the diagonal elements of $\bSigma_{\rm c}$
to exceed the corresponding elements of $\bSigma_{\rm f}$ when there
is positive correlation within clusters under the alternative.
However, since this is not necessarily the case, two\tkk-sided tests
based on $\tau_\sigma^2$ may also be of interest. The asymptotic
theory in \Cref{sec:theory} handles both cases.
\end{remark}

\begin{remark}
\label{rem:Moulton}
Consider again the special case in which the null is 
heteroskedasticity with no clustering. When the elements of $\bix$
display little intra-cluster correlation, the contrast $\hat\theta$,
and hence the absolute value of $\tau_\sigma$, will tend to be small,
even if the residuals display a great deal of intra-cluster
correlation. This is what we should expect, because in that case the 
so\tkk-called Moulton factor, the ratio of clustered to non-clustered
standard errors \citep{Moulton_1986}, will be relatively small. Of
course, the opposite will be true when the elements of $\bix$ display
a lot of intra-cluster correlation. Thus, all else equal, SV tests may 
well yield different results for different choices of~$\bix$.
\end{remark}

When $k>1$, so that $\hat\btheta$ is a vector, the variance estimator 
analogous to \eqref{varfast} is
\begin{equation}
\label{var2sided}
\widehat\var(  \hat\btheta ) =
2\sum_{g=1}^G \biH_k \bigg(\sum_{h=1}^{M_g} \hat\bis_{gh}\hat\bis_{gh}^\top
\otimes\sum_{h=1}^{M_g}\hat\bis_{gh}\hat\bis_{gh}^\top\!\bigg) \biH_k^\top 
-2\sum_{g=1}^G \sum_{h=1}^{M_g}\biH_k \big(\hat\bis_{gh}\hat\bis_{gh}^\top
\otimes \hat\bis_{gh}\hat\bis_{gh}^\top\big) \biH_k^\top .
\end{equation}
Here $\biH_k$ is the so\tkk-called elimination matrix, which satisfies
$\vech (\bm{S})=\biH_k \vec (\bm{S})$ for any $k \times k$ symmetric 
matrix $\bm{S}$ \citep[p.~354]{Harville_1997}, and $\otimes$ denotes 
the Kronecker product. A studentized (Wald) statistic is then given by
\begin{equation}
\label{eq:tauGs}
\tau_\Sigma =\hat\btheta^\top 
\widehat\var(  \hat\btheta )^{-1} \hat\btheta .
\end{equation}
In \Cref{sec:theory}, we show that $\tau_\Sigma$ is asymptotically 
distributed as $\chi^2(k(k+1)/2)$.

\begin{remark}
\label{rem:Phillips}
As pointed out by a referee, \citet{CP_2018} develop simple measures
of the discrepancy between two positive definite symmetric matrices. When
bootstrapped, these measures can be used as alternative test
statistics for cases where $k>1$. Preliminary simulations suggest that
these bootstrap tests can work well, although not (in general) better
than our proposed bootstrap tests based on \eqref{eq:tauGs}. A full
analysis is beyond the scope of this paper and is therefore left for
future work.
\end{remark}

\begin{remark}
\label{rem:two-way} It is possible to use the test statistics
\eqref{eq:taus} and \eqref{eq:tauGs} for testing one\tkk-way against
two\tkk-way clustering. Suppose there are two alternative clustering
dimensions, labeled A and~B, and their intersection is labeled~I.
These could correspond to, say, state~(A) and year~(B), with the
intersection denoting observations that correspond to the same year in
the same state. Let $\hat\bSigma_j$ denote the (one\tkk-way) CRVE in
\eqref{Sighat} under clustering dimension $j \in \{
\rm{A},\rm{B},\rm{I} \}$. Then the two\tkk-way CRVE \citep*{CGM_2011}
is given by \eqref{covbeta} with
\begin{equation}
\label{eq:twowayCRVE}
\hat\bSigma_{\tkk\rm TW} = \hat\bSigma_{\rm A} + \hat\bSigma_{\rm B}
  - \hat\bSigma_{\tkk\rm I}.
\end{equation}
If we test the null of one\tkk-way clustering by~A against the
alternative of clustering by both A and~B, then $\hat\bSigma_{\rm c}=
\hat\bSigma_{\tkk\rm TW}$ and $\hat\bSigma_{\rm f}=\hat\bSigma_{\rm A}$.
Therefore, the vector of contrasts in \eqref{thetaSigma} becomes
\begin{equation}
\label{eq:twowaytheta}
\hat\btheta = \vech ( \hat\bSigma_{\tkk\rm TW} - \hat\bSigma_{\rm A} )
  = \vech ( \hat\bSigma_{\rm B} - \hat\bSigma_{\tkk\rm I} ).
\end{equation}
The result in \eqref{eq:twowaytheta} shows that testing the null of 
one\tkk-way clustering by~A against the alternative of two\tkk-way
clustering by~A and~B must lead to the same test statistic as testing
the null of one\tkk-way clustering by~I against the alternative of
one\tkk-way clustering by~B.

Although it is straightforward to derive a test statistic based on 
\eqref{eq:twowaytheta}, the asymptotic analysis of this statistic
would be different from the analysis for testing nested one\tkk-way
clustering in \Cref{sec:theory} below. For example, to derive the
asymptotic null distribution of a statistic based on
\eqref{eq:twowaytheta} for testing clustering by~I against clustering
by~B would mean analyzing it under the DGP that clustering is in fact
by~A. Therefore, we leave this analysis for future work.


Similarly, if there is two\tkk-way clustering under both the null and
alternative hypotheses, it may be feasible to calculate
score\tkk-variance statistics similar to \eqref{eq:taus} and
\eqref{eq:tauGs}. However, the analysis of the asymptotic null
distribution would require completely different, and technically
nontrivial, techniques 
\citep*{DDG_2021,MNW_multi,Menzel_2021,Chiang_2022}.
\end{remark}

In this section, we have proposed two score\tkk-variance tests of
\eqref{hypotheses}. They both involve comparing different variance
estimates of the empirical scores, namely, the two scalars in
\eqref{Sigmascalar} for the $\tau_\sigma$ test and the matrices in
\eqref{Sigmac} and \eqref{Sigmaf} for the $\tau_\Sigma$ test. The
former is a special case of the latter, and it can be obtained for the
same models as the latter by partialing out all regressors except one,
as in \Cref{sec:partial}. This special case is particularly
interesting, because the $\tau_\sigma$ test can be directional, and
also because many equations simplify neatly in the scalar case.

As we show in \Cref{sec:simulations}, the finite\tkk-sample properties
of our asymptotic tests are often good but could sometimes be better,
especially when the number of clusters under the alternative is quite
small. In such cases, we therefore recommend the use of bootstrap
tests based on the statistics \eqref{eq:taus} and~\eqref{eq:tauGs},
which often perform much better in finite samples, as we also show in
\Cref{sec:simulations}. These bootstrap implementations are described
next.

\subsection{Bootstrap Implementation}
\label{subsec:bootstrap}

The simplest way to implement a bootstrap test based on any of our
test statistics is to compute a bootstrap $P$~value, say $\hat
P^*$\tn, and reject the null hypothesis when it is less than the level
of the test. The bootstrap methods that we propose are based on either
the ordinary wild bootstrap \citep{Wu_1986, Liu_1988} or the wild
cluster bootstrap \citep{CGM_2008}. These bootstrap methods are
normally used to test hypotheses about $\bbeta$, and we are not aware
of any previous work in which they have been used to test hypotheses
about the variances of parameter estimates. The asymptotic validity of
the bootstrap tests that we now describe is established in
\Cref{subsec:boot}.

The key idea of the wild bootstrap is to obtain the bootstrap
disturbances by multiplying the residuals by realizations of an
auxiliary random variable with mean~0 and variance~1. In contrast to
many applications of the wild bootstrap, the residuals in this case
are unrestricted, meaning that they do not impose a null hypothesis on
$\bbeta$. This is because we are not testing any restrictions on
$\bbeta$ when testing the level of clustering. In the special case of 
testing the null of heteroskedasticity, as in \Cref{rem:hetnull}, we 
use the ordinary wild bootstrap. When the null involves clustering, 
we use the wild cluster bootstrap. Because the test statistics depend 
only on residuals, the value of $\bbeta$ in the bootstrap DGP does
not matter, and so we set it to zero.

The \th{b} wild (cluster) bootstrap sample is thus generated by 
$\biy^{*b} = \biu^{*b}$, where the vector of bootstrap disturbances 
$\biu^{*b}$ has typical element given by either $u_{ghi}^{*b} = 
v_{ghi}^{*b} \hat u_{ghi}$ for the wild bootstrap or $u_{ghi}^{*b} = 
v_{gh}^{*b} \hat u_{ghi}$ for the wild cluster bootstrap. The auxiliary
random variables $v_{ghi}^{*b}$ and $v_{gh}^{*b}$ are assumed to follow
the Rademacher distribution, which takes the values $+1$ and $-1$ with
equal probabilities. Notice that there is one such random variable per
observation for the wild bootstrap and one per cluster for the wild
cluster bootstrap. Other distributions can also be used; see
\citet{DF_2008}, \citet{DMN_2019}, and \citet{Webb_2022}.

The algorithm for a wild (cluster) bootstrap\tkk-based implementation
of our tests is as follows. It applies to both $\tau_\sigma$ and
$\tau_\Sigma$. For simplicity, the algorithm below simply refers to
one test statistic, $\tau$. However, it is easy to perform two or more
tests at the same time, using just one set of bootstrap samples for
all of them. For example, if there are three possible regressors of 
interest, we might perform four tests, one with $k=3$ based on 
$\tau_\Sigma$ and three with $k=1$ based on different versions 
of~$\tau_\sigma$.

\begin{algorithm}[Bootstrap test implementation]
\label{alg:BS}
Let $B >\!> 1$ denote the number of bootstrap replications, and let 
$\tau$ denote the chosen test statistic.

\begin{enumerate}

\item Estimate model \eqref{model} by OLS to obtain the 
residuals~$\hat\biu$.

\item Compute the empirical score vector $\hat\bis$ and use it to
compute~$\tau$.

\item For $b=1,\ldots,B$,

\begin{enumerate}

\item generate the vector of bootstrap dependent variables $\biy^{*b}
= \biu^{*b}$ from the residual vector $\hat\biu$ using the wild 
cluster bootstrap corresponding to the null hypothesis, or the 
ordinary wild bootstrap if the null does not involve clustering.

\item Regress $\biy^{*b}$ on $\biX$ to obtain the bootstrap residuals
$\hat\biu^{*b}$, and use these, together with $\biX$, to compute 
$\tau^{*b}$, the bootstrap analog of~$\tau$.

\end{enumerate}

\item Compute the bootstrap $P$~value $\hat{P}^\ast =B^{-1}\sum_{b=1}^B 
\IF(\tau^{*b} > \tau)$.

\end{enumerate}
\end{algorithm}

As usual, if $\alpha$ is the level of the test, then $B$ should be
chosen so that $(1-\alpha)(B+1)$ is an integer \citep{RM_2007}. Numbers
like 999 and 9,999 are commonly used because they satisfy this condition
for conventional values of~$\alpha$. Power increases in~$B$, but it 
does so very slowly once $B$ exceeds a few hundred \citep{DM_2000}.

\begin{remark}
\label{rem:BSonesided}
When $\tau$ is defined as $\tau_\sigma$, \Cref{alg:BS} yields a
one\tkk-sided upper-tail test. When $\tau$ is defined as $|\tau_\sigma|$,
$\tau_\sigma^2$, or $\tau_\Sigma$, it yields a two\tkk-sided test; see
\Cref{rem:onesided}.
\end{remark}

\begin{remark}
\label{rem:critvals}
If desired, bootstrap critical values can be calculated as quantiles
of the $\tau^{*b}$. For example, when $B=999$ and the $\tau^{*b}$ are
sorted from smallest to largest, the 0.05 critical value for a
one\tkk-sided upper-tail test is $\tau^{*b'}$ for $b' = (1-0.05)(B+1) = 
950$.
\end{remark}

\begin{remark}
\label{rem:ordinary}
We could use the ordinary wild bootstrap instead of the wild cluster 
bootstrap in \Cref{alg:BS}, even when the null hypothesis involves 
clustering. The same intuition as in \citet{DMN_2019} applies, whereby
the ordinary wild bootstrap would lead to asymptotically valid tests 
because the statistics are asymptotically pivotal. There may be cases,
like the ones considered in \citet{MW-EJ} and/or ones in which the
number of fine clusters is small, in which the wild bootstrap would
perform better than the wild cluster bootstrap. However, we believe
that such cases are likely to be rare.
\end{remark}

\subsection{Choosing the Level of Clustering by Sequential Testing}
\label{subsec:level}

In many applications, there are several possible levels of clustering.
In such situations, we suggest a sequential testing procedure. The
statistical principle upon which we base our testing procedure is the
intersection-union (IU) principle 
\citep[e.g.,][]{BergerSinclair_1984}, whereby a hypothesis is rejected
if and only if the hypothesis itself, along with any hypotheses nested
within it, are all rejected. The IU principle leads naturally to a
bottom-up testing strategy for the level of clustering.
\citet{BergerSinclair_1984} show that the IU principle does not imply
an inflation of the family-wise rejection rate in the context of
multiple testing; that is, there is no accumulation of size due to
testing multiple hypotheses. We prove a similar result for our
sequential procedure below.

Suppose the potential levels of clustering are sequentially nested, 
and denote their $\bSigma$ matrices by $\bSigma_0,
\bSigma_1,\ldots,\bSigma_p$; see \eqref{covbeta} and \eqref{Sigmas}.
Here we assume that $\bSigma_0$ corresponds to no clustering, c.f.\
\Cref{rem:hetmodel,rem:hetnull}, and that, in addition, there are $p$
potential levels of clustering of the data. All these levels of
clustering are assumed to be nested from fine to increasingly more 
coarse clustering.

In this situation, following the IU statistical principle mentioned 
above, we reject clustering at level $m$ if and only if levels 
$0,\ldots,m$ are all rejected. That is, the natural testing strategy 
here is to test clustering at level $m$ against the coarser level 
$m+1$ sequentially, for $m=0,1,\dots,p-1$, and choose the level of 
clustering in the first non-rejected test. Algorithmically, we perform 
the following sequential testing procedure.

\begin{algorithm}[Nested sequential testing procedure]
\label{alg:seq}
Let $m=0$. Then:
\begin{enumerate}
\item Test $\H{0}\!: \limN \tk \bSigma_m \bSigma_{m+1}^{-1} = \bfI$ 
against $\H{1}\!: \limN \tk \bSigma_m \bSigma_{m+1}^{-1} \neq \bfI$.
\item If the test in step~1 does not reject, choose $\hat{m}=m$ and 
   stop.
\item If $m=p-1$ and the test in step~1 rejects, choose $\hat{m}=p$ 
	and stop.
\item If $m \leq p-2$ and the test in step~1 rejects, increment 
	$m$ by $1$ and go to step~1.
\end{enumerate}
\end{algorithm}

We can equivalently state the sequential testing problem in 
\Cref{alg:seq} as a type of estimation problem. Specifically,
\begin{equation}
\label{mhat}
\hat{m} = \min \{ m \in \{ 0,1,\dots,p \} \textrm{ such that }
\H{0}\!: \plimN \tk \bSigma_m \bSigma_{m+1}^{-1} = \bfI 
\textrm{ is not rejected} \}.
\end{equation}
Of course, the $\hat{m}$ resulting from \Cref{alg:seq} and from 
\eqref{mhat} will be identical.


Because each individual test will reject a false null hypothesis with 
probability converging to one, this procedure will (at least 
asymptotically) never choose a level of clustering that is too fine. 
In other words, $\hat{m}$ defined in either \Cref{alg:seq} or 
\eqref{mhat} is (nearly) consistent. Precise asymptotic 
properties of the proposed sequential procedure are established in 
\Cref{subsec:seq}, and finite\tkk-sample performance is investigated 
by Monte Carlo simulation methods in \Cref{subsec:seqsims}.

\subsection{Other Tests for the Level of Clustering}
\label{subsec:other}

To our knowledge, only two other tests for the appropriate level of
clustering have been proposed. Like our $\tau_\sigma$ test, these
concern the standard error for a single coefficient. The best-known of
them is due to \citet*{Ibragimov_2016}, and we refer to it as the IM
test. It is a one\tkk-sided test that is derived from the procedure of
\citet*{Ibragimov_2010}. The IM test is based on the assumption that
$G$ is fixed while the number of observations tends to infinity. In
this respect, it differs from our score\tkk-variance tests, for which
the asymptotic theory in \Cref{subsec:asy} requires that $G\to\infty$.
However, see the discussion in \Cref{subsec:fixedG}.

There are two versions of the IM test. The first involves estimating
the model separately for every coarse cluster. Unfortunately, this is
impossible to do for models that involve treatment effects whenever
the treatment is invariant within clusters. Even with treatment at the
fine\tkk-cluster level, it may not be possible to estimate the model
for every coarse cluster. This is the case, for example, in the
empirical example of \Cref{sec:example}. \citet{Ibragimov_2016}
therefore also propose a two\tkk-sample version of their test
statistic that can be used for testing the level of clustering for
treatment models when entire clusters are treated or not treated.
Differences between estimates for treatment and control clusters can
be used to estimate the treatment effects and perform a test of fine
clustering.

Another test for the appropriate level of clustering for a single
coefficient has very recently been proposed by \citet{Cai_2022}. This
test is based on randomization inference. Like the IM test, and unlike
our $\tau_\sigma$ test, it is necessarily one\tkk-sided and treats $G$
as fixed. It also requires that the fine clusters be large, so that,
unlike both the IM test and our tests, it cannot be used to test the
null hypothesis of independence at the observation level.

\citet{Cai_2022} presents results from a number of simulation
experiments for his test, the IM test, and the bootstrap version of
our $\tau_\sigma$ test. They suggest that the IM test and our test are
much more similar to each other than they are to Cai's test. The IM
test always rejects more often than the bootstrap version of our
$\tau_\sigma$ test, both when the null hypothesis is true and when it
is false. It can over-reject quite severely in some cases, especially
when the ratio of fine to coarse clusters is small. In the first
version of this paper, we presented some figures comparing rejection
frequencies for the IM test and the $\tau_\sigma$ test. In the
interests of space, however, we have omitted these results, because
they are broadly similar to those from the experiments in
\citet{Cai_2022}.

At this point, what is known about the properties of our tests, the IM
test, and Cai's test suggests that none of them is to be preferred in
every case. They can all provide useful information about the
appropriate level at which to cluster. Two attractive features of our
tests, which are not shared by the other two, are that the
$\tau_\sigma$ test can be either one\tkk-sided or two\tkk-sided and
that the $\tau_\Sigma$ test is based on more than one coefficient of
interest.

\subsection{Inference about Regression Coefficients}
\label{subsec:infreg}

The ultimate purpose of using any test for the appropriate level of
clustering is to make more reliable inferences about the
coefficient(s) of interest, that is, some element(s) of $\bbeta$ in
\eqref{model}. This may or may not involve some sort of formal
pre\tkk-testing or model averaging procedure.

For simplicity, suppose we are attempting to construct a confidence
interval for $\beta_1$, the (scalar) coefficient of interest, when
there are just two levels of clustering, fine and coarse. Without a
testing procedure, an investigator must choose between fine and coarse
clustering on the basis of prior beliefs about which level is
appropriate. With a testing procedure like the ones
proposed in this paper, an investigator can instead choose the level
of clustering based on the outcome of a test. This involves choosing a
level $\alpha$ for the test and deciding whether to use a
one\tkk-sided or a two\tkk-sided test. We then form the interval based 
on coarse clustering when the test rejects, and we form the interval
based on fine clustering when it does not reject.

Of course, this procedure can never work as well as the infeasible
procedure of simply choosing the correct level of clustering. It
inevitably suffers from some of the classic problems associated with
pre\tkk-testing \citep[e.g.,][]{LP_2005}. When there is actually fine
clustering, the pre\tkk-test will sometimes make a Type~I error and
reject, leading to an interval that is usually too long. When there is
actually coarse clustering, the pre\tkk-test will sometimes make a
Type~II error and fail to reject, leading to an interval that is
usually too short. In \Cref{subsec:pretest}, we report the results of
some simulation experiments that compare confidence intervals based on
several alternative procedures.

As \citet{Ibragimov_2016} point out, it probably makes sense to report
confidence intervals for $\beta_1$ based on all clustering levels that
appear plausible. The resulting inferences are then explicitly
conditional on the level of clustering. Because our tests, like the
other ones discussed in \Cref{subsec:other}, provide evidence on the
plausibility of each level of clustering, they can reduce the number
of intervals that need to be reported. For example, if the hypothesis
of independence is strongly rejected against one or more clustering
structures, then it would not be necessary to report a confidence
interval based on a heteroskedasticity-robust standard error. But if
the $P$~value for the hypothesis of fine clustering against coarse
clustering is neither extremely small nor very large, then it might
well seem reasonable to report confidence intervals based on both
levels. Whatever intervals an investigator chooses to report, tests
for the appropriate clustering level can provide valuable information
about which ones are empirically more plausible. These tests may thus
be thought of as robustness checks \citep{Cai_2022}.

\section{Asymptotic Theory}
\label{sec:theory}

In \Cref{subsec:asy}, we derive the asymptotic distributions of the
two score\tkk-variance test statistics under the null hypothesis and 
show that they are divergent under the alternative. Then we prove the
validity of the bootstrap implementation (\Cref{subsec:boot}) and
prove asymptotic results for the sequential testing procedure
(\Cref{subsec:seq}). We first state and discuss the assumptions needed
for our proofs, which may be found in \Cref{sec:proofs}.

\begin{assumption}
\label{as:cluster}
The sequence $\bis_{gh}=\sum_{i=1}^{N_{gh}}\biX_{ghi}^\top u_{ghi}$ is 
independent across both $g$ and~$h$.
\end{assumption}

\begin{assumption}
\label{as:moments}
For all $g,h$, it holds that $\E (\bis_{gh}) = \bzero$ and $\var 
(\bis_{gh}) =\bSigma_{gh}$. Furthermore, $\sup_{g,h,i}\E \Vert\bis_{ghi}
\Vert^{2\lambda}<\infty$ for some $\lambda >1$.
\end{assumption}

\begin{assumption}
\label{as:X}
The regressor matrix $\biX$ satisfies $\sup_{g,h,i} \E \Vert \biX_{ghi} 
\Vert^2 <\infty$ and $N^{-1}\biX^\top\biX \Pto \bXi$, where $\bXi$ is 
finite and positive definite.
\end{assumption}

\begin{assumption}
\label{as:eigen}
Let $\omega_{\min}(\cdot)$ and $\omega_{\max}(\cdot)$ denote the 
minimum and maximum eigenvalues of the argument. Then 
$\inf_{g,h}N_{gh}^{-1}\omega_{\min} (\bSigma_{gh} )>0$ and 
$\sup_{g,h}\omega_{\max} \big(\bSigma_{gh} (\sum_{h=1}^{M_g}
\bSigma_{gh})^{-1} \big)<1$.
\end{assumption}

\begin{assumption}
\label{as:size}
For $\lambda$ defined in \Cref{as:moments}, the cluster sizes satisfy
\par \vspace{\abovedisplayskip}
\hfill $\displaystyle \frac{\sup_g N_g^2 \sup_{g,h} N_{gh}^2}%
{\sum_{g=1}^G \omega_{\min}\big(\sum_{h=1}^{M_g}\bSigma_{gh}\big)^2} 
\longto 0 
\quad \textrm{and} \quad 
\frac{N^{1/\lambda}\sup_g N_g \sup_{g,h}N_{gh}^{3-1/\lambda}}%
{\sum_{g=1}^G \omega_{\min}\big(\sum_{h=1}^{M_g}\bSigma_{gh}\big)^2} 
\longto 0.$
\end{assumption}

\Cref{as:cluster} is the assumption of (at most) ``fine'' clustering,
which implies that the null hypothesis in \eqref{hypotheses} is
satisfied, even without taking the limit. In fact, it is slightly 
weaker than that, because we do not make the stronger assumption that 
all observations in any fine cluster are independent of those in a
different fine cluster; we only assume that the cluster sums are
independent across fine clusters. The moment conditions in 
\Cref{as:moments} and the multicollinearity condition in \Cref{as:X} 
are standard in linear regression models.

Next, the conditions in \Cref{as:eigen} rule out degenerate cases. The 
minimum eigenvalue condition rules out perfect negative correlation 
between scores within fine clusters. The maximum eigenvalue condition 
ensures that the variance of a single fine cluster cannot dominate the 
sum of the variances within a coarse cluster. It is basically satisfied if 
$M_g > 1$ for all~$g$. The latter holds by construction of the test 
statistics, because any coarse cluster with $M_g=1$ will not contribute 
to $\hat\btheta$, and hence not to the test statistic.

The conditions in \Cref{as:size} restrict the amount of heterogeneity
of cluster sizes that is allowed under both the null and the
alternative. Neither the fine cluster sizes nor the coarse cluster
sizes are required to be bounded under these conditions, which allow
the cluster sizes to diverge with the sample size. The first condition
is used in the proofs to replace residuals with disturbances and for 
convergence of the variance. The second condition trades off moments
and cluster size heterogeneity to rule out the possibility that one
cluster dominates the test statistic in the limit in such a way that
the central limit theorem does not apply; technically, it is used to
verify Lyapunov's condition for the central limit theorem. When
$\lambda \to \infty$, the second condition is implied by the first.

The denominators of both terms in \Cref{as:size} show that these 
conditions trade off intra-cluster dependence and cluster-size 
heterogeneity. That is, the greater the amount of intra-cluster 
correlation, the larger are the denominators in \Cref{as:size}, which 
allows larger clusters without dominating the limit; a similar
tradeoff was found in \citet{DMN_2019}. Furthermore, more homogeneity
in cluster sizes allows for fewer and larger clusters. We illustrate
these tradeoffs in the following remarks.

\begin{remark}
\label{rem:tradeoff}
Under \Cref{as:cluster,as:eigen}, both denominators in \Cref{as:size} 
are bounded from below by $\sum_{g=1}^G N_g^2 \geq c N \inf_g N_g$, and 
a sufficient condition for \Cref{as:size} is
\begin{equation}
\label{suff1a}
\sup_{g,h} N_{gh}^2 \left( \frac{\sup_g N_g}{\inf_g N_g} \right) 
\left( \frac{\sup_g N_g}{N} \right) \longto 0
\quad\textrm{and}\quad
\left( \frac{\sup_g N_g}{\inf_g N_g} \right)^{\!\lambda}
\left( \frac{\sup_{g,h}N_{gh}^{3\lambda-1}}{N^{\lambda-1}} \right)
\longto 0.
\end{equation}
If the cluster sizes are bounded under the alternative, i.e.\ $\sup_g
N_g < \infty$, then \eqref{suff1a} is easily satisfied. Note that
$\sup_g N_g /N \to 0$, and hence $G\to\infty$, is implied by
\Cref{as:size}, and it is therefore not stated explicitly. Suppose, 
on the other hand, that \Cref{as:eigen} were strengthened to assume that 
$\inf_{g,h} N_{gh}^{-2} \omega_{\min}(\bSigma_{gh})>0$, as would be the 
case if a random-effects or factor-type model were assumed under the 
null. In that case, \Cref{as:size} and \eqref{suff1a} could be weakened 
substantially.
\end{remark}

\begin{remark}
\label{rem:sizes}
It is interesting to consider a setup for clusters that are relatively 
homogeneous, but possibly unbounded, in size. To make this concrete, 
suppose the coarse clusters have size $N_g =O( N^\alpha )$ for 
$g=1,\ldots,G$, where $\alpha \in [0,1)$ and `$O (\cdot )$' is to be 
understood as an exact rate subject to $N_g$ being an integer. Because 
$\sum_{g=1}^G N_g = N$, it then holds that $G = O ( N^{1-\alpha} )$. 
Similarly, for each $g$, the fine clusters have size $N_{gh} = O ( 
N_g^\gamma )$ for $h=1,\ldots,M_g =O ( N_g^{1-\gamma})$. That is, when 
$\alpha$ is large (small), there are few large (many small) coarse 
clusters. Similarly, when $\gamma$ is large (small), there are few large 
(many small) fine clusters per coarse cluster. Under this setup, 
\eqref{suff1a} is satisfied if $\alpha (2\gamma+1)<1$ and $\alpha\gamma 
<(\lambda-1)/(3\lambda-1)$.

The important implication of this setup is that the implied
restrictions on the cluster sizes in \Cref{as:size} are very weak. In
fact, if we assume that the fine cluster sizes are bounded (i.e.,
$\gamma = 0$), which applies, for example, in the important special
case in which the scores are independent but heteroskedastic under the
null, then we can allow $G=O(N^{1-\alpha})$ for any $\alpha <1$. That
is, the number of coarse clusters can be arbitrarily close to $O(1)$.
For example, we allow $G=O(N^{0.1})$ and $N_g = O(N^{0.9})$, which
corresponds to very few and very large coarse clusters. In this sense,
our asymptotic framework can nearly accommodate the fixed-$G$ setup;
see \Cref{subsec:fixedG}.
\end{remark}


\subsection{Theory for Asymptotic Tests}
\label{subsec:asy}

\begin{theorem}
\label{thm:asy}
Let \Cref{as:cluster,as:moments,as:eigen,as:X,as:size} be satisfied. 
Then, as $N \to \infty$, it holds that
\begin{align*}
\var (\btheta )^{-1/2} \hat\btheta &\dto \N (0,\bfI), &
\var (\btheta )^{-1}\widehat\var (\hat\btheta ) &\Pto \bfI,
\quad\textrm{and}\\
\frac{\hat\theta}{\sqrt{\var(\theta)}} &\dto \N (0,1), &
\frac{\widehat\var (\hat\theta)}{\var (\theta)} &\Pto 1.
\end{align*}
\end{theorem}

\begin{remark}
\label{rem:self-norm}
Observe that the statement of the asymptotic distributions in 
\Cref{thm:asy} only concerns quantities that are self-normalized. For 
example, in the scalar case, these are either $\hat\theta$ divided 
by its true standard error or the estimated variance of $\hat\theta$ 
divided by the true variance. This is because the appropriate rates of 
convergence are not known in general; see the discussion 
below~\eqref{hypotheses}.
\end{remark}

The asymptotic distributions of the test statistics follow
immediately from \Cref{thm:asy}.

\begin{corollary}
\label{cor:asy}
Let \Cref{as:cluster,as:moments,as:eigen,as:X,as:size} be satisfied. 
Then, as $N \to \infty$, it holds that
\begin{equation*}
\tau_\Sigma \dto \chi^2 \big(k(k+1)/2\big)
\quad\textrm{and}\quad
\tau_\sigma \dto \N(0,1).
\end{equation*}
\end{corollary}

We next consider the asymptotic behavior of the test statistics under
the alternative. Because \Cref{as:cluster} implies that \H{0} is true,
we do not make that assumption. Instead, we impose the following
conditions:

\begin{assumption}
\label{as:clusteralt}
The sequence $\bis_g = \biX_g^\top \biu_g = \sum_{h=1}^{N_g}\bis_{gh}$ 
is independent across~$g$.
\end{assumption}

\goodbreak

\begin{assumption}
\label{as:sizealt}
The cluster sizes satisfy
\par \vspace{\abovedisplayskip}
\hfill $\displaystyle \frac{\sup_g N_g^{3/2}N^{1/2}}%
{\sum_{g=1}^G \omega_{\min}(\bSigma_g)} \longto 0.$
\end{assumption}

\goodbreak

\Cref{as:clusteralt} is the assumption of (at most) coarse clustering. 
This assumption is very general, and departures from the null could be 
very small and inconsequential. In order for our tests to be 
able to detect departures from the null hypothesis, with probability 
converging to one in the limit, we need to impose sufficient 
correlation within the coarse clusters. That is, we need $\bSigma_g = 
\sum_{h_1=1}^{M_g}\sum_{h_2=1}^{M_g}\E (\bis_{gh_1}\bis_{gh_2}^\top)$
to be sufficiently large, in aggregate. This condition is embodied in 
\Cref{as:sizealt}.

\begin{remark}
\label{rem:tradeoffalt}
As in \Cref{rem:tradeoff}, there is a tradeoff between cluster size 
heterogeneity and intra-cluster correlation, in this case correlation 
within coarse clusters. Specifically, under \Cref{as:eigen}, the 
denominator in \Cref{as:sizealt} is bounded from below by $\sum_{g=1}^G 
N_g = N$\tn, and hence a sufficient condition for \Cref{as:sizealt} is 
\begin{equation}
\label{sizealt}
\frac{\sup_g N_g^3}{N}\longto 0.
\end{equation}
Suppose instead that \Cref{as:eigen} were strengthened to assume that 
$\inf_g N_g^{-2}\omega_{\min}(\bSigma_g)>0$ (as in
\Cref{rem:tradeoff}, this could be due to a random-effects model or a
factor-type model). That is, more correlation is assumed within the
coarse clusters, so that there is a stronger departure from the null
hypothesis. In this case, the denominator in \Cref{as:sizealt} is
bounded from below by $\sum_{g=1}^G N_g^2 \geq \inf_g N_g N$.
Therefore, a sufficient condition for \Cref{as:sizealt} is
\begin{equation}
\label{sizealt2}
\frac{\sup_g N_g^3}{\inf_g N_g^2 N}\longto 0.
\end{equation}
With relatively homogeneous coarse clusters as in \Cref{rem:sizes}, 
i.e.\ coarse clusters where $\sup_g N_g$ and $\inf_g N_g$ are of the 
same order of magnitude, the condition \eqref{sizealt2} reduces to 
$\sup_g N_g /N \to 0$, which is clearly minimal and implied by 
\Cref{as:size}.
\end{remark}

\begin{theorem}
\label{thm:cons}
Let \Cref{as:moments,as:X,as:eigen,as:size,as:clusteralt,as:sizealt} be 
satisfied, and suppose \H{0} in \eqref{hypotheses} is not true. Then, 
as $N \to \infty$, it holds that
\begin{equation*}
\tau_\Sigma \Pto +\infty \quad\textrm{and}\quad |\tau_\sigma| \Pto +\infty.
\end{equation*}
\end{theorem}

It follows immediately from \Cref{thm:cons} that tests based on either 
of our statistics reject with probability converging to one under the
alternative. That is, they are consistent tests.

Of course, power will depend in a complicated way on many aspects of 
the model and DGP, including the number of large clusters and their
sizes, because these will affect the number of correlations that need
to be estimated between fine clusters within coarse clusters; see 
\eqref{thetaknown}. Power will also depend on the true values of these
correlations. If they are mostly non-zero and non-trivial, then power
will be higher with larger coarse clusters.





\subsection{Theory for Bootstrap Tests}
\label{subsec:boot}

We now demonstrate the asymptotic validity of the bootstrap 
implementation of our tests. To this end, let $\tau$ denote either of
our statistics, and let the cumulative distribution function of $\tau$
under \H{0} be denoted $P_0 (\tau \leq x)$. The corresponding
bootstrap statistic is denoted $\tau^\ast$\tn. As usual, let $P^\ast$
denote the bootstrap probability measure, conditional on a given
sample, and let $\E^\ast$ denote the corresponding expectation
conditional on a given sample. 

\begin{theorem}
\label{thm:boot}
Let \Cref{as:moments,as:X,as:eigen,as:size,as:clusteralt} be satisfied 
with $\lambda \geq 2$, and assume that $\E^\ast |v^\ast|^{2\lambda}
<\infty$. Then, as $N \to \infty$, it holds for any $\epsilon >0$ that
\begin{equation*}
P \big( \sup_{x \in \mathbb R} \big| P^\ast (\tau^\ast \leq x) - 
P_0 (\tau \leq x) \big| > \epsilon \big) \longto 0 .
\end{equation*}
\end{theorem}

First, note that the bootstrap theory requires a slight strengthening of
the moment condition since at least four moments are now required. 
Second, \Cref{thm:boot} shows that the bootstrap $P$~values in 
\Cref{alg:BS} are asymptotically valid under \Cref{as:cluster} and \H{0}.
Third, note that neither the null hypothesis nor \Cref{as:cluster} is 
imposed in \Cref{thm:boot}. Thus \Cref{thm:asy,thm:cons,thm:boot} 
together show immediately that the bootstrap tests are consistent. We 
summarize these results in the following corollary.

\begin{corollary}
\label{cor:boot}
Let \Cref{as:moments,as:X,as:eigen,as:size} be satisfied with $\lambda 
\geq 2$, and assume that $\E^\ast |v^\ast|^{2\lambda} < \infty$. As 
$N \to \infty$, it holds that:
\begin{itemize}
\item[(i)] If \Cref{as:cluster} is satisfied and \H{0} is true, then
$\hat P^\ast \dto \U(0,1)$, where $\U(0,1)$ is a uniform random 
variable on~$[0,1]$.
\item[(ii)] If \Cref{as:clusteralt,as:sizealt} are satisfied and \H{0} 
is not true, then $\hat P^\ast \Pto 0$.
\end{itemize}
\end{corollary}

\subsection{Theory for Sequential Testing Procedure}
\label{subsec:seq}

The next theorem provides theoretical justification for the sequential
testing procedure given in \Cref{alg:seq}.

\begin{theorem}
\label{thm:seq}
Let $\hat m$ be defined in \Cref{alg:seq} or \eqref{mhat}. Suppose 
\Cref{as:cluster} is satisfied when the ``fine'' clustering level in 
\eqref{hypotheses} is $m = m_0 \in \{ 0,1,\ldots ,p \}$ (and hence also 
when $m>m_0$), and suppose \H{0} in \eqref{hypotheses} is not true for 
clustering levels $m<m_0$. Suppose also that 
\Cref{as:moments,as:X,as:eigen,as:size,as:sizealt} are satisfied, and 
let $\alpha$ denote the nominal level of the tests. As $N \to \infty$, it holds 
that
\begin{itemize}
\item[(i)] if $m_0 \leq p-1$, then $P(\hat m \leq m_0-1) \to 0$, 
$P(\hat m =m_0 ) \to 1-\alpha$, and $P(\hat m \geq m_0+1) \to \alpha$;
\item[(ii)] if $m_0 = p$, then $P ( \hat m \leq m_0-1 ) \to 0$ and
$P(\hat m = m_0 ) \to 1$.
\end{itemize}
\end{theorem}

The results in \Cref{thm:seq} show that $\hat m$ defined in 
\Cref{alg:seq} or \eqref{mhat} is asymptotically correct with
probability converging to $1-\alpha$ when $m_0 \leq p-1$ and with
probability converging to~1 when $m_0 =p$.  It is worth emphasizing
that the sequential procedure will never ``under-estimate'' the
clustering level, at least asymptotically, because $\hat m < m_0$ with
probability converging to~0.

\subsection{Fixed-$G$ Asymptotic Theory}
\label{subsec:fixedG}



In the literature on cluster-robust inference, a few authors have 
considered an alternative asymptotic framework, referred to as
fixed-$G$ asymptotics, in which the number of clusters is fixed as
$N\to\infty$ while cluster sizes diverge; key early papers are
\citet{Ibragimov_2010} and \citet*{BCH_2011}. However, fixed-$G$ 
asymptotics are proven under the very strong assumption that a central 
limit theorem applies to the normalized scores for each cluster. This 
assumption seriously limits the amount of intra-cluster dependence. For 
example, it rules out common models such as many types of random-effects
and factor models. See \citet{MNW-guide} for a detailed discussion.


Nonetheless, we now briefly consider an asymptotic framework in which 
the number of coarse clusters, $G$, is fixed, but there are many fine 
clusters within each coarse cluster, i.e.\ $M_g \to\infty$ for all~$g$.
For simplicity, we consider the scalar case with $k=1$. Let $\sigma_g^2 =
\Var (\sum_{h=1}^{M_g}s_{gh})=\sum_{h=1}^{M_g}\sigma_{gh}^2$ (under
the null) and define the weights $w_g^2 = \lim_{M_g\to\infty} 
\sigma_g^2 / \Var (\theta)^{1/2}$, where $\Var (\theta )$ is given in 
\eqref{tauvar}. Then suppose, for all $g$, that 
(i)~$\sigma_g^{-1}\sum_{h=1}^{M_g}s_{gh}\dto\N (0,1)$,
(ii)~$\sigma_g^{-1}\sum_{h=1}^{M_g}s_{gh}^2\Pto 1$, and
(iii)~$w_g^2  \in [0,\infty )$.
The high-level condition~(i) is typical of the fixed-$G$ literature and 
imposes very strong limitations on the amount of intra-cluster 
dependence that is allowed. Condition~(ii) is a homogeneity assumption, 
and condition~(iii) ensures that one cluster does not dominate the sum 
in the limit. Under the null hypothesis and these conditions, it can be 
proven that
\begin{equation}
\label{chi square}
\frac{\theta}{\sqrt{\Var (\theta)}}\dto
\sum_{g=1}^G w_g^2 (\chi_{1,g}^2 -1),
\end{equation}
where $\chi_{1,g}^2$ for $g=1,\ldots ,G$ denote independent $\chi_1^2$ 
random variables. Under suitable additional regularity conditions, we 
conjecture that $\tau_\sigma = \hat\theta / \sqrt{\widehat\Var (\hat\theta )}$ 
has the same asymptotic distribution as in \eqref{chi square}.

The limiting distribution in \eqref{chi square} is a weighted sum of 
independent $\chi_1^2$ random variables. Because the weights $w_g^2$
depend on unknown parameters, the distribution is non-pivotal and
hence cannot be used for inference. Under the extreme homogeneity
condition that the $w_g^2$ are the same for all $g$, the distribution
simplifies to $(\chi_G^2-G)/\sqrt{2G}$, which is a centered and
normalized $\chi_G^2$. This distribution is pivotal and could be used
for inference, although the conditions under which it is derived are 
extraordinarily strong.

Continuing with this type of fixed-$G$ asymptotic argument, we could 
instead assume that $N_{gh}\to\infty$ for all $g,h$, while the $M_g$ 
and $G$ are fixed. That is, the number of observations within each 
fine cluster diverges, but there are only a fixed number of fine and 
coarse clusters. This setup is quite similar to the previous one. We 
conjecture that the asymptotic distribution would again be a weighted 
sum of $\chi^2_1$ random variables similar to the one in \eqref{chi
square}, but the summation would extend over $\sum_{g=1}^G M_g =
G_{\rm f}$ elements.

In either case, if the weights $w_g^2$ are not too heterogeneous, the 
fixed-$G$ limiting distributions could be well approximated by a 
standard normal distribution, at least when the number of clusters is 
not very small. In the setup with $M_g \to\infty$, this would be the 
number of coarse clusters,~$G$. In the setup with $N_{gh}\to\infty$,
it would be the number of fine clusters,~$G_{\rm f}$. Thus, in the
end, the normal limit theory obtained under large\tkk-$G$ asymptotics
in \Cref{thm:asy} and \Cref{cor:asy} may also provide a good
approximation under fixed-$G$ asymptotics. A full analysis of
fixed-$G$ asymptotic theory for our model and test statistics would be
interesting, but it is beyond the scope of this paper and is
consequently left for future work.

\section{Dimension Reduction by Partialing Out}
\label{sec:partial}

As discussed in \Cref{rem:partial}, it is commonly the case in
empirical work that the number of regressors is very large and that
most of the regressors are not of primary interest. Comparing
large-dimensional CRVE matrices by the methods in \Cref{sec:tests} is
impractical. Fortunately, it is easy to solve this problem by
partialing out the regressors that are not of primary interest prior
to performing our tests.

Suppose the full set of regressors is partitioned as $\biX =[\biX_1,\; 
\biX_2]$, where $\biX_1$ denotes the $N\times k_1$ matrix of the 
regressors of interest and $\biX_2$ denotes the $N \times k_2$ matrix 
of other regressors, with $k = k_1 + k_2$. Similarly, partition 
$\bbeta^\top =[\bbeta_1^\top\tn, \; \bbeta_2^\top ]$, where the 
coefficients corresponding to the regressors of interest are in the 
$k_1 \times 1$ parameter vector $\bbeta_1$ and the rest are collected in
$\bbeta_2$. If the coefficient vector of interest is actually a linear 
combination of the elements of $\bbeta_1$ and $\bbeta_2$, we can 
redefine $\biX$ as a nonsingular affine transformation of the original 
$\biX$ matrix, so that $\bbeta_1$ has the desired interpretation.

We regress each column of $\biX_1$ on $\biX_2$ and define $\biZ$ as 
the matrix of residuals from those $k_1$ regressions. The model 
\eqref{model} can then be rewritten as
\begin{equation}
\label{newmodel}
\biy =  \biZ\bbeta_1 + \biX_2\bdelta + \biu, \qquad
\biZ = \biM_{\biX_2}\biX_1,
\end{equation}
where $\biM_{\biX_2}=\bfI_N -\biX_2
(\biX_2^\top\biX_2)^{-1}\biX_2^\top$ is the orthogonal projection
matrix that projects off (or partials out) $\biX_2$. The regressor
matrices $\biZ$ and $\biX_2$ are orthogonal, and the models
\eqref{model} and \eqref{newmodel} have exactly the same explanatory
power and the same disturbances, $\biu$. The coefficient $\bbeta_1$ in
\eqref{newmodel} is identical to the one defined in the previous
paragraph, but the coefficient $\bdelta$ is different 
from~$\bbeta_2$.

Using the orthogonality between $\biZ$ and $\biX_2$, the 
OLS estimate of $\bbeta_1$ is, c.f.\ \eqref{betahat},
\begin{equation}
\label{beta1hat}
\hat\bbeta_1 = (\biZ^\top\biZ)^{-1} \biZ^\top\biy
=  \bbeta_{1,0} +  (\biZ^\top\biZ)^{-1}\sum_{g=1}^G \biZ_g^\top\biu_g,
\end{equation}
where $\bbeta_{1,0}$ is the true value of $\bbeta_1$.
The relationship between $\biZ$ and $\biX$ can be written as
\begin{equation}
\label{Z and Q}
\biZ = \biX \biQ \quad\textrm{with}\quad
\biQ =[ \bfI_{k_1} , \; -\biX_1^\top\biX_2 (\biX_2^\top\biX_2 )^{-1}]^\top .
\end{equation}
Therefore, the score for $\bbeta_1$ is $\biZ_g^\top \biu_g = \biQ^\top 
\biX_g^\top \biu_g = \biQ^\top \bis_g$. Thus, from \eqref{beta1hat} and 
\eqref{Z and Q}, we obtain the following sandwich formula, c.f.\ 
\eqref{covbeta},
\begin{equation}
\label{covbeta1}
\widehat\var (\hat\bbeta_1 ) = 
(\biZ^\top\biZ)^{-1} \biQ^\top \hat\bSigma\tk \biQ (\biZ^\top\biZ)^{-1}\tn .
\end{equation}
Under \Cref{as:X}, $\biQ \Pto [ \bfI_{k_1} , \; -\bXi_{12}\bXi_{22}^{-1}
]^\top = \biA$, say, so that the middle matrix in \eqref{covbeta1} is 
clearly an estimator of $\biA^\top \bSigma \biA$.

The matrix $\biQ$ in \eqref{Z and Q} and its limit $\biA$ can be
viewed as mechanisms for dimension reduction. They transform the
problem from one involving the $k\times k$ matrix $\bSigma$ to one
involving the $k_1\times k_1$ matrix $\biA^\top \bSigma \biA$. The
latter is the variance of $\biA^\top \biX^\top \biu$, and it depends
on the clustering structure in the same way as $\bSigma$. For the
model \eqref{newmodel}, we consequently replace the hypotheses in
\eqref{hypotheses} with 
\begin{equation}
\label{newhypotheses}
\H{0}\!: \limN \tk (\biA^\top \bSigma_{\rm f}\biA ) \tk (\biA^\top 
\bSigma_{\rm c} \biA )^{-1} = \bfI 
\quad \textrm{and} \quad
\H{1}\!: \limN \tk ( \biA^\top \bSigma_{\rm f}\biA ) \tk (\biA^\top 
\bSigma_{\rm c} \biA )^{-1} \neq \bfI.
\end{equation}
Furthermore, from \eqref{Z and Q}, we see that we can use the same 
algebra for the model in \eqref{newmodel} as for the model in 
\eqref{model} to define the test statistics, i.e.\ \eqref{Sigmac}, 
\eqref{Sigmaf}, and so on, but now with empirical scores $\biQ^\top 
\hat\bis_g$ and $\biQ^\top \hat\bis_{gh}$ instead of $\hat\bis_g$ and 
$\hat\bis_{gh}$, respectively. This also applies to the bootstrap 
implementation in \Cref{alg:BS}. Of course, degrees-of-freedom 
corrections like the factor $m_c$ in \eqref{Sighat} need to reflect the 
total number of estimated coefficients.

Since very few regression models in economics contain just one
regressor, the $\tau_\sigma$ test will almost always involve
partialing out. It seems likely that the $\tau_\Sigma$ test will also
involve partialing out in the vast majority of cases, so that the
dimension of the vector $\hat\btheta$ upon which the $\tau_\Sigma$ 
test is based will be $k_1$ rather than~$k$.

\begin{remark}
\label{rem:empscores}
The empirical scores $\biQ^\top\hat s_{gh}$ and
$\biQ^\top\hat\bis_{gh}$ depend on the matrix $\biZ =
\biM_{\biX_2}\biX_1$, which is the residual matrix from regressing
$\biX_1$ on $\biX_2$. Therefore, different choices for $\biX_1$ will
yield different empirical scores, and hence different test statistics;
see \Cref{rem:Moulton}. This is also reflected in the hypotheses in
\eqref{newhypotheses}, where different choices for $\biX_1$ will yield
a different $\biA$ matrix and hence different null and alternative
hypotheses.
\end{remark}

\begin{remark}
\Cref{thm:asy,thm:cons,thm:boot,thm:seq} continue to hold with the new
definitions given in this section, with $\bbeta_1$ replacing $\bbeta$,
$\biZ$ replacing $\biX$, and $k_1$ replacing $k$. Because the matrix
$\biQ \Pto \biA$ under \Cref{as:X}, it acts only as a fixed constant
in all asymptotic arguments; that is, $\biQ^\top\bis_{gh} = \biA^\top
\bis_{gh}(1+o_P(1))$. Thus, the same proofs apply with $\bis_{gh}$ 
replaced by $\biA^\top\bis_{gh}$.
\end{remark}

\begin{remark}
\label{rem:sufficient}
Careful inspection of the proofs shows that, in the setup of this section, 
we can replace $\bSigma_{gh}$ with $\biA^\top\bSigma_{gh}\biA$ in 
\Cref{as:size}. This could be attractive in some cases. Suppose, for 
example, that $\biX_1$ and $\biX_2$ are (asymptotically) orthogonal, 
such that $\biA^\top\bSigma\biA$ is equal to the diagonal block of 
$\bSigma$ corresponding to $\biX_1^\top\biu$. Suppose also that $\biX_1$
and $\biu$ are both finely clustered, but the $\biX_2$ are independent. 
Then $\biA^\top\bSigma_{gh}\biA$ satisfies the condition in 
\Cref{rem:tradeoff}, while $\bSigma_{gh}$ only satisfies the 
corresponding condition in \Cref{as:eigen}, and hence using 
$\biA^\top\bSigma_{gh}\biA$ in \Cref{as:size} would lead to a weaker 
condition.
\end{remark}

\section{Simulation Experiments}
\label{sec:simulations}

Most of the papers cited in the second paragraph of \Cref{sec:intro}
employ simulation experiments to study the finite\tkk-sample
properties of methods for cluster-robust inference. To our know\-ledge,
all of these papers use some sort of random-effects, or
single\tkk-factor, model to generate the data. The key feature of
these models is that all of the intra-cluster correlation for every
cluster $g$ arises from a single random variable, say $\xi_g$, which
affects every observation within that cluster \emph{equally}. This
yields disturbances that are equi-correlated within each cluster.

Although this type of DGP is convenient to work with and can readily
generate any desired level of intra-cluster correlation, it cannot be
used when a regression model has cluster fixed effects. Because the
fixed effects completely explain the $\xi_g$, the residuals are always
uncorrelated. Thus, for models with cluster fixed effects, it is
always valid to use heteroskedasticity-robust (HR) standard errors
whenever the intra-cluster correlation of the disturbances arises
solely from a random-effects model. In such cases, the null hypothesis
of our tests is satisfied, and they will have no (asymptotic) power.
Of course, this is the desired outcome both in the statistical sense,
because the null is satisfied, and in the practical sense, because
cluster-robust (CR) standard errors are not needed.

In practice, HR and CR standard errors often differ greatly in models 
with cluster fixed effects; see, for example, \citet{BDM_2004}, 
\citet{JGM-CJE}, and \Cref{sec:example}. Therefore, whatever processes
are generating intra-cluster correlation in real-world data must be
more complicated than simple random-effects models. Since we wish to
investigate models with cluster fixed effects, we need to employ a DGP
for which cluster fixed effects do not remove all of the 
intra-cluster correlation. To this end, we generate both the
regressors and the disturbances in our experiments using factor models
of the form
\begin{equation}
\label{facDGP}
\begin{aligned}
z_{gi} &= \rho^{1/2}\tk\xi^1_g + (1-\rho)^{1/2}\tkk\zeta_{gi}
  \;\;\mbox{ if $i$ is odd}\cr
z_{gi} &= \rho^{1/2}\tk\xi^2_g + (1-\rho)^{1/2}\tkk\zeta_{gi}
  \;\;\mbox{ if $i$ is even.}\cr
\end{aligned}
\end{equation}
Here $\xi^1_g$ and $\xi^2_g$ are random effects, distributed as
standard normal, which apply respectively to the odd-numbered and
even-numbered observations within the \th{g} cluster. The $\zeta_{gi}$
are also distributed as standard normal. Under the DGP \eqref{facDGP},
the $z_{gi}$ have variance one, and the intra-cluster correlation of
the odd (or even) observations is $\rho \ge 0$.

The DGP \eqref{facDGP} can be interpreted in a variety of ways,
depending on the nature of the data. The idea is that there are two
types of observations within each cluster, and all the intra-cluster
correlation is within each type. For example, with clustering at the
geographical level, there might be two sub\tkk-regions. With
clustering at the industry level, there might be two types of firm.
The key assumption is that the researcher knows which cluster an
observation belongs to, but not which type. Including cluster fixed
effects explains some of the intra-cluster correlation by estimating
an average of $\xi^1_g$ and $\xi^2_g$ for each cluster, but it does
not explain all of it. Thus cluster-robust inference is still needed,
and our tests should still have power.

In practice, of course, there might be more than than two types within
each cluster, and the numbers of observations in each would almost
certainly not be the same. It would be easy to make the DGP
\eqref{facDGP} more complicated. However, our objective is not to
mimic any actual dataset, but simply to generate data in a way that
allows cluster fixed effects to be combined with cluster-robust
standard errors.

The DGP \eqref{facDGP} makes no reference to fine and coarse clusters.
It could be used to generate either finely or coarsely clustered data.
The regressors $\biX_1$ (that is, the ones whose coefficients are of
interest; see \Cref{sec:partial}) are generated using \eqref{facDGP},
and they are always coarsely clustered. This ensures that, if the
disturbances are either independent ($\rho=0$), finely clustered, or
coarsely clustered, the scores are also independent, finely clustered,
or coarsely clustered, respectively.

In all experiments, each of the regressors in $\biX_1$ is generated
independently. This implies that there is no correlation among the
coefficient estimates. It might seem that the extent of any such
correlation would be important for the properties of the $\tau_\Sigma$
tests. However, that is not the case. We find numerically that the
$\tau_\Sigma$ statistic is invariant to any transformation of $\biX_1$
that does not change the subspace spanned by its columns. Thus there
is no loss of generality in generating the columns of $\biX_1$
independently.

\subsection{Performance under the Null Hypothesis}
\label{subsec:null}

Our first set of experiments is designed to investigate the rejection
frequencies of asymptotic and bootstrap score\tkk-variance tests under
the null hypothesis. The model is
\begin{equation}
\label{simmod}
y_{ghi} = \sum_{\ell=1}^{k_1} \beta_\ell X^\ell_{ghi} + \biX^2_{gh}\bdelta 
+ u_{ghi},
\end{equation}
where the regressors $X^\ell_{ghi}$ are generated independently across
$\ell$ by \eqref{facDGP} at the coarse level with $\rho=0.5$. The
additional regressors in $\biX^2_{gh}$ are either a constant term or a
set of cluster fixed effects. When testing fine against coarse
clustering, the fixed effects are at the fine level, and the
disturbances are finely clustered with $\rho=0.1$. When testing
independence against (coarse) clustering, the fixed effects are at the
coarse level, and the disturbances are independent. The number of
coarse clusters, which in this section we denote by $G_{\rm c}$, is
allowed to vary. In the first set of experiments, there are always
four fine clusters in each coarse cluster, so that $G_{\rm f} =
4G_{\rm c}$.

\begin{figure}[tb]
\begin{center}
\caption{Rejection frequencies for $\tau_\Sigma$ tests at 
0.05 level, $G_{\rm c}$ varying}
\label{fig:1}
\includegraphics[width=\textwidth]{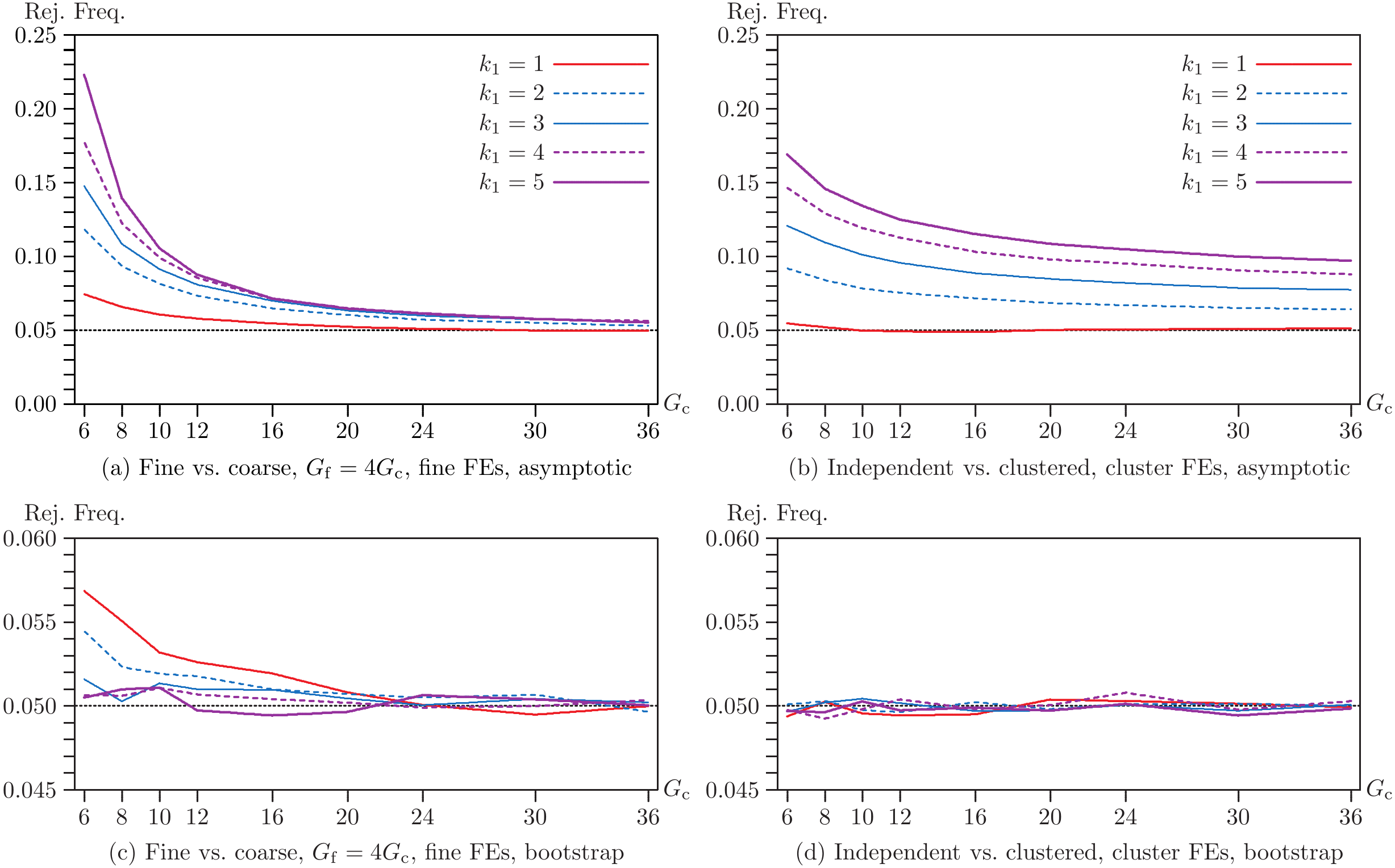}
\end{center}
{\footnotesize
\textbf{Notes:} The regressors are generated by \eqref{facDGP} with
$\rho=0.5$ and $1 \leq k_1 \leq 5$. The regressand is generated 
by~\eqref{simmod}. The disturbances are independent standard normals in
Panels~(b) and~(d) and finely clustered with $\rho=0.1$ in Panels~(a)
and~(c). $G_{\rm c}$ denotes the number of coarse clusters. Each
coarse cluster contains 400 observations, so that $N=400G_{\rm c}$. In 
Panels~(a) and~(c), there are $G_{\rm f} = 4G_{\rm c}$ fine clusters, 
each containing 100 observations. Bootstrap tests employ $B=399$. 
Panel~(c) uses the wild cluster bootstrap, and Panel~(d) uses the 
ordinary wild bootstrap. There are 400,000 replications.}
\end{figure}

\Cref{fig:1} plots rejection frequencies at the 0.05 level for
$\tau_\Sigma$ tests against $G_{\rm c}$, which varies from 6 to 36. We
started at $G_{\rm c}=6$ to avoid singularities when $k_1=5$ and
stopped at $G_{\rm c}=36$ because the results were hardly changing at
that point. The values $k_1=1,\ldots,5$ imply that the number of
degrees of freedom for the tests is 1, 3, 6, 10, or~15. Panels~(a)
and~(c) concern tests of fine clustering against coarse clustering,
and panels~(b) and~(d) concern tests of independence against
clustering. The top two panels report rejection frequencies for
asymptotic tests at the 0.05 level, and the bottom two report
comparable ones for bootstrap tests. Notice that the vertical axes for
the asymptotic tests are much longer than the ones for the bootstrap
tests, because the latter work very much better.

One striking feature of \Cref{fig:1} is that, for the asymptotic
tests, over-rejection increases sharply with~$k_1$. This should not
have been a surprise in view of the fact that, like the information
matrix test \citep{White_1982}, the $\tau_\Sigma$ test has degrees of
freedom that are $O(k_1^2)$. \citet{DM_1992} found a similar tendency
for the rejection rate of the information matrix test (in particular,
the popular $N\tn R^2$ form of it) to increase rapidly with the number
of coefficients being tested.

When $G_{\rm c}$ is small, asymptotic tests of fine against coarse
clustering, in Panel~(a), over-reject more severely than tests of
independence, in Panel~(b). When $k_1=1$, there is almost no
over-rejection for the tests of independence in Panel~(b). For
$k_1\ge2$, there is also more over-rejection in Panel~(a) than in
Panel~(b) when $G_{\rm c}=6$, but the over-rejection diminishes much
more rapidly as $G_{\rm c}$ increases in Panel~(a) than in Panel~(b).

The bootstrap versions of the tests perform very much better than the
asymptotic ones. There is slight over-rejection in Panel~(c) for
smaller values of $G_{\rm c}$, which is really only noticeable for
$k_1=1$ and $k_1=2$. In Panel~(d), the bootstrap tests of independence
work perfectly, except for experimental errors.

The bootstrap tests can be computationally demanding when the sample
size is large, particularly for larger values of $k_1$. This is
especially true for tests where the null hypothesis is no clustering,
because the calculations in \eqref{Sigmac}, \eqref{Sigmaf}, and
\eqref{var2sided} involve score vectors of which the size is the
number of clusters under the null hypothesis. This number is $N$ for
tests of no clustering but only $G_{\rm f}$ for tests of fine
clustering.

\begin{figure}[tb]
\begin{center}
\caption{Rejection frequencies for $\tau_\Sigma$ tests at 
0.05 level, $G_{\rm f}$ or $N_g$ varying}
\label{fig:2}
\includegraphics[width=\textwidth]{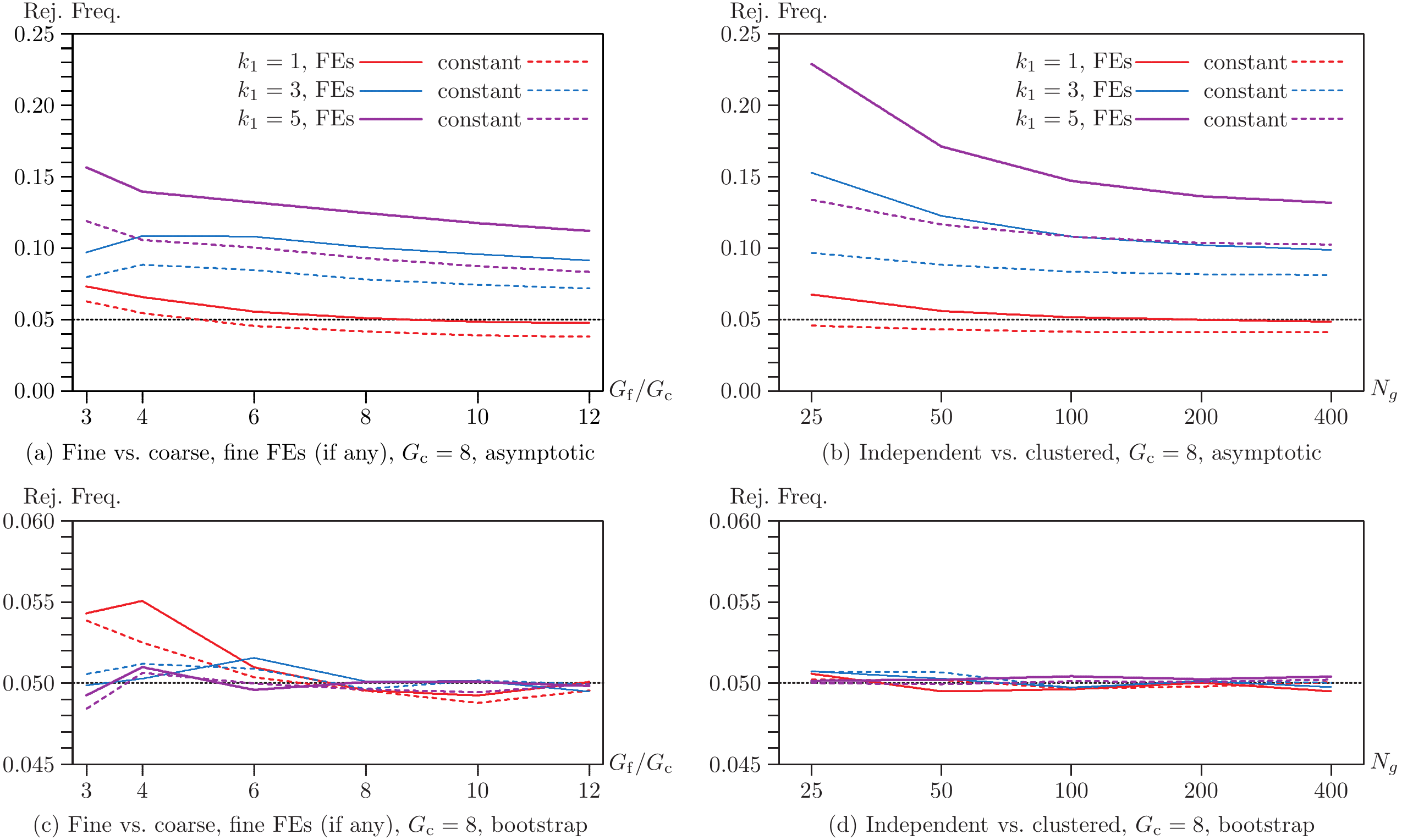}
\end{center}
{\footnotesize
\textbf{Notes:}
The regressors are generated as in \Cref{fig:1}, but only for $k_1=1$,
3, and~5. There are $G_{\rm c}=8$ coarse clusters. In Panels~(a)
and~(c), there are between 3 and 12 fine clusters per coarse cluster,
each with 100 observations. In Panels~(b) and~(d), there are just coarse
clusters, with between 25 and 400 observations per coarse cluster.
Bootstrap tests employ $B=399$. Panel~(c) uses the wild cluster
bootstrap, and Panel~(d) uses the ordinary wild bootstrap. There are
400,000 replications.}
\end{figure}

In \Cref{fig:2}, we hold the number of coarse clusters constant at
$G_{\rm c}=8$ and allow either the number of fine clusters per coarse
cluster or the $N_g$ to vary. Results are shown for two specifications
of \eqref{simmod}. For the first of these, there are fine fixed
effects when the null is fine clustering and cluster fixed effects
when the null is independence, as in \Cref{fig:1}. For the second,
there is just a constant term. To make the figure readable, results
are shown only for $k_1=1$, 3, and~5.

In Panels~(a) and~(c), the horizontal axis shows the number of fine
clusters per coarse cluster, which varies between 3 and~12, so that
the total number of fine clusters varies between 24 and~96. The
rejection frequencies for asymptotic tests of fine against coarse
clustering drop somewhat as $G_{\rm f}/G_{\rm c}$ increases. For the
model with fixed effects, the asymptotic tests for $k_1=1$ work almost
perfectly for $G_{\rm f}/G_{\rm c} \ge 8$, and all the bootstrap tests
work almost perfectly for $G_{\rm f}/G_{\rm c} \ge 6$. The asymptotic
tests always reject less often for the model with a constant term than
for the model with fixed effects. For $k_1=1$, the former actually
under-reject modestly for larger values of $G_{\rm f}/G_{\rm c}$.

In Panels~(b) and~(d), the horizontal axis shows the number of
observations per coarse cluster, which varies between 25 and 400, on a 
log scale. It is evident that the asymptotic tests of independence 
perform better as the clusters become larger, although the curves are 
pretty flat at $N_g=400$. The asymptotic tests with just a constant
over-reject much less than the tests with fixed effects. When $k_1=1$,
these tests under-reject for all values of~$N_g$. All the bootstrap
tests work essentially perfectly.

\begin{figure}[tb]
\begin{center}
\caption{Rejection frequencies for $\tau_\sigma$ tests at 
0.05 level, $G_{\rm c}$ varying}
\label{fig:3}
\includegraphics[width=\textwidth]{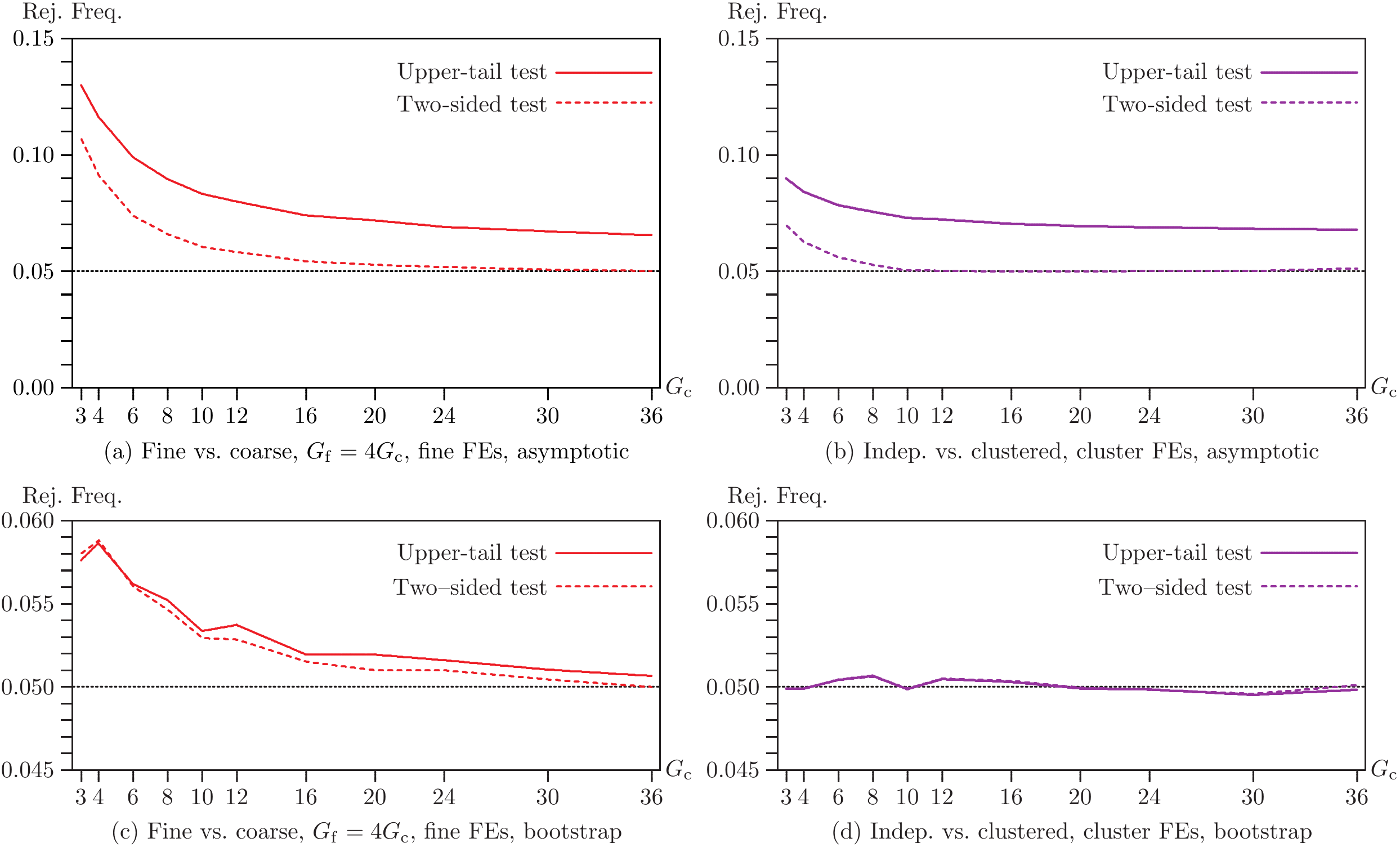}
\end{center}
{\footnotesize
\textbf{Notes:} There is one regressor, which is generated by
\eqref{facDGP} with $\rho=0.5$. In Panels~(a) and~(c), the disturbances
are finely clustered with~$\rho=0.1$ and $G_{\rm f} = 4G_{\rm c}$ fine 
clusters, each with 100 observations. In Panels~(b) and~(d), they are
independent standard normals. $G_{\rm c}$ denotes the number of coarse
clusters, each of which contains 400 observations, so that
$N=400G_{\rm c}$. Bootstrap tests employ $B=399$. Panel~(c) uses the 
wild cluster bootstrap, and Panel~(d) uses the ordinary wild bootstrap. 
There are 400,000 replications.}
\end{figure}

\Cref{fig:3} shows rejection frequencies for both upper-tail and
two\tkk-sided $\tau_\sigma$ tests. The experimental design is
essentially the same as for \Cref{fig:1}, except that, since $k_1=1$,
results for $G_{\rm c}=3$ and $G_{\rm c}=4$ are included. The
asymptotic upper-tail tests over-reject noticeably more often than the
asymptotic two\tkk-sided tests. In contrast, the bootstrap upper-tail 
and two\tkk-sided tests perform identically (and extremely well). Thus
it seems to be valuable to bootstrap both types of $\tau_\sigma$ test,
but particularly important to bootstrap upper-tail tests.

\subsection{The Power of Bootstrap Tests}
\label{subsec:alt}

In the next set of experiments, we turn our attention to power,
focusing on the special case of the $\tau_\sigma$ test for a single
coefficient. The data are generated by \eqref{simmod}, with one
regressor and coarse fixed effects. As usual, the regressor is
generated by \eqref{facDGP} with coarse clustering and $\rho=0.5$. The
disturbances are generated by the same model, with $\rho$ varying
between 0.00 and~0.10. We report results only for bootstrap tests with
$B=999$. Using 999 instead of 399 reduces the, already quite small,
power loss caused by using a finite number of bootstrap samples
\citep{DM_2000}.

\begin{figure}[tb]
\begin{center}
\caption{Power of bootstrap $\tau_\sigma$ tests at 0.05 level when there
is coarse clustering}
\label{fig:4}
\includegraphics[width=\textwidth]{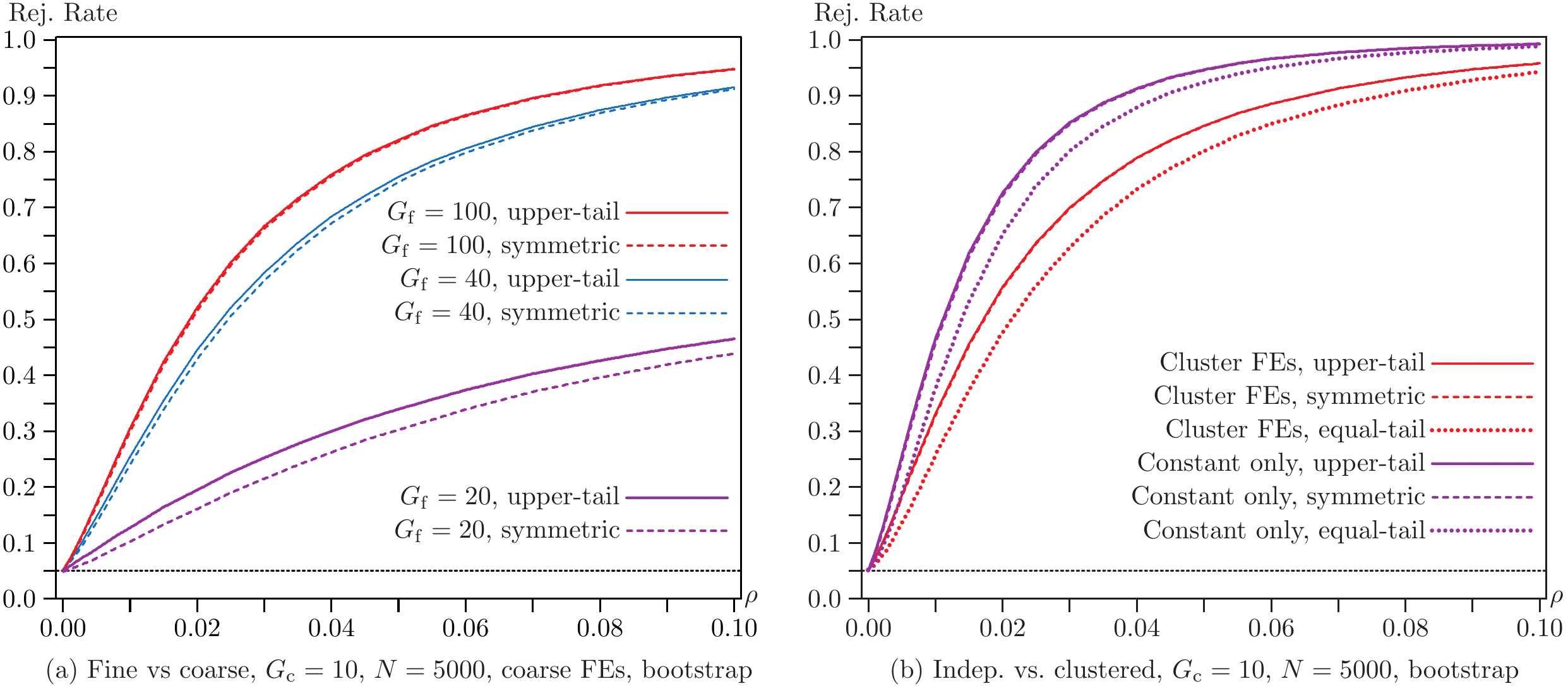}
\end{center}
{\footnotesize
\textbf{Notes:} The data are generated by \eqref{simmod} with coarse
fixed effects and coarse (or no) clustering. There are 5000
observations, 10 coarse clusters, 20, 40, or 100 fine clusters,
400,000 replications, and 999 bootstraps.}
\end{figure}

\Cref{fig:4} shows the power of either two or three types of bootstrap
$\tau_\sigma$ tests against coarse clustering as a function of the
value of $\rho$ for the disturbances. The three types are upper-tail,
symmetric, and equal-tail. As can be seen in both panels, all tests
reject extremely close to 5\% of the time when the null hypothesis is
true. In Panel~(a), the null hypothesis is fine clustering for three
different values of~$G_{\rm f}$. Power increases greatly when the
number of fine clusters goes from 20 to~40. It increases further, but
much more modestly, when $G_{\rm f}$ goes from 40 to~100. The
upper-tail tests are more powerful than the symmetric ones, but only
slightly more when $G_{\rm f}=100$. To avoid making the figure
unreadable, Panel~(a) omits the equal-tail tests, which have much less
power than the other two types of tests.

In Panel~(b), the null hypothesis is independence, and the alternative
is clustering with 10 (coarse) clusters. There are three bootstrap
tests for a model with just a constant term and three tests for a
model with cluster fixed effects. As expected, the tests are more
powerful when there is just a constant term, since the fixed effects
explain some of the intra-cluster correlation. For each set of tests,
the upper-tail test is slightly more powerful than the symmetric test,
which in turn is substantially more powerful than the equal-tail test.

The results in Panel~(b) illustrate the fact that it generally makes
no sense to use equal-tail bootstrap SV tests. These tests are
designed to reject equally often in each tail under the null
hypothesis. Since the mean of the test statistics under the null is
positive in our experiments, the equal-tail test implicitly uses
asymmetric critical values, with the positive one being larger in
absolute value than the negative one. This reduces its power against
$\sigma^2_{\rm c} > \sigma^2_{\rm f}$, which is precisely the
alternative we want SV tests to have power against.

Up to this point, all the simulations have involved equal-sized
clusters. There are many ways in which coarse-cluster sizes,
fine-cluster sizes, and the numbers of fine clusters per coarse
cluster could vary. We next allow cluster sizes to vary for
one level of clustering. \Cref{fig:5} considers $\tau_\sigma$ tests of
no clustering and plots rejection frequencies against a measure of
cluster size variation. The $N$ observations are allocated among $G$
clusters using the equation
\begin{equation}
N_g = \left[\frac{N \exp(\delta g/G)}{\sum_{j=1}^G \exp(\delta
j/G)}\right]
\enspace
\hbox{for } g=1,\ldots,G-1,
\label{gamma-eq}
\end{equation}
where $\delta\ge0$, $[\cdot]$ denotes the integer part of its
argument, and $N_G = N - \sum_{g=1}^{G-1} N_g$. This scheme has been
used in \citet{MW-JAE}, \citet{DMN_2019}, and several other papers. In
the experiments of \Cref{fig:5}, $G=10$ and $N=1000$. When $\delta=0$,
$N_g=100$ for all~$g$. For $\delta=1$, the $N_g$ range from 61 to 155;
for $\delta=2$, from 34 to 213; and for $\delta=4$, from 9 to~340.
There is one regressor and 10 cluster fixed effects.

In Panel~(a) of \Cref{fig:5}, the null hypothesis of no clustering is
true. The upper-tail asymptotic test over-rejects noticeably for small
values of~$\delta$, but rejection frequencies decline as $\delta$
increases, and they are less than 0.05 for $\delta=4$. In contrast, the
upper-tail bootstrap test rejects almost exactly 5\% of the time for
all values of~$\delta$. In Panel~(b), the null hypothesis is false.
Both tests have substantial power when $\delta$ is small, but it falls
as $\delta$ increases. This makes sense, because the total number of
off-diagonal elements in all the clusters increases with~$\delta$,
causing the number of terms in the variance \eqref{tauvar} to
increase. The asymptotic test has noticeably more power than the
bootstrap test for small values of~$\delta$, but it has less power for
the largest values, where it under-rejects under the null. The power
differences almost certainly just reflect the size distortions of the
asymptotic tests.

The results in \Cref{fig:5} suggest that the finite\tkk-sample
performance of SV tests inevitably depends on the pattern of
cluster sizes, although probably much less for bootstrap tests than
for asymptotic ones.

\begin{figure}[tb]
\begin{center}
\caption{Cluster size variation and the performance of upper-tail
$\tau_\sigma$ tests}
\label{fig:5}
\includegraphics[width=\textwidth]{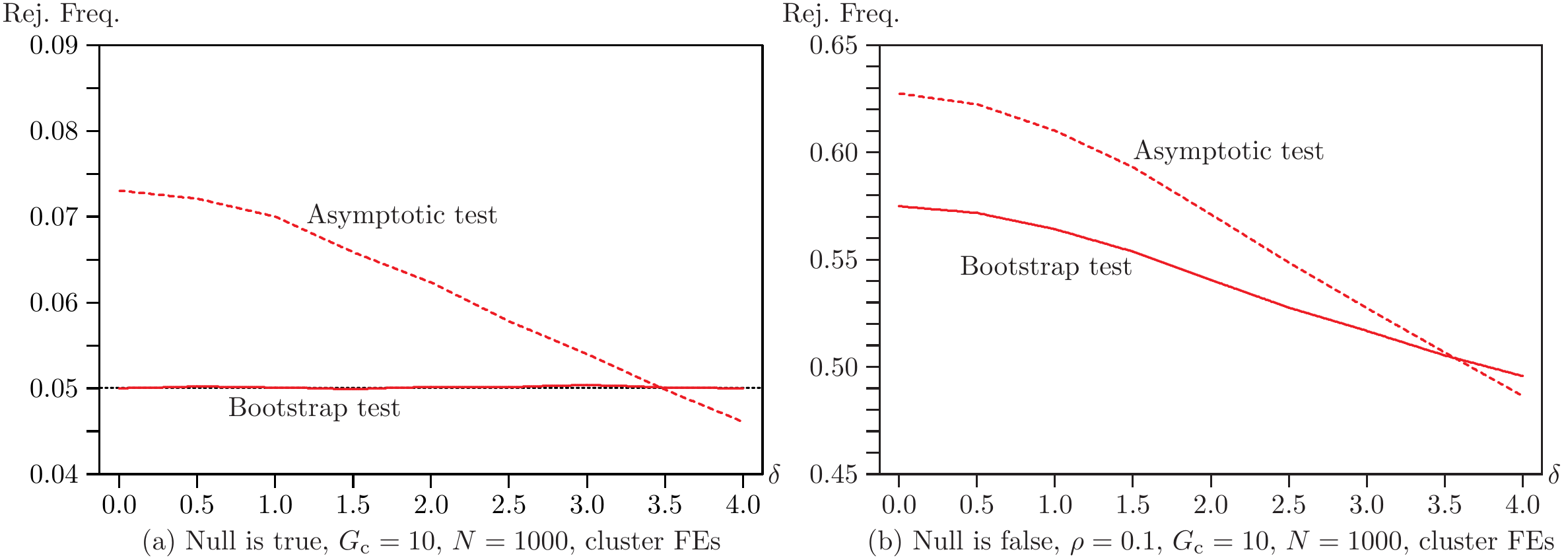}
\end{center}
{\footnotesize
\textbf{Notes:} The regressor is generated by \eqref{facDGP} with
$\rho=0.5$. The disturbances are generated by \eqref{facDGP} with
$\rho=0.0$ in Panel~(a) and $\rho=0.1$ in Panel~b). There are 10
clusters, 1000 observations, and cluster fixed effects. Cluster sizes
vary according to~\eqref{gamma-eq}. All tests are at the nominal 0.05
level. There are 400,000 replications and 399 bootstraps.}
\end{figure}

\subsection{The Sequential Testing Procedure}
\label{subsec:seqsims}

Our next set of experiments concerns the sequential testing procedure
of \Cref{subsec:seq}, using bootstrap tests. These experiments are
quite similar to the ones in \Cref{fig:3,fig:4}, except that there are
8 coarse clusters, 48 fine clusters, and 2400 observations. As in
\Cref{subsec:alt}, there are 999 bootstrap samples. The model always
contains coarse\tkk-level fixed effects, and all of the tests are at
the 0.05 level. The figure shows the outcomes of sequential,
upper-tail $\tau_\sigma$ tests as~$\rho$, the intra-cluster
correlation for each set of disturbances generated by \eqref{facDGP},
varies within either coarse or fine clusters.

\begin{figure}[tb]
\begin{center}
\caption{Outcomes for sequential upper-tail bootstrap tests at 0.05
level}
\label{fig:6}
\includegraphics[width=\textwidth]{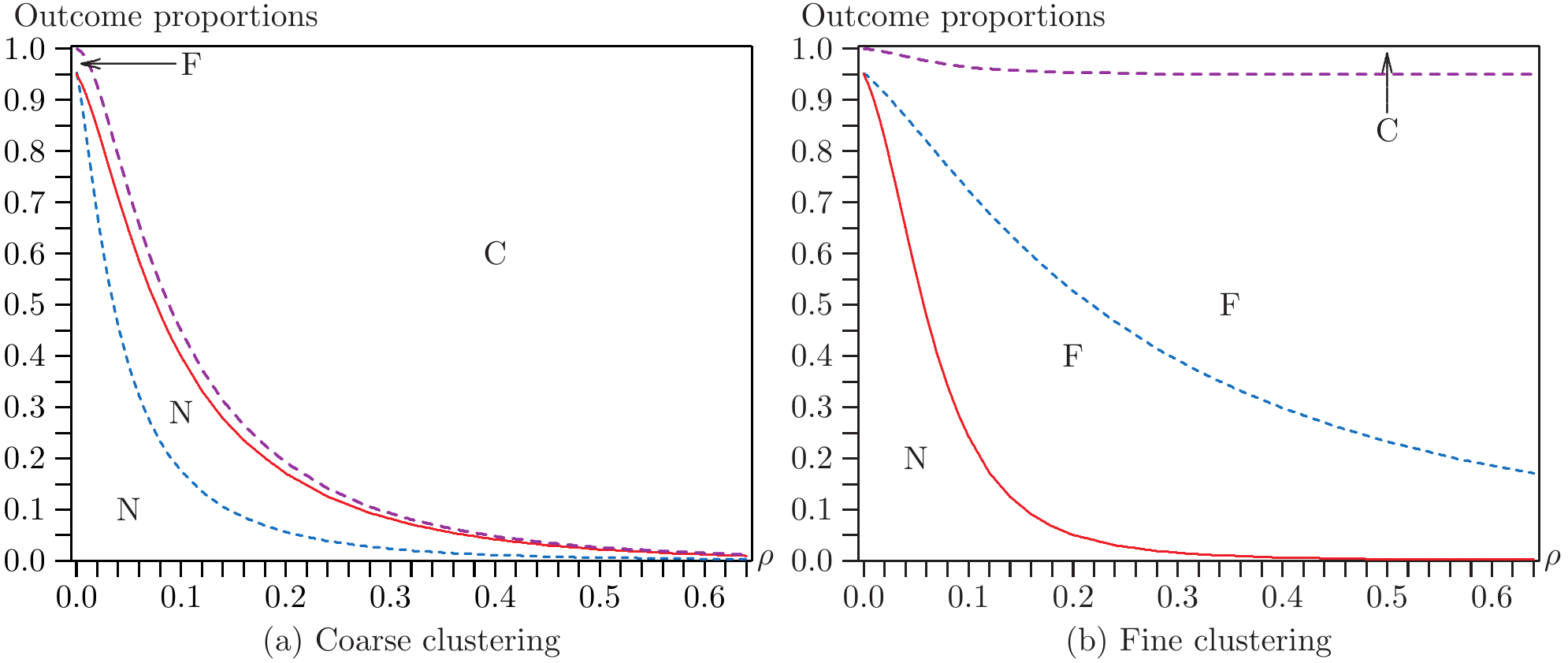}
\end{center}
{\footnotesize
\textbf{Notes:} There is one regressor, generated by \eqref{facDGP}
with $\rho=0.5$, plus cluster fixed effects. The regressand is
generated by \eqref{simmod} with clustered disturbances at either the
coarse level (left panel) or the fine level (right panel), for $\rho$
between 0.0 and~0.64. There are 8 coarse clusters, 48 fine clusters,
and 2400 observations. Bootstrap tests use $B=999$, and there are
400,000 replications. The solid red and dashed purple curves separate
the three outcomes of the sequential procedure; the red curve
separates N from~F, and the purple curve F from~C. The dashed blue
curve shows the outcome of a direct test of N against~C.}
\end{figure}

In Panel~(a) of \Cref{fig:6}, there is coarse clustering in the DGP,
except when $\rho=0$. In that case, as expected, the procedure
chooses no clustering (N) almost exactly 95\% of the time, fine
clustering (F) almost exactly 4.75\% of the time, and coarse
clustering (C) almost exactly 0.25\% of the time. These results 
illustrate why the sequential testing algorithm does not inflate the 
Type~I error. In this case, the true null is rejected almost exactly
$\alpha$\% of the time. Amongst the replications with false 
positives, the test concludes that fine clustering is appropriate
about $(1-\alpha)$\% of the time and that coarse clustering is
appropriate the remaining $\alpha$\% of the time.

As $\rho$ increases, the procedure chooses N or F less and less often. 
For very small values of $\rho$, it chooses N or F more often than C, 
but that changes quickly as $\rho$ increases. The gap between the solid
red and dashed purple curves shows the fraction of the time that F is
(incorrectly) chosen. This gap is always small, and it vanishes as
$\rho$ becomes large.

The sequential procedure inevitably has less power than testing no
clustering directly against coarse clustering. The outcome of testing
N directly against C at the 0.05 level is shown by the blue dashed
curve in Panel~(a). The gap between this curve and the purple dashed 
curve that separates the F and C regions shows the power loss from 
using the sequential procedure. This power loss arises for two reasons.
First, the test of N against F has less power than the test of N 
against~C; see \Cref{fig:4}. Second, even when N is correctly rejected 
against~F, the latter is sometimes not rejected against~C. When the 
investigator finds coarse clustering more plausible than fine 
clustering, it may therefore make sense to test no clustering directly 
against the former rather than to employ the sequential procedure.

In Panel~(b) of \Cref{fig:6}, there is fine clustering in the DGP,
except when $\rho=0$. The sequential procedure again works very well.
As $\rho$ increases, it incorrectly chooses no clustering a rapidly
diminishing fraction of the time. For larger values of $\rho$, it 
incorrectly chooses coarse clustering about 5.2\% of the time,
because the bootstrap SV tests over-reject slightly with only 
8 coarse and 48 fine clusters. Once again, the outcome of testing N 
directly against~C is shown by the blue dashed line. This test works 
much less well than the sequential procedure, often failing to reject 
the false null hypothesis that the disturbances are not clustered.
This is not surprising, since the alternative involves clustering at a
coarser level than the DGP.

\subsection{Making Inferences about a Regression Coefficient}
\label{subsec:pretest}

In \Cref{subsec:infreg}, we discussed several procedures for making
inferences about a single regression coefficient when clustering may
be either fine or coarse. We now investigate some of these procedures,
notably pre\tkk-test ones based on SV tests. There are four simulation
experiments, each involving 12 coarse clusters. In two of them, we
pre\tkk-test the null of no clustering, and in the other two we
pre\tkk-test the null of fine clustering with 96 fine clusters.

The model is a variant of \eqref{simmod}, with eight regressors plus
coarse-level fixed effects, so that $k=G_{\rm c}+8=20$. The regressors 
are generated by \eqref{facDGP} with $\rho=0.5$. The disturbances
$u_{ghi}$ are generated as a convex combination of two disturbances,
$\epsilon^{\rm c}_{gi}$ and $\epsilon^{\rm f}_{ghi}$, with weights
$\eta$ and $1-\eta$ respectively, rescaled so that the $u_{ghi}$ have
unit variance. The $\epsilon^{\rm c}_{gi}$ are generated by
\eqref{facDGP} with $\rho=0.25$. When the pre\tkk-test null hypothesis 
is fine clustering, the $\epsilon^{\rm f}_{ghi}$ are generated in the 
same way as the $\epsilon^{\rm c}_{gi}$, but for 96 fine clusters 
instead of 12 coarse ones. When the pre\tkk-test null hypothesis is no 
clustering, the $\epsilon^{\rm f}_{ghi}$ are i.i.d.\ normal.

The parameter $\eta$ determines the amount of correlation within
coarse clusters. The pre\tkk-test null hypotheses are true when 
$\eta=0$, so that there is either no intra-cluster correlation or only 
correlation within the fine clusters. The pre\tkk-test null hypotheses 
are false when $\eta>0$, and the DGP moves further away from the 
pre\tkk-test null as $\eta$ increases. In the experiments, we vary 
$\eta$ from 0 to~1.

There are several asymptotically valid standard errors for coarse
clustering, fine clustering, and no clustering. The best-known
variance matrix estimator with clustering, often referred to as
CV$_1$, is the usual sandwich estimator \eqref{covbeta} with 
$\hat\bSigma_{\rm c}$ given by \eqref{Sighat} or \eqref{Sigmac}.
However, recent work \citep{Hansen-jack,MNW-bootknife,MNW-influence}
suggests that the cluster jackknife, or CV$_3$, estimator usually
performs better than CV$_1$, so we use the former for inference about
the regression coefficient. For the case of no clustering, we use the
HC$_3$ standard error of \citet{MW_1985}, which is a jackknife
estimator analogous to CV$_3$.

We focus on inference about $\beta_1$, one of the $\beta_\ell$ in
\eqref{simmod}. The pre\tkk-test estimators that we study are based on
upper-tail $\tau_\sigma$ tests. Upper-tail tests are more powerful
than two\tkk-sided tests, so that the former make fewer Type~II
errors; see \Cref{fig:4}. Moreover, even when the difference between
$\var_{\rm c}(\hat\beta_1)$ and $\var_{\rm f}(\hat\beta_1)$ is
positive, $\widehat\var_{\rm c}(\hat\beta_1)$ can be smaller than
$\widehat\var_{\rm f}(\hat\beta_1)$. This happens frequently in our
experiments when $\eta$ is greater than~0 but small. Thus, 
investigators who do not wish to reject fine clustering in favor of
coarse clustering when the coarse standard error is smaller than the
fine one will choose to employ upper-tail pre\tkk-tests.

\begin{figure}[tb]
\begin{center}
\caption{Root mean squared errors of four standard error estimates}
\label{fig:7}
\includegraphics[width=\textwidth]{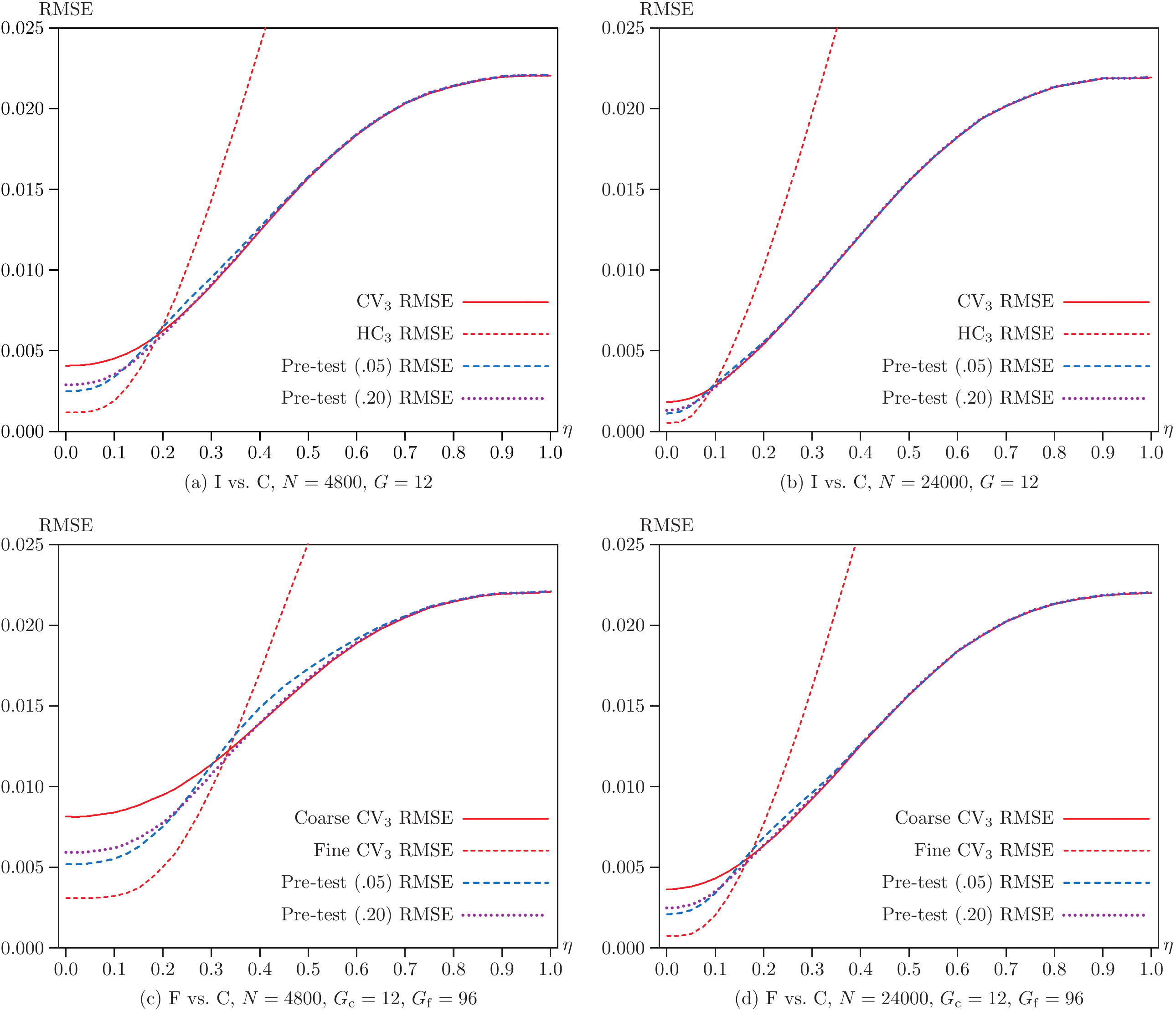}
\end{center}
{\footnotesize \textbf{Notes:} The regressors are generated by
\eqref{facDGP} with coarse clustering and $\rho=0.5$, and the
disturbances are generated as discussed in the second paragraph of
this subsection. When $\eta=0$, there is either no clustering (top 
panels) or fine clustering (bottom panels), depending on the 
pre\tkk-test null hypothesis. When $\eta>0$, there is coarse clustering.
The pre\tkk-test estimators are based on upper-tail $\tau_\sigma$ 
tests. There are $400,\tn000$ replications.}
\end{figure}

The choice among various standard errors is an estimation problem.
Thus, it seems reasonable to compare them on the basis of root mean
squared error (RMSE). When the pre\tkk-test null hypothesis is no 
clustering, the standard error is based on HC$_3$, CV$_3$, or the one 
chosen by pre\tkk-tests at either the 0.05 or 0.20 level. When the 
pre\tkk-test null is fine clustering, the standard error is based on fine 
CV$_3$, coarse CV$_3$, or the one chosen by pre\tkk-tests at the same
two levels. \Cref{fig:7} shows the RMSEs associated with each of these
standard errors. In Panels~(a) and~(b), the pre\tkk-test null hypothesis
is no clustering. In Panels~(c) and~(d), it is fine clustering, with 96
clusters. There are 4800 observations in Panels~(a) and~(c) and 24,000
in Panels~(b) and~(d).

The HC$_3$ or fine CV$_3$ standard errors are the most accurate when
$\eta=0$, and they continue to be the most accurate for small values
of~$\eta$. However, for larger values of $\eta$, they are by far the
least accurate, because they are severely biased. In contrast, the
coarse CV$_3$ standard errors are the least accurate when $\eta$ is
small, but for moderate and larger values of $\eta$ they are the most
accurate. The two pre\tkk-test standard errors are substantially more
accurate than the coarse CV$_3$ ones for small values of~$\eta$ and
almost identical to the latter for large values of~$\eta$. In between,
there is always a region where the pre\tkk-test standard errors are
slightly less accurate than the coarse CV$_3$ ones. This is barely 
noticeable for pre\tkk-tests at the 0.20 level, but it is quite
noticeable for pre\tk-tests at the 0.05 level, especially in
Panel~(c), where the SV tests have the least power.

In our view, the 0.20 pre\tkk-test standard errors in \Cref{fig:7}
perform substantially better than any of the others. They are much
more accurate than coarse CV$_3$ standard errors for small values of
$\eta$, slightly less accurate for some intermediate values, and
essentially identical for larger values. Since using a more accurate
standard error yields a confidence interval that provides a better
sense of how reliable a coefficient estimate is, it seems reasonable
to base confidence intervals on 0.20 pre\tkk-test standard errors.

\begin{figure}[tb]
\begin{center}
\caption{Coverage of four 95\% confidence intervals}
\label{fig:8}
\includegraphics[width=\textwidth]{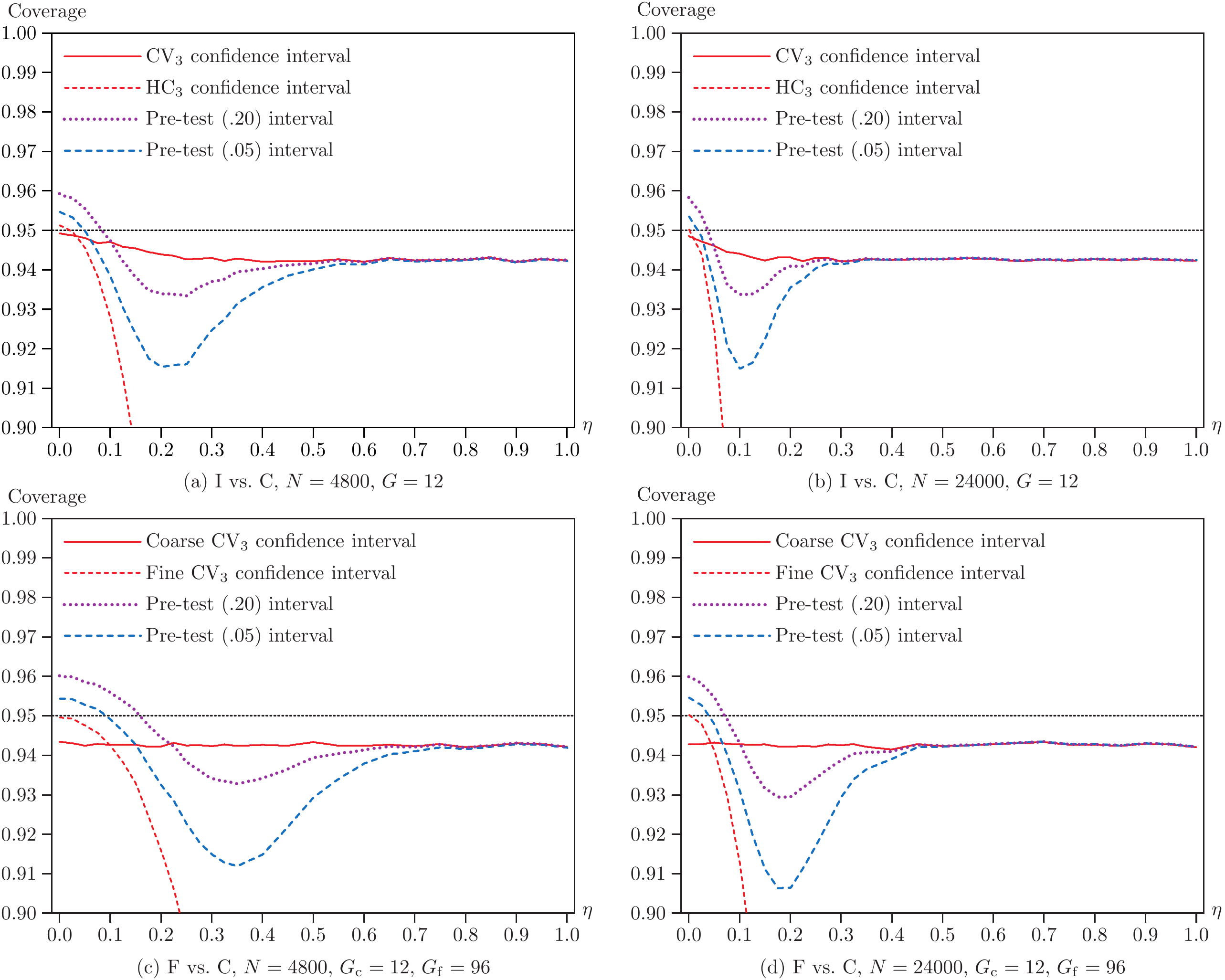}
\end{center}
{\footnotesize \textbf{Notes:} These results are for the same experiments
as in \Cref{fig:7}.}
\end{figure}

Of course, using a more accurate standard error does not guarantee
better coverage. \Cref{fig:8} shows the coverage of confidence intervals
using the four standard errors in \Cref{fig:7}. The coarsely-clustered 
intervals always under-cover to some extent. With only 12 clusters, that
is not surprising. If we had used CV$_1$ instead of CV$_3$ to construct 
the intervals, they would have under-covered to a somewhat greater
extent. On the other hand, coverage would almost certainly have been
closer to 95\% if we had used the wild cluster bootstrap
\citep{MNW-bootknife}, but that would have been computationally very
demanding to simulate. The coverage using HC$_3$ and the 
finely-clustered CV$_3$ is almost exactly 95\% when $\eta=0$, but they 
always under-cover for $\eta>0$, and the under-coverage is very severe 
for most values of~$\eta$. Indeed, their coverage always rapidly drops 
below 0.90, the lower limit of the vertical axis.

The pre\tkk-test intervals over-cover slightly when $\eta=0$, which is
a consequence of Type~I errors in the pre\tkk-tests. However, they
under-cover more than the coarsely-clustered CV$_3$ intervals for
intermediate values of $\eta$ because of Type~II errors. The
under-coverage is much more pronounced for pre\tkk-tests at the 0.05
level than for pre\tkk-tests at the 0.20 level. Because the sample size
is five times larger in Panels~(b) and~(d) than in Panels~(a) and~(c),
the pre\tkk-tests are more powerful, and the pre\tkk-test intervals
converge more rapidly to the coarsely clustered CV$_3$ interval as
$\eta$ increases.

To save computer time and programming effort, we use asymptotic SV
tests in these experiments. In consequence, the levels of the
pre\tkk-tests are not exactly 0.05 and~0.20. In particular, the actual
levels of tests at the 0.20 level are noticeably lower than~0.20, and
the ones for tests at the 0.05 level are somewhat higher than~0.05. If
we had used bootstrap pre\tkk-tests, the under-coverage for moderate
values of $\eta$ would have been a bit smaller for tests at the 0.20
level and a bit larger for tests at the 0.05 level. But all the curves
for pre\tkk-test confidence intervals would have looked very similar.
They would also have looked very similar if we had used CV$_1$ and
HC$_1$ instead of CV$_3$ and HC$_3$.

\section{Empirical Example}
\label{sec:example}

We now illustrate the use of our score\tkk-variance tests in a
realistic empirical setting. We employ the widely-used data from the
Tennessee Student Teacher Achievement Ratio (STAR) experiment
\citep{Finn_1990, Mosteller_1995}. We use these data to estimate a
cross\tkk-sectional model similar to one in \citet{Krueger_1999}. The
STAR experiment randomly assigned students either to small-sized
classes, regular-sized classes without a teacher's aide, or
regular-sized classes with a teacher's aide. We are interested in the
effect of being in a small class, or being in a class with an aide, on
standardized test scores in reading.

We estimate the following cross\tkk-sectional regression model:
\begin{equation}
\label{eq:starcs}
\text{read-one}_{sri} = \alpha + \beta_s\tk \text{small-class}_{sr} 
  +\beta_a\tk \text{aide\tkk-class}_{sr} + \bix^\top_{sri}\bdelta + u_{sri}.
\end{equation}
The outcome variable $\text{read-one}_{sri} $ is the reading score in
grade one of student~$i$ in classroom~$r$ in school~$s$. We are 
interested in $\beta_s$ and $\beta_a$, which are the coefficients for 
the small-class and aide\tkk-class dummies. Small-class equals~1 if a
student attended a small class in grade one and equals~0 otherwise;
aide\tkk-class is constructed in the same way for classes with or
without a teacher's aide. Additional control variables are collected
in the vector of regressors $\bix_{sri}$. These include dummy
variables for whether the student was male, non-white, or received
free lunches, as well as a dummy variable for whether the student's
teacher was non-white. They also include the teacher's years of
experience and the student's reading score in kindergarten. Finally,
there are dummy variables for the student's quarter of birth, the
student's year of birth, and the teacher's highest degree. There are
thus 17 coefficients in total, not counting the constant term or the
school fixed effects, if any.

\begin{table}[tp]
\caption{STAR Example}
\label{tab:starcs}
\vspace*{-0.5em}
\begin{tabular*}{\textwidth}{@{\extracolsep{\fill}}
llcd{2.3}d{1.3}d{1.3}cd{2.3}d{1.3}d{1.3}d{1.3}}
\toprule
&    &   & \multicolumn{3}{c}{Without School FE}
&    &   & \multicolumn{3}{c}{With School FE} \\
\cmidrule{4-6}\cmidrule{9-11}
\multicolumn{2}{l}{Estimates}  & & \multicolumn{1}{c}{HC$_3$(N)}
& \multicolumn{1}{c}{CV$_3$(R)} 
& \multicolumn{1}{c}{CV$_3$(S)} &&
& \multicolumn{1}{c}{HC$_3$(N)} & \multicolumn{1}{c}{CV$_3$(R)} & 
\multicolumn{1}{c}{CV$_3$(S)} \\
\midrule
small & $\hat\beta_s$ & & 9.211 & 9.211 & 9.211 & &&8.095 &8.095 & 8.095 \\
      & s.e. & & 1.633 & 3.273 & 3.253 && &1.556 &3.028 & 3.120 \\ 
      & $t$-stat. & & 5.640  & 2.814  & 2.831 && &5.203 &2.673 & 2.595  \\
aide  & $\hat\beta_a$ & & 6.245 & 6.245 & 6.245 & &&4.170 & 4.170 &4.170  \\
      & s.e. & & 1.664 & 3.343 & 2.847 & &&1.587 & 2.814 & 2.429 \\ 
      & $t$-stat. & & 3.752 & 1.868  & 2.194 && &2.628 & 1.482 & 1.717 \\
\midrule
&    &   & \multicolumn{3}{c}{Without School FE}
&        & \multicolumn{4}{c}{With School FE} \\
\cmidrule{4-6}\cmidrule{8-11}
\multicolumn{2}{l}{Cluster tests} &
& \multicolumn{1}{c}{SV stat.} & 
\multicolumn{1}{c}{asy.\ $P$} & \multicolumn{1}{c}{boot $P$} & & 
\multicolumn{1}{c}{SV stat.} & \multicolumn{1}{c}{asy.\ $P$} & 
\multicolumn{1}{c}{boot $P$} & \multicolumn{1}{c}{IM $P$} \\
\midrule
small    & H$_{\rm N}$ vs H$_{\rm R}$ & &28.388 &0.000 &0.000
 & &12.757 &0.000 &0.000  & \multicolumn{1}{c}{---} \\
 & H$_{\rm N}$ vs H$_{\rm S}$ & &16.409 &0.000 &0.000
 & &18.308 &0.000 &0.000  &0.251 \\
 & H$_{\rm R}$ vs H$_{\rm S}$ & & -0.101 &0.540 &0.543
 & &4.366 &0.000 &0.004 &0.000 \\
aide  & H$_{\rm N}$ vs H$_{\rm R}$ & &25.693 &0.000 &0.000
  & &7.625  &0.000 &0.000  & \multicolumn{1}{c}{---} \\
 & H$_{\rm N}$ vs H$_{\rm S}$ & &10.102 &0.000 &0.000
 & &7.696  &0.000 &0.000 &0.438 \\
 & H$_{\rm R}$ vs H$_{\rm S}$ & & -1.765 &0.961 &0.973
 & &1.871 &0.031 &0.344 & 0.000 \\
both &H$_{\rm N}$ vs H$_{\rm R}$ & &1075.469 &0.000 &0.000
 & &180.448 &0.000 &0.000 & \multicolumn{1}{c}{---} \\
 &H$_{\rm N}$ vs H$_{\rm S}$ & &322.367 &0.000 &0.000
 & &385.950 &0.000 &0.000 & \multicolumn{1}{c}{---} \\
 &H$_{\rm R}$ vs H$_{\rm S}$ & &5.215 &0.157 &0.171
 & &28.673 &0.000 &0.011 & \multicolumn{1}{c}{---} \\
\bottomrule
\end{tabular*}
\vskip 6pt {\footnotesize \textbf{Notes:} There are 3,989 observations
and either 330 classroom clusters (denoted R for ``room'') or 75 school
clusters (denoted~S). The null hypotheses of no clustering, classroom
clustering, and school clustering are called H$_{\rm N}$, H$_{\rm
R}$, and H$_{\rm S}$, respectively. Values of the $\tau_\sigma$
statistic (for ``small'' and ``aide'') or the $\tau_\Sigma$ statistic
(for ``both'') are shown under ``SV~stat.'' All other numbers in the
lower panel are $P$~values. For the $\tau_\sigma$ tests, asymptotic
$P$~values are upper-tail and based on the $\N(0,1)$ distribution. For
the $\tau_\Sigma$ tests, they are based on the $\chi^2(3)$
distribution. Upper-tail bootstrap $P$~values use $B=99,\tn999$. IM
tests use $S=9,\tn999$. Data and \texttt{Stata} files may be found at
\url{http://qed.econ.queensu.ca/pub/faculty/mackinnon/svtest/}.}
\end{table}

OLS estimates for the model \eqref{eq:starcs} are presented in the top
half of \Cref{tab:starcs}. Two variants of the model are estimated. In
the left panel, there is just a constant term. In the right panel,
there are school fixed effects. It is impossible to use classroom
fixed effects, because treatment was assigned at the classroom level.
Three sets of standard errors and $t$-statistics are reported for each
variant of the model. For each set, the first column reports results
that are heteroskedasticity-robust (HR), using HC$_3$ standard errors.
The next two columns report results that are cluster-robust (CR) at
either the classroom~(R) level or the school~(S) level, using CV$_3$
standard errors. As in \Cref{subsec:pretest}, we employ HC$_3$ and 
CV$_3$, instead of the more commonly-used HC$_1$ and CV$_1$ estimators, 
because the former tend to yield more reliable inferences. The
HR results would have been very similar if we had used HC$_1$ instead
of HC$_3$. However, some of the CR results would have been noticeably
different if we had used CV$_1$ instead of CV$_3$. The reason for this
is interesting, and we discuss it below.

Because treatment was assigned at the classroom level, it seems 
plausible that clustering at that level would be appropriate. However,
since there are multiple classrooms per school, and students from the
same school probably have many common characteristics and peer
effects, it might also seem natural to cluster at the school level
instead of the classroom level; even more so if assignment was not
entirely random.

Unfortunately, the dataset does not contain a classroom indicator. One
was created by using the information on the school~ID, teacher's race,
teacher's experience, teacher's highest degree, teacher's career
ladder stage, and treatment status. It is possible that this procedure
occasionally grouped two classes into one class, when two teachers in
the same school had exactly the same observable characteristics.
However, since the largest observed class had only 29 students, this 
seems unlikely to have happened often. Moreover, it would not be a 
problem, because the true classes would always be nested within the 
larger, assumed class. What would be a problem is if classes were 
incorrectly partitioned, but this cannot happen.

For the model without school fixed effects, the estimated impact on
test scores of being in a small class is $\hat\beta_s = 9.211$. Based
on an HR standard error of 1.63, the $t$-statistic for the null
hypothesis that $\beta_s=0$ is~5.64. When we instead use CR standard
errors clustered at the classroom level, the standard error for
$\beta_s$ increases to 3.23, and the $t$-statistic decreases to~2.81.
Using CR standard errors clustered at the school level yields almost
identical results; the standard error is 3.25, and the $t$-statistic
is~2.83. In this case, the level at which we cluster makes no qualitative
difference. For the model with school fixed effects, the estimate of
the small-class effect is somewhat lower at $\hat\beta_s=8.095$. The
HR $t$-statistic is now 5.20, the classroom-level CR $t$-statistic is
2.67, and the school-level CR $t$-statistic is~2.57. Once again, the
level at which we cluster does not change the conclusions.

The estimated effect on test scores of being in a class with an aide
is $\hat \beta_a = 6.245$ without school fixed effects and $\hat
\beta_a = 4.170$ with them. Based on the HR $t$-statistics, there
seems to be fairly strong evidence that $\beta_a \neq 0$ for both
models. However, when we cluster at the classroom level, we cannot
reject this null hypothesis at the 0.05 level for either
specification. When we cluster at the school level, we can do so for 
the model without fixed effects ($P=0.031$), but not for the model
with fixed effects.

The lower panel of \Cref{tab:starcs} shows the values of our SV test
statistics, and the associated upper-tail asymptotic and bootstrap 
$P$~values, for the two coefficients of interest, both individually and
jointly. It also shows results for the IM test for the model with
school fixed effects, when that test can be calculated. For each
specification, we consider three hypotheses: H$_{\rm N}$ is no
clustering with possible heteroskedasticity, H$_{\rm R}$ is
classroom-level clustering, and H$_{\rm S}$ is school-level
clustering. These are nested as $\text{H}_{\rm N} \subseteq
\text{H}_{\rm R} \subseteq \text{H}_{\rm S}$.

For testing H$_{\rm N}$ against H$_{\rm R}$, the SV tests, both
asymptotic and bootstrap, very strongly reject the null in all cases.
IM tests cannot be computed for this hypothesis, because the procedure
requires the model to be estimated classroom by classroom, and the two
treatment variables are invariant at that level. For testing H$_{\rm N}$
against H$_{\rm S}$, the SV tests also very strongly reject the
null in all cases. This is not surprising. Since there is overwhelming
evidence against H$_{\rm N}$ when tested against H$_{\rm R}$, and
classrooms are nested within schools, there is inevitably also strong
evidence against H$_{\rm N}$ when tested against H$_{\rm S}$.

IM tests can be computed when testing against H$_{\rm S}$, but only
for the model with school fixed effects. For both coefficients, the IM
tests suggest that H$_{\rm N}$ should not be rejected. This is
inconsistent with the results of the score\tkk-variance tests and
surprising in view of the standard errors reported in the top part of
the table; see below for further discussion.

The results for testing H$_{\rm R}$ against H$_{\rm S}$ differ
depending on the model, the coefficient(s) of interest, and the
testing procedure. Consider first the model with no fixed effects. 
Here, both $\tau_\sigma$ statistics are negative, so of course 
upper-tail tests do not reject the null. This reflects the fact that, 
for both coefficients, the CR standard errors for school clustering are
smaller than those for classroom clustering. The $\tau_\Sigma$ test
for both coefficients jointly is always two\tkk-sided. With $P$~values 
of 0.157 (asymptotic) and 0.171 (bootstrap), it also fails to reject 
the null hypothesis. Thus we conclude that the classroom level is the 
right one at which to cluster for the model with just a constant term.

Consider next the model with school fixed effects. As we noted in 
\Cref{rem:Moulton,rem:empscores}, the ``correct'' level of clustering 
may be different for different hypotheses. This is what we find here.  
For $\hat\beta_s$, all three SV tests reject the null hypotheses and 
consequently suggest that school clustering is appropriate. In 
contrast, for $\hat\beta_a$, the SV tests suggest quite clearly (at 
least when using bootstrap $P$-values) that classroom clustering is 
appropriate.

Closer examination reveals that, for the model with school fixed
effects, the asymptotic and bootstrap tests for H$_{\rm R}$ against
H$_{\rm S}$ always yield quite different $P$~values. This is easily
seen for $\beta_a$, where the bootstrap $P$~value of 0.344 is more
than ten times the asymptotic $P$~value of~0.031. But it is also true
for the other two tests. For $\beta_{\rm s}$, the $\tau_\sigma$ test
statistic of 4.366 has an asymptotic $P$~value of 0.000006 and a
bootstrap $P$~value of~0.0044. For the joint test of both
coefficients, the $\tau_\Sigma$ test statistic of 28.673 has an
asymptotic $P$~value of 0.000003 and a bootstrap $P$~value of~0.0109.
In the latter two cases, the bootstrap $P$~values are small, but they
are many times larger than the asymptotic ones.

The differences between asymptotic and bootstrap $P$~values for SV
tests of classroom against school clustering in the model with school
fixed effects arise because there are only a few classrooms per
school. The average is 4.4, and most schools have just 3 or 4
classrooms. Because the residuals are orthogonal to the school fixed
effects, they must add to zero over all classrooms in each school.
This mechanically creates negative correlation between the residuals
across classrooms within each school, even if the disturbances are
uncorrelated across classrooms. The negative correlation of the
residuals leads to spurious correlation of the empirical scores
whenever a regressor of interest, after being projected off the fixed
effects and the other regressors, is correlated across classrooms
within schools. Because student characteristics probably vary at the
school level, this sort of correlation seems very likely.

In principle, the spurious correlation of the empirical scores could
be either positive or negative. For the model \eqref{eq:starcs}, it is
evidently positive and quite large. This explains why the bootstrap
tests yield much larger $P$~values than the asymptotic tests.
Equivalently, the bootstrap critical values are greater than the
asymptotic ones. For example, the test statistic for H$_{\rm N}$
against H$_{\rm R}$ for $\beta_a$ is~7.625. The asymptotic critical
value for an upper-tail test at the 0.05 level is 1.645, but the
bootstrap critical value is~3.423.

Whenever there is a dummy variable that affects only a few clusters
(in this case the classrooms within each school), OLS residuals will
be negatively correlated across those clusters, even when the
disturbances are uncorrelated. This distortion of the residuals can
cause cluster-robust inference to be severely misleading; see, among
others, \citet{MW-JAE,MW-EJ} and \citet{Chaisemartin_2022}. However,
CV$_3$ standard errors are almost certainly much more reliable in such
cases than CV$_1$ standard errors. As \citet{MNW-bootknife} explains,
the cluster jackknife implicitly involves transforming the empirical
scores in a way that undoes at least part of the distortion induced by
least squares. This is evidently happening here.

With 330 clusters, we would normally expect CV$_1$ and CV$_3$ standard
errors to be almost identical. But this is not the case for the model
with fixed effects and classroom clustering. The CV$_1$ standard
errors with classroom clustering for $\hat\beta_s$ and $\hat\beta_a$
are 2.322 and 2.109, respectively. These are much smaller than the
CV$_3$ standard errors of 3.028 and 2.814 reported in
\Cref{tab:starcs}. The latter are almost certainly much more reliable
than the former. Note that the CV$_1$ standard error for $\hat\beta_a$
with school clustering is 2.422, which is almost identical to the
CV$_3$ one in the table and greater than~2.109. Thus the ratio of the
S and R standard errors is greater than one for CV$_1$ and less than
one for CV$_3$. Because the former ratio is greater than one, the
$\tau_\sigma$ statistic is positive.

In additional simulation experiments not reported here, we generated
artificial samples using the actual regressors for the STAR model.
When there are no school fixed effects, all the SV tests, both
asymptotic and bootstrap, work very well. However, when there are
fixed effects, the asymptotic tests over-reject severely (up to about
70\% of the time). The bootstrap tests perform almost perfectly when
testing H$_{\rm N}$ against either H$_{\rm R}$ or H$_{\rm S}$, but
they reject between 7\% and 9\% of the time for the tests of H$_{\rm
R}$ against H$_{\rm S}$. We also performed some experiments in which
the number of classrooms per school was doubled. All tests performed
very much better in this case. These results suggest that, when there
are fixed effects at the coarse level with few fine clusters per coarse
cluster, and the asymptotic and bootstrap $P$~values differ sharply,
the former should not be believed, and the latter should be taken with
a grain of salt.

The IM tests are undoubtedly also affected by the odd properties of
OLS residuals with school fixed effects. However, many of the
differences between the score\tkk-variance tests and the IM tests in
\Cref{tab:starcs} probably arise because calculating the latter for
the model \eqref{eq:starcs} is tricky. The problem is that estimating
all the coefficients for every one of the 75 schools is infeasible.
For 34 schools, it is impossible to estimate at least one of $\beta_s$
and $\beta_a$ (17 schools in the case of $\beta_s$ and 21 schools in
the case of~$\beta_a$). This means that the IM tests have to be based
on either 58 or 54 coarse clusters, instead of all~75. Additionally,
the other regressors that are included vary across clusters, so that
the coefficients $\beta_s$ and $\beta_a$ may have different
interpretations for different clusters. The IM tests may effectively
be testing different null hypotheses than the score\tkk-variance
tests, which are always based on estimates for the entire sample.

In summary, our score\tkk-variance tests suggest that clustering at
either the classroom or school level is essential, because the null
hypothesis of no clustering is always strongly rejected against both
alternatives. Which of these levels we should cluster at depends on
the model and the coefficient(s) of interest. With just a constant 
term, the sequential testing procedure, using either asymptotic or 
bootstrap tests, suggests that we should choose H$_{\rm R}$ and cluster 
at the classroom level. However, with school fixed effects, we should
apparently choose H$_{\rm R}$ if interest focuses on $\beta_a$ and
H$_{\rm S}$ if it focuses on $\beta_s$ or on both coefficients. Both
choices lead us to conclude that the effect of small classes is
positive and significant at the 0.05 level, while the effect of a
teacher's aide is also positive but not significant at that level.

The fact that we obtain different results for the three SV tests
should not be surprising. The test statistics depend on empirical
scores, and they are different for the three tests because the $\biZ$
matrices in \eqref{newmodel}, which are vectors for the $\tau_\sigma$
tests, are different; see \Cref{rem:empscores}. For the model with
fixed effects, the residuals are clearly correlated at the school
level. While part of this correlation is evidently spurious and caused
by the fixed effects, the bootstrap results suggest that the
disturbances are surely correlated at the school level, because the
$\tau_\sigma$ test for $\beta_s$ and the $\tau_\Sigma$ test for the
two coefficients both reject quite strongly. For $\beta_a$ by itself,
however, the scores are apparently not correlated, leading the 
$\tau_\sigma$ test not to reject in that case.

\section{Conclusion}
\label{sec:conclusion}

Empirical research that uses cluster-robust inference typically
assumes that the level of clustering is known. When it is unknown, the
consequences can be serious. Clustering at too fine a level can result
in tests that over-reject severely and confidence intervals that
under-cover dramatically. However, clustering at too coarse a level
can lead to loss of power and to confidence intervals that vary
greatly in length across samples and are, on average, excessively
long.

We have proposed two direct tests for the level of clustering in a
linear regression model, which we call score\tkk-variance (or SV)
tests. Both tests are based on the variances of the scores for two
nested levels of clustering, because it is these variances that appear
in the ``filling'' of the sandwich covariance matrices that correspond
to the two levels. Under the null hypothesis that the finer level is 
appropriate, many of these variances are zero. The test statistics are
functions of the empirical counterparts of those variances. Tests
based on them can be used either to test the null of no clustering
against an alternative of clustering at a certain level or to test the
null of ``fine'' clustering against an alternative of ``coarser''
clustering. We have also proposed a sequential procedure which can be
used to determine the correct level of clustering without inflating
the family-wise error rate; see \Cref{subsec:level}.

The simplest of our two tests is based on the statistic $\tau_\sigma$.
It has the form of a $t$-statistic and tests whether the variance of a
particular coefficient estimate is the same for two different levels
of clustering. It will be attractive whenever interest focuses on a
single coefficient, and it can be implemented as either a
one\tkk-sided, upper-tail test or as a two\tkk-sided test. Since
upper-tail $\tau_\sigma$ tests have more power than two\tkk-sided ones
(\Cref{subsec:alt}), we believe that they will usually be the
procedure of choice. The second variant, based on the Wald-like
statistic $\tau_\Sigma$, tests whether the covariance matrix of a
vector of coefficient estimates is the same for two different levels
of clustering. It is necessarily two\tkk-sided.

Our tests can be implemented as either asymptotic tests or as wild
bootstrap tests. In \Cref{sec:theory} and \Cref{sec:proofs}, we derive
the asymptotic distribution of our tests, prove that they are
consistent tests, and also prove the validity of the wild bootstrap
implementations. In the simulation experiments of
\Cref{sec:simulations}, the asymptotic tests often work well for tests
of a single coefficient, but they can be seriously over-sized for
tests of several coefficients. The problem is most severe when testing
a moderate number of fine clusters against a small number of coarse
clusters. For the empirical example of \Cref{sec:example}, where
several regressors, including the key ones, vary only at the
fine\tkk-cluster level, the asymptotic tests seem to be quite
over-sized when there are school fixed effects. When the asymptotic
tests are seriously over-sized, the bootstrap tests always perform
much better.

Our score\tkk-variance tests are very different from the other tests
for the correct level of clustering proposed in \citet{Ibragimov_2016}
and \citet{Cai_2022}; see \Cref{subsec:other}. All these tests may
provide valuable information, although we believe that SV tests are
particularly intuitive. As we discuss in \Cref{subsec:infreg}, SV
tests can be used either as formal pre\tkk-tests for choosing the
level at which to cluster or simply as robustness checks.

Both our simulation results and the empirical example suggest that SV
tests can have excellent power. In many cases, with both actual and
simulated data, the value of the test statistic is so far beyond any
reasonable critical value that we can reject the null hypothesis with
something very close to certainty even without bothering to use the
bootstrap. However, when our tests are used as pre\tkk-tests to choose
the level of clustering, they inevitably make some Type~I errors when
the true clustering level is fine, and they inevitably make some Type~II
errors when the true clustering level is coarse but the sample size
and the extent of coarse clustering are not large enough for rejection
to occur all the time; see \Cref{subsec:pretest}.

The score\tkk-variance tests we have proposed are intended to provide
guidance for applied researchers. In our view, it should be routine to
report the results of SV tests whenever more than one level of
clustering is plausible. This is especially important when
investigators are considering the use of heteroskedasticity-robust
standard errors or clustering at a very fine level, such as by
individual or by family. In practice, however, it may be safest to
report inferences based on more than one level of clustering, along
with the outcomes of SV tests, as we did in \Cref{sec:example}.


\appendix 
\numberwithin{equation}{section} 
\numberwithin{figure}{section} 
\numberwithin{table}{section} 

\makeatletter 
\def\@seccntformat#1{\@ifundefined{#1@cntformat}
  {\csname the#1\endcsname\quad}
  {\csname #1@cntformat\endcsname}}
\newcommand\section@cntformat{}
\makeatother

\section{Appendix A: Proofs of Main Results}
\label{sec:proofs}

\subsection{Proof of \protect\Cref{thm:asy}}
\label{sec:proof1}

We give the proof for $\tau_\sigma$ only because the proof for 
$\tau_\Sigma$ is essentially the same but with more complicated 
notation. Also, because the factors $m_{\rm c}$ and $m_{\rm f}$ both 
converge to~1, we can ignore them in the proof.

Recall the contrast $\theta = \sum_{g=1}^G\sum_{h_1=1}^{M_g}
\sum_{h_2\neq h_1}^{M_g} s_{gh_1} s_{gh_2}$ defined in 
\eqref{thetaknown}. To prove the first result of the theorem, we show 
that
\begin{align}
\label{proof1a}
\frac{\hat\theta - \theta}{\sqrt{\var (\theta)}}
&\Pto 0 \quad\textrm{and}\\
\label{proof1b}
\frac{\theta}{\sqrt{\var (\theta)}} &\dto \N (0,1).
\end{align}

Under \Cref{as:cluster,as:moments}, it holds that $\sigma_g^2 = 
\sum_{h=1}^{M_g}\sigma^2_{gh}$. From \eqref{tauvar} and 
\Cref{Lemma:varhat} we then find that
\begin{equation}
\label{lowerbound}
\var (\theta ) = 2 \sum_{g=1}^G \sum_{h_1=1}^{M_g}
\sum_{h_2\neq h_1}^{M_g} \sigma^2_{gh_1} \sigma^2_{gh_2} 
\geq c \sum_{g=1}^G \sigma_g^4.
\end{equation}
It follows from \Cref{Lemma:theta}(i) and \eqref{lowerbound} that the 
left-hand side of \eqref{proof1a} is 
\begin{equation*}
O_P \left( \frac{\sup_{g,h}N_{gh}\sup_g N_g}{(\sum_{g=1}^G \sigma_g^4 )^{1/2}}
\right) = o_P (1)
\end{equation*}
by the first condition of \Cref{as:size}. This proves \eqref{proof1a}.

To prove \eqref{proof1b}, we write
\begin{equation}
\label{MDS}
\theta = \sum_{g=1}^G \sum_{h=1}^{M_g} w_{gh}
\quad\textrm{with}\quad
w_{gh}=2 s_{gh}\sum_{j=1}^{h-1} s_{gj},
\end{equation}
where we note that $w_{gh}$ is a martingale difference sequence with 
respect to the filtration $\mathcal F_{gh} = \sigma ( \{ s_{mn} 
\}_{m=1,\ldots,g-1,n=1,\ldots,M_m},\{ s_{gn} \}_{n=1,\ldots,h} )$,
i.e.\ $\E (w_{gh} | \mathcal F_{g,h-1})=0$.
Then \eqref{proof1b} follows from the martingale central limit theorem 
\citep[e.g.,][Theorem~2]{Brown_1971} if
\begin{align}
\label{CLT1}
&\var (\theta)^{-\lambda}\sum_{g=1}^G \sum_{h=1}^{M_g} 
\E | w_{gh}|^{2\lambda} \longto 0 \quad\textrm{for some } \lambda >1 , \\
\label{CLT2}
&\var (\theta)^{-1}\sum_{g=1}^G \sum_{h=1}^{M_g} \E ( w_{gh}^2 | 
\mathcal F_{g,h-1} ) \Pto 1.
\end{align}

We first prove the Lyapunov condition in \eqref{CLT1}. We find 
$\E | s_{gh} |^{2\lambda} \leq cN_{gh}^{2\lambda}$ by 
\Cref{LemmaS}. We also find that
\begin{equation}
\label{mom-sum}
\E \Big| \sum_{j=1}^{h-1} s_{gj} \Big|^{2\lambda}
\leq c \E \Big| \sum_{j=1}^{h-1} s_{gj}^2 \Big|^{\lambda}
\leq c\Big|\sum_{j=1}^{h-1}(\E s_{gj}^{2\lambda})^{1/\lambda}
   \Big|^{\lambda}
\leq c\Big|\sum_{j=1}^{M_g}(N_{gj}^{2\lambda})^{1/\lambda}\Big|^{\lambda}
\leq c  N_g^\lambda \sup_{g,h}N_{gh}^\lambda ,
\end{equation}
where the first inequality is Marcinkiewicz-Zygmund, the second is 
Minkowski, and the third is due to \Cref{LemmaS}. Thus, we obtain the 
bound
\begin{equation}
\label{Lyapunov}
\E |w_{gh}|^{2\lambda} \leq 2^{2\lambda} \E | s_{gh} |^{2\lambda} 
\E \Big| \sum_{j=1}^{h-1} s_{gj} \Big|^{2\lambda} 
\leq c N_{gh}^{2\lambda} N_g^\lambda \sup_{g,h}N_{gh}^\lambda ,
\end{equation}
and hence
\begin{equation}
\label{Lyapunov2}
\sum_{g=1}^G \sum_{h=1}^{M_g} \E | w_{gh}|^{2\lambda}
\leq c \sup_{g,h}N_{gh}^{3\lambda-1} \sup_g N_g^\lambda N .
\end{equation}
Combining \eqref{lowerbound} and \eqref{Lyapunov2}, the Lyapunov 
condition in \eqref{CLT1} is satisfied by the second condition of 
\Cref{as:size}.

We next prove convergence of the conditional variance in \eqref{CLT2}.
Because $\var(\theta)$ equals $\sum_{g=1}^G\sum_{h=1}^{M_g}\E
(w_{gh}^2)$, we decompose $\E (w_{gh}^2 | \mathcal F_{g,h-1}) -\E
(w_{gh}^2 ) = q_{1,gh} + q_{2,gh}$, where $q_{1,gh}=\sigma_{gh}^2
\sum_{j=1}^{h-1}( s_{gj}^2 -\sigma_{gj}^2 )$ and
$q_{2,gh}=\sigma_{gh}^2 \sum_{j_1=1}^{h-1}\sum_{j_2\neq j_1}^{h-1}
s_{gj_1}s_{gj_2}$. Then \eqref{CLT2} follows if 
\begin{equation}
\label{qm}
\var (\theta)^{-1}\sum_{g=1}^G\sum_{h=1}^{M_g} q_{m,gh} \Pto 0
\quad\textrm{for }m=1,2.
\end{equation}

For $m=1$, we reverse the summations and find that 
$\sum_{h=1}^{M_g}q_{1,gh} = \sum_{h=1}^{M_g}r_{1,gh}$, where 
$r_{1,gh}=( s_{gh}^2 -\sigma_{gh}^2) \sum_{j=h+1}^{M_g}\sigma_{gj}^2$
is mean zero and independent across both $g$ and~$h$. We prove 
convergence in $L_{\lambda}$-norm. We find $\E | r_{1,gh}|^\lambda 
\leq c \E | s_{gh} |^{2\lambda} (\sum_{j=h+1}^{M_g}\sigma_{gj}^2
)^\lambda \leq c N_{gh}^{2\lambda} ( \sum_{j=1}^{M_g} 
\sigma_{gj}^2 )^\lambda = c N_{gh}^{2\lambda} \sigma_g^{2\lambda}$ 
using \Cref{LemmaS} and $\sum_{j=1}^{M_g} \sigma_{gj}^2 = \sigma_g^2$. 
By the Marcinkiewicz-Zygmund and Minkowski inequalities we find that
$\E \big| \sum_{g=1}^G\sum_{h=1}^{M_g} r_{1,gh} \big|^\lambda \leq 
c \big( \sum_{g=1}^G\sum_{h=1}^{M_g}(\E |r_{1,gh}|^\lambda)^{2/\lambda}
\big)^{\lambda/2}$, and hence
\begin{equation*}
\E \Big| \sum_{g=1}^G\sum_{h=1}^{M_g} r_{1,gh} \Big|^\lambda \leq c 
\Big( \sum_{g=1}^G\sum_{h=1}^{M_g} N_{gh}^4 \sigma_g^4\Big)^{\lambda/2}
\leq 
\sup_{g,h} N_{gh}^{3\lambda/2}\sup_g N_g^{\lambda/2}\Big( 
\sum_{g=1}^G \sigma_g^4 \Big)^{\lambda/2} .
\end{equation*}
Combining this with the bound \eqref{lowerbound}, the result \eqref{qm} 
for $m=1$ follows if
\begin{equation*}
\sup_{g,h} N_{gh}^{3\lambda/2}\sup_g N_g^{\lambda/2}\Big( 
\sum_{g=1}^G \sigma_g^4 \Big)^{-\lambda/2} \longto 0 ,
\end{equation*}
which is satisfied by the first condition of \Cref{as:size}.

For $m=2$, we use symmetry and reverse the summations to find 
$\sum_{h=1}^{M_g}q_{2,gh} = \sum_{h=1}^{M_g}r_{2,gh}$, where 
$r_{2,gh} = 2 s_{gh}\sum_{j_1=h+1}^{M_g} \sigma_{gj_1}^2
\sum_{j_2=1}^{h-1} s_{gj_2} = w_{gh} \sum_{j=h+1}^{M_g} 
\sigma_{gj}^2$ is a martingale difference sequence with respect 
to~$\mathcal F_{gh}$. We prove convergence in mean square. By 
\eqref{Lyapunov} with $\lambda =1$ the variance is $\E (r_{2,gh}^2 ) 
\leq \E (w_{gh}^2 ) \sigma_g^4 \leq c N_{gh}^2 N_g \sup_{g,h}N_{gh} 
\sigma_g^4$, and hence
\begin{equation*}
\E \Big( \sum_{g=1}^G\sum_{h=1}^{M_g} r_{2,gh} \Big)^2 
= \sum_{g=1}^G\sum_{h=1}^{M_g} \E ( r_{2,gh}^2 )
\leq c \sup_{g,h}N_{gh}^2 \sup_g N_g^2 \sum_{g=1}^G \sigma_g^4 .
\end{equation*}
Combining this with the bound \eqref{lowerbound}, the result \eqref{qm} 
for $m=2$ follows by the first condition of \Cref{as:size}. This 
completes the proof of \eqref{CLT2} and hence of \eqref{proof1b}. 

It remains to show the second part of \Cref{thm:asy}. This follows 
directly from \Cref{Lemma:varhat} by application of \Cref{as:size} to 
the remainder terms.

\subsection{Proof of \protect\Cref{thm:cons}}

As in the proof of \Cref{thm:asy}, we give the proof for $\tau_\sigma$ 
only, and we ignore the asymptotically irrelevant factors $m_{\rm c}$ 
and $m_{\rm f}$. Under the conditions of \Cref{thm:cons}, and 
specifically under \Cref{as:clusteralt}, we find from \eqref{Sigmas} 
that $\bSigma_{\rm c} = \sum_{g=1}^G \sigma_g^2$. However, it is 
important to note that, under the conditions of \Cref{thm:cons}, 
$\sigma_g^2 = \var ( s_g ) \neq \sum_{h=1}^{M_g}\sigma_{gh}^2$.

We decompose the test statistic as follows:
\begin{equation*}
\frac{\hat\theta}{\widehat\Var (\hat\theta )^{1/2}}=
\frac{\sum_{g=1}^G \sigma_g^2}{\widehat\Var (\hat\theta )^{1/2}}
\bigg(\frac{\hat\theta - \theta}{\sum_{g=1}^G \sigma_g^2}
+\frac{\theta - \E (\theta )}{\sum_{g=1}^G \sigma_g^2}
+\frac{\E (\theta )}{\sum_{g=1}^G \sigma_g^2}\bigg) ,
\end{equation*}
where we note that $\E (\theta )/\sum_{g=1}^G \sigma_g^2 = 
(\bSigma_{\rm c} - \bSigma_{\rm f})\bSigma_{\rm c}^{-1}$ is non-zero in 
the limit under the alternative hypothesis, \H{1} in \eqref{hypotheses}.
Thus, it suffices to prove that
\begin{equation}
\label{proof:alt}
\frac{\theta - \E (\theta )}{\sum_{g=1}^G \sigma_g^2} \Pto 0,
\quad
\frac{\hat\theta - \theta}{\sum_{g=1}^G \sigma_g^2} \Pto 0,
\quad\textrm{and}\quad
\frac{\widehat\Var (\hat\theta )^{1/2}}{\sum_{g=1}^G \sigma_g^2} 
\Pto 0.
\end{equation}

For the first result in \eqref{proof:alt}, we prove convergence in mean 
square. The second moment of the numerator is
\begin{align}
\nonumber
\E\big(\theta - \E (\theta )\big)^2 = \var (\theta )
&=\sum_{g=1}^G \var \Big(\sum_{h_1=1}^{M_g}\sum_{h_2\neq h_1}^{M_g}
s_{gh_1} s_{gh_2}\Big)
=\sum_{g=1}^G \var \Big( s_g^2 - \sum_{h=1}^{M_g} s_{gh}^2\Big) \\
\label{proof:alt1}
&\leq c \sum_{g=1}^G N_g^4
\leq c \sup_g N_g^3 N ,
\end{align}
where the second equality is by \Cref{as:clusteralt} and the penultimate
inequality is by \Cref{LemmaS} (applying the Cauchy-Schwarz inequality 
to the covariance terms). Hence, $\theta - \E (\theta )$ is 
$O_P(\sup_g N_g^{3/2} N^{1/2})$, which proves the first result in 
\eqref{proof:alt} by \Cref{as:sizealt}. The second result in 
\eqref{proof:alt} follows directly from \Cref{Lemma:theta}(ii) and 
\Cref{as:sizealt}. Finally, by the same methods as applied in the proof 
of \eqref{proof3a}, we find that
\begin{equation*}
\widehat\var (\hat\theta ) = 2 \sum_{g=1}^G \sum_{h_1=1}^{M_g} 
\sum_{h_2 \neq h_1}^{M_g} \hat{s}_{gh_1}^2 \hat{s}_{gh_2}^2
= O_P (\sup_{g,h}N_{gh}^2 \sup_g N_g N ),
\end{equation*}
which proves the third result in \eqref{proof:alt} by \Cref{as:sizealt}.

\subsection{Proof of \protect\Cref{thm:boot}}

As in the proofs of \Cref{thm:asy,thm:cons}, we give the proof for 
$\tau_\sigma$ only. The proof for $\tau_\Sigma$ is essentially the
same but with slightly more complicated notation. The bootstrap 
probability measure is denoted $P^\ast$\tn, and expectation under this 
measure is denoted $\E^\ast$\tn. We define the bootstrap contrast
$\theta^\ast = \sum_{g=1}^G \sum_{h_1=1}^{M_g} \sum_{h_2\neq
h_1}^{M_g} s_{gh_1}^\ast s_{gh_2}^\ast$, and similarly the 
bootstrap variance estimator, and so on.

We prove the bootstrap analog of \Cref{thm:asy}, but under the 
conditions of \Cref{thm:boot}, which will establish the required result. 
Specifically, for all $x \in \mathbb{R}$ and all $\epsilon >0$, we 
prove that
\begin{equation}
\label{boot1}
P^\ast \bigg( \frac{\hat\theta^\ast}%
{\sqrt{\var^\ast (\theta^\ast )}} \leq x \bigg) \Pto \Phi (x)
\quad\textrm{and}\quad
P^\ast \bigg(\tk \bigg| \frac{\widehat\var (\hat\theta^\ast)}%
{\var^\ast (\theta^\ast )} -1 \bigg| > \epsilon\tn \bigg) \Pto 0,
\end{equation}
as $N\to\infty$, where $\Phi (x)$ denotes the cumulative distribution 
function of the standard normal distribution. Clearly, \eqref{boot1} 
implies that $P^\ast ( \tau_\sigma^\ast \leq x ) \Pto \Phi (x)$. From 
\Cref{cor:asy} we have the result that $P_0 (\tau_\sigma \leq x) \to 
\Phi (x)$. Because $\Phi (x)$ is everywhere continuous, the desired 
result then follows by application of the triangle inequality and 
Polya's Theorem.

Thus, we need to prove \eqref{boot1}. We first note that, even though 
\Cref{as:cluster} is not imposed, it nonetheless holds by construction
that, under the bootstrap probability measure $P^\ast$, the bootstrap
data are clustered according to the fine structure in
\Cref{as:cluster}. Therefore, the proof of \eqref{boot1} largely
follows that of \Cref{thm:asy}. One main difference is that $\sigma_g^2
= \var ( s_g) \neq \sum_{h=1}^{M_g} \sigma^2_{gh}$ because 
\Cref{as:cluster} is not imposed in \Cref{thm:boot}.

We first establish the bootstrap equivalent of the lower bound 
in~\eqref{lowerbound}, 
\begin{equation}
\label{BSlowerbound}
\var^\ast (\theta^\ast ) \geq c\big(1+o_P(1)\big) \sum_{g=1}^G \Big( 
\sum_{h=1}^{M_g}\sigma_{gh}^2 \Big)^{\!2} ,
\end{equation}
where we have used the fact that $\sigma_g^4$ in \eqref{lowerbound} needs to 
be replaced by $( \sum_{h=1}^{M_g}\sigma_{gh}^2 )^2$ under the 
assumptions of \Cref{thm:boot}. To prove \eqref{BSlowerbound}, we first
use $s_{gh}^\ast = \hat{s}_{gh} v_{gh}^\ast$, where $v_{gh}^\ast$ 
is independent across both $g$ and~$h$, such that, c.f.\ \eqref{tauvar} 
and~\eqref{varfast},
\begin{equation*}
\var^\ast (\theta^\ast ) = \var^\ast\! \Big( \sum_{g=1}^G 
\sum_{h_1=1}^{M_g}\sum_{h_2\neq h_1}^{M_g} s_{gh_1}^\ast 
s_{gh_2}^\ast\Big)
= 2 \sum_{g=1}^G \sum_{h_1=1}^{M_g}\sum_{h_2\neq h_1}^{M_g} 
\hat{s}_{gh_1}^2 \hat{s}_{gh_2}^2 = \widehat\var (\hat\theta ) .
\end{equation*}
The result in \eqref{BSlowerbound} now follows from \Cref{Lemma:varhat} 
by application of \Cref{as:size} to the remainder terms.

We next prove the following four results, which imply \eqref{boot1}. 
For all $x \in \mathbb{R}$ and all $\epsilon >0$, 
\begin{align}
\label{BSproof1a}
&P^\ast \bigg(\tk \bigg| \frac{\hat\theta^\ast - \theta^\ast}%
{\sqrt{\var^\ast (\theta^\ast )}} \bigg| >\epsilon \tn\bigg)  \Pto0, \\
\label{BSproof1b}
&P^\ast \bigg( \frac{\theta^\ast}%
{\sqrt{\var^\ast (\theta^\ast )}} \leq x \bigg) \Pto \Phi (x), \\
\label{BSproof2a}
&P^\ast \bigg(\tk \bigg| \frac{\widehat\var (\theta^\ast)}%
{\var^\ast (\theta^\ast )} -1 \bigg| > \epsilon \tn\bigg) \Pto 0, \\
\label{BSproof2b}
&P^\ast \bigg(\tk \bigg| \frac{\widehat\var (\hat\theta^\ast)%
-\widehat\var (\theta^\ast )}{\var^\ast (\theta^\ast )} 
\bigg| > \epsilon \tn\bigg) \Pto 0,
\end{align}
as $N\to\infty$. The proofs of \eqref{BSproof1a} and \eqref{BSproof1b} 
are nearly identical to the corresponding proofs of \eqref{proof1a}
and \eqref{proof1b}. Similarly, the proofs of \eqref{BSproof2a} and
\eqref{BSproof2b} are nearly identical to the corresponding proofs of
\eqref{proof2a} and \eqref{proof2b}. We therefore merely highlight the
differences.

First, \eqref{BSproof1a} follows by Markov's inequality and application 
of \Cref{Lemma:thetaast}, the lower bound \eqref{BSlowerbound}, and 
\Cref{as:size}.

Consider now \eqref{BSproof1b}. Under the bootstrap probability
measure, $w_{gh}^\ast = 2 s_{gh}^\ast \sum_{j=1}^{h-1} s_{gj}^\ast$
is a martingale difference sequence with respect to the filtration
\begin{equation}
\mathcal F_{gh}^\ast = \sigma \big( \{ v_{mn}^\ast 
\}_{m=1,\ldots,g-1,n=1,\ldots,M_m},\{ v_{gn}^\ast \}_{n=1,\ldots,h} 
\big).
\nonumber
\end{equation}
To verify the bootstrap equivalent of the Lyapunov condition,
we apply the same proof as for \eqref{CLT1}. Replacing $\E$ with
$\E^\ast$\tn, the bounds \eqref{mom-sum}--\eqref{Lyapunov2} hold under
the bootstrap measure with the right-hand sides being $O_P$ of the
indicated order by \eqref{shat}, \eqref{Opbeta}, \eqref{boundX}, and 
\Cref{LemmaS}. Thus, in particular, 
\begin{equation}
\nonumber
\sum_{g=1}^G \sum_{h=1}^{M_g} \E^\ast |w_{gh}^\ast|^{2\lambda}
= O_P \big( \sup_{g,h} N_{gh}^{3\lambda -1} \sup_g N_g N \big),
\end{equation}
which together with \eqref{BSlowerbound} and \Cref{as:size} verifies 
the Lyapunov condition for \eqref{BSproof1b}. For the proof of 
convergence of the conditional variance, we apply the same proof
as for \eqref{CLT2} with $r_{1,gh}^\ast = ( s_{gh}^{\ast 2}
-\hat{s}_{gh}^2 ) \sum_{j=h+1}^{M_g} \hat{s}_{gj}^2$ and
$r_{2,gh}^\ast = w_{gh}^\ast \sum_{j=h+1}^{M_g} \hat{s}_{gj}^2$. For
both terms we prove mean square convergence (because $\lambda \geq 2$). 
The arguments are nearly identical to those in the proof of \eqref{CLT2},
with all bounds being $O_P$ of the indicated order, using \eqref{shat}, 
\eqref{Opbeta}, \eqref{boundX}, and \Cref{LemmaS}. This completes the 
proof of~\eqref{BSproof1b}.

For the proof of \eqref{BSproof2a}, we follow the proof of
\eqref{proof2a} and obtain $q_{3,gh}^\ast = ( s_{gh}^{\ast 2} -
\hat{s}_{gh}^2 )\sum_{j=1}^{h-1} s_{gh}^{\ast 2}$. We then apply 
the same proof as for $q_{3,gh}$ with $\lambda =2$. Specifically, we find
that there exists a set $\mathcal{A}^\ast$ with $P^\ast
(\mathcal{A}^\ast )\Pto 1$, and on this set we have
\begin{equation*}
\var^\ast (q_{3,gh}^\ast ) = O_P \bigg(\! N_{gh}^4 \Big( \sum_{h-1}^{M_g}
\sigma_{gh}^2 \Big)^{\!2} \bigg),
\end{equation*}
where we used again \eqref{shat}, \eqref{Opbeta}, \eqref{boundX}, and 
\Cref{LemmaS}. Because $q_{3,gh}^\ast$ is a martingale difference 
sequence, the proof of \eqref{BSproof2a} is concluded in the same way 
as that of \eqref{proof2a}.

Finally, we prove \eqref{BSproof2b}. As in 
\eqref{proof3a}--\eqref{proof4c}, we write $\widehat\var 
(\hat\theta^\ast) - \widehat\var (\theta^\ast )$~as
\begin{align}
\label{BSproof3a}
2& (\hat\beta_1^\ast - \hat\beta_1)^2 \sum_{g=1}^G \sum_{h_1=1}^{M_g} 
\sum_{h_2\neq h_1}^{M_g}( \hat{s}_{gh_1}^\ast \hat{s}_{gh_2}^\ast
+ s_{gh_1}^\ast s_{gh_2}^\ast)
\sum_{i=1}^{N_{gh_1}} x_{gh_1i}^2 \sum_{j=1}^{N_{gh_2}} x_{gh_2j}^2  \\
\label{BSproof4a}
&+8 (\hat\beta_1^\ast - \hat\beta_1) \sum_{g=1}^G 
\sum_{h_1=1}^{M_g} \sum_{h_2\neq h_1}^{M_g}
s_{gh_1}^{\ast 2} s_{gh_2}^\ast \sum_{i=1}^{N_{gh_2}}x_{gh_2i}^2\\
\label{BSproof4b}
&-4 (\hat\beta_1^\ast - \hat\beta_1)^2\sum_{g=1}^G \sum_{h_1=1}^{M_g} 
\sum_{h_2\neq h_1}^{M_g} \Big( s_{gh_1}^\ast \sum_{i=1}^{N_{gh_2}}
x_{gh_2i}^2 + s_{gh_2}^\ast \sum_{j=1}^{N_{gh_1}} x_{gh_1j}^2 \Big) s_{gh_1} 
\sum_{\ell=1}^{N_{gh_2}} x_{gh_2\ell}^2 \\
\label{BSproof4c}
&-4 (\hat\beta_1^\ast - \hat\beta_1)^2 \sum_{g=1}^G \sum_{h_1=1}^{M_g} 
\sum_{h_2\neq h_1}^{M_g} s_{gh_1}^\ast \Big( \sum_{j=1}^{N_{gh_1}}
x_{gh_1j}^2 \Big) \Big( \sum_{i=1}^{N_{gh_2}}x_{gh_2i}^2 \Big)^{\!2}.
\end{align}
For \eqref{BSproof3a}, \eqref{BSproof4b}, and \eqref{BSproof4c}, we 
use \eqref{boundX} and \eqref{Opbetaast} together with 
\Cref{LemmaS}, and find that 
\begin{equation*}
\E^\ast |\eqref{BSproof3a}| = O_P \bigg( N^{-1}\sup_{g,h} N_{gh} 
\sum_{g=1}^G \sum_{h_1=1}^{M_g} \sum_{h_2\neq h_1}^{M_g} 
N_{gh_1}^2 N_{gh_2}^2 \bigg)
= O_P \big( \sup_{g,h} N_{gh}^3 \sup_g N_g \big).
\end{equation*}
By the same argument, we also find the same bound for \eqref{BSproof4b}
and~\eqref{BSproof4c}. Using \eqref{BSlowerbound} and the first condition
of \Cref{as:size} shows the required result for these terms. For
\eqref{BSproof4a}, we apply the Cauchy-Schwarz inequality as in
\eqref{proof5},
\begin{equation}
\label{BSproof5}
\eqref{BSproof4a}^2 \leq 64 (\hat\beta_1 - \beta_{1,0})^2 
\Big( \sup_{g,h} \sum_{i=1}^{N_{gh_1}} x_{gh_1i}^2 \Big)^{\!2}
\bigg( \sum_{g=1}^G \Big( \sum_{h_1=1}^{M_g} s_{gh_1}^{\ast 2}
 \Big)^{\!2} \bigg)
\bigg( \sum_{g=1}^G \Big( \sum_{h_2=1}^{M_g} s_{gh_2}^\ast
\Big)^{\!2} \bigg) .
\end{equation}
The first two factors on the right-hand side satisfy \eqref{Opbetaast}
and \eqref{boundX}, respectively. The third factor is non-negative,
and, under the bootstrap probability measure, it has a mean which is $O_P \big(
\sum_{g=1}^G ( \sum_{h=1}^{M_g} N_{gh}^2 )^2 \big) = O_P ( \sup_{g,h}
N_{gh} \sup_g N_g N)$ using \eqref{shat}, \eqref{Opbeta},
\eqref{boundX}, and \Cref{LemmaS}. The last factor is non-negative and, under the bootstrap probability measure, it has a mean which is $O_P\big( \sum_{g=1}^G
( \sum_{h=1}^{M_g} \sigma_{gh}^2 )^2 \big)$ using again \eqref{shat},
\eqref{Opbeta}, \eqref{boundX}, and \Cref{LemmaS}. The proof for
\eqref{BSproof5} is now completed in the same way as that of
\eqref{proof5}. This completes the proof of \eqref{BSproof2b} and
hence of \Cref{thm:boot}.

\subsection{Proof of \protect\Cref{thm:seq}}

The result that $P( \hat m \leq m_0 -1) \to 0$ is a direct consequence 
of \Cref{thm:cons} for the asymptotic tests and of \Cref{cor:boot}(ii) 
for the bootstrap tests. In case~(ii), where $m_0 = p$, there is 
nothing more to prove. In case~(i) we have $m_0 \leq p-1$. Because 
$P( \hat m \leq m_0 -1) \to 0$, the sequential procedure will reach 
the test of the null hypothesis $m = m_0$ with probability converging 
to~one. This is a test of a true null, so we find from \Cref{cor:asy} 
and \Cref{cor:boot}(i) that $P(\hat m = m_0) \to 1-\alpha$, which 
proves the theorem.

\subsection{Auxiliary Lemmas}

\begin{lemma}
\label{LemmaS}
Let \Cref{as:moments} be satisfied. Then
\begin{equation*}
\sup_{g,h} N_{gh}^{-\xi} \E \Vert \bis_{gh} \Vert ^\xi = O(1) 
\quad\textrm{and}\quad
\sup_g N_g^{-\xi} \E \Vert \bis_g \Vert^\xi = O(1)
\quad\textrm{for } 1 \leq \xi \leq 2\lambda .
\end{equation*}
\end{lemma}

\begin{proof}
This is Lemma~A.2 of \citet*{DMN_2019}.
\end{proof}

\begin{lemma}
\label{Lemma:theta}
Let \Cref{as:moments,as:X,as:eigen} be satisfied. Let $\hat\btheta$ be 
defined by \eqref{thetaSigma} and also define $\btheta = \sum_{g=1}^G 
\sum_{h_1=1}^{M_g}\sum_{h_2 \neq h_1}^{M_g} \vech ( \bis_{gh_1}
\bis_{gh_2}^\top )$; c.f.\ \eqref{thetaknown}. Then
\begin{itemize}
\item[(i)] Under \Cref{as:cluster}, $\Vert \hat\btheta - \btheta \Vert 
= O_P ( \sup_{g,h}N_{gh} \sup_g N_g ) $.
\item[(ii)] Under \Cref{as:clusteralt}, $\Vert \hat\btheta - \btheta \Vert 
= O_P ( \sup_g N_g^2 ) $.
\end{itemize}
\end{lemma}

\begin{proof}
We give the proof in the scalar case only. The proof for the 
multivariate case is nearly identical but with more complicated notation.

First note that
\begin{equation}
\label{shat}
\hat s_{ghi} = x_{ghi}\hat u_{ghi}
= x_{ghi} ( u_{ghi} - x_{ghi} (\hat\beta - \beta_0))
= s_{ghi} - x_{ghi}^2 (\hat\beta - \beta_0).
\end{equation}
From \eqref{Sigmascalar}, \eqref{thetascalar}, and \eqref{thetaknown}, 
we then find the difference
\begin{align}
\label{diffA}
\hat\theta - \theta 
={}& \sum_{g=1}^G\sum_{h_1=1}^{M_g}\sum_{h_2\neq h_1}^{M_g}
\Big( \sum_{i=1}^{N_{gh_1}}x_{gh_1i}^2 (\hat\beta -\beta_0)\!\Big)
\Big(\sum_{i=1}^{N_{gh_2}}x_{gh_2i}^2 (\hat\beta-\beta_0)\!\Big)\\
\label{diffB}
&-2\sum_{g=1}^G\sum_{h_1=1}^{M_g}\sum_{h_2\neq h_1}^{M_g}
s_{gh_1}\sum_{i=1}^{N_{gh_2}} x_{gh_2i}^2 (\hat\beta - \beta_0) .
\end{align}

Using \eqref{betahat} we find that $\hat\beta - \beta_0 = 
(\bix^\top\bix)^{-1} \sum_{g=1}^G\sum_{h=1}^{M_g} s_{gh}$. Under 
\Cref{as:cluster} we have
\begin{equation}
\label{varS}
\var \Big( \sum_{g=1}^G \sum_{h=1}^{M_g} s_{gh} \Big) 
= \sum_{g=1}^G \sum_{h=1}^{M_g} \var ( s_{gh} )
\leq c \sum_{g=1}^G\sum_{h=1}^{M_g} N_{gh}^2 \leq c N \sup_{g,h} N_{gh}
\end{equation}
using \Cref{as:moments} and \Cref{LemmaS}. Similarly, under 
\Cref{as:clusteralt},
\begin{equation}
\label{varSalt}
\var \Big( \sum_{g=1}^G \sum_{h=1}^{M_g} s_{gh} \Big) 
= \sum_{g=1}^G \var \Big( \sum_{h=1}^{M_g} s_{gh} \Big)
= \sum_{g=1}^G \var ( s_g )
\leq c \sum_{g=1}^G N_g^2 \leq c N \sup_g N_g .
\end{equation}
Hence, using also \Cref{as:X},
\begin{equation}
\label{Opbeta}
\begin{aligned}
|\hat\beta -\beta_0| &= O_P\big( N^{-1/2}\sup_{g,h} N_{gh}^{1/2}
\big) \quad\textrm{under \Cref{as:cluster}},\\
|\hat\beta -\beta_0| &= O_P\big( N^{-1/2}\sup_g N_g^{1/2}\big)
\quad\textrm{under \Cref{as:clusteralt}}.
\end{aligned}
\end{equation}
We also need the simple bounds
\begin{equation}
\label{boundX}
\sup_{g,h} N_{gh}^{-1}\sum_{i=1}^{N_{gh}} x_{ghi}^2 = O_P (1)
\quad\textrm{and}\quad 
\sup_g N_g^{-1}\sum_{h=1}^{M_g}\sum_{i=1}^{N_{gh}} x_{ghi}^2 
= O_P (1) ,
\end{equation}
which follow from the uniform moment bound in \Cref{as:X}. Using 
\eqref{boundX}, we find that the absolute value of the right-hand side 
of \eqref{diffA} is bounded by
\begin{equation}
\label{boundA}
(\hat\beta - \beta_0)^2 \sum_{g=1}^G \Big(\sum_{h=1}^{M_g}
\sum_{i=1}^{N_{gh}} x_{ghi}^2 \Big)^{\!\!2}
= (\hat\beta - \beta_0)^2 O_P \Big( \sum_{g=1}^G N_g^2 \Big) 
= (\hat\beta - \beta_0)^2 O_P \big( N \sup_g N_g \big),
\end{equation}
which proves the result for \eqref{diffA} using \eqref{Opbeta}.

Next, we write \eqref{diffB} as
\begin{align}
\label{diffB1}
\eqref{diffB} ={}& -2\sum_{g=1}^G\sum_{h_1=1}^{M_g}
s_{gh_1}\sum_{h_2=1}^{M_g}\sum_{i=1}^{N_{gh_2}}x_{gh_2i}^2 
(\hat\beta - \beta_0)\\
\label{diffB2}
&+2\sum_{g=1}^G\sum_{h=1}^{M_g}s_{gh}\sum_{i=1}^{N_{gh}} 
x_{ghi}^2 (\hat\beta - \beta_0).
\end{align}
Under \Cref{as:cluster}, \eqref{varS}, \eqref{Opbeta}, and 
\eqref{boundX} show that $|\eqref{diffB1}| = O_P ( \sup_{g,h}N_{gh} 
\sup_g N_g ) $ and that $|\eqref{diffB2}| = O_P ( \sup_{g,h}
N_{gh}^2 )$. Under \Cref{as:clusteralt}, \eqref{varSalt}, \eqref{Opbeta},
and \eqref{boundX} show that $|\eqref{diffB1}| = O_P ( \sup_g N_g^2 )$ 
and that $|\eqref{diffB2}| = O_P ( \sup_{g,h} N_{gh} \sup_g N_g )$. This
proves the result for \eqref{diffB}.
\end{proof}

\begin{lemma}
\label{Lemma:thetaast}
Let \Cref{as:moments,as:X,as:eigen,as:clusteralt} be satisfied. Let 
$\hat\btheta^\ast = \sum_{g=1}^G \sum_{h_1=1}^{M_g}
\sum_{h_2 \neq h_1}^{M_g} \vech ( \hat\bis_{gh_1}^\ast
\hat\bis_{gh_2}^{\ast \top} )$ and $\btheta^\ast = \sum_{g=1}^G 
\sum_{h_1=1}^{M_g}\sum_{h_2 \neq h_1}^{M_g} \vech ( \bis_{gh_1}^\ast
\bis_{gh_2}^{\ast \top} )$. Then
\begin{equation*}
\E^\ast \Vert \hat\btheta^\ast - \btheta^\ast \Vert 
= O_P \big( \sup_{g,h}N_{gh} \sup_g N_g \big) .
\end{equation*}
\end{lemma}

\begin{proof}
The proof is very similar to that of \Cref{Lemma:theta}. Again, we give 
the proof in the scalar case only since the multivariate case is nearly 
identical but with more complicated notation. We first write 
\begin{align}
\label{BSdiffA}
\hat\theta^\ast -\theta^\ast 
={}& \sum_{g=1}^G\sum_{h_1=1}^{M_g}\sum_{h_2\neq h_1}^{M_g} \Big( 
\sum_{i=1}^{N_{gh_1}} x_{gh_1i}^2 (\hat\beta^\ast - \hat\beta)\!\Big)
\Big( \sum_{i=1}^{N_{gh_2}} x_{gh_2i}^2 (\hat\beta^\ast - 
\hat\beta)\!\Big) \\
\label{BSdiffB}
&-2\sum_{g=1}^G\sum_{h_1=1}^{M_g}\sum_{h_2\neq h_1}^{M_g} 
s_{gh_1}^\ast \sum_{i=1}^{N_{gh_2}} x_{gh_2i}^2 
(\hat\beta^\ast - \hat\beta) .
\end{align}
As in \eqref{varS}, we find that
\begin{equation*}
\var^\ast \Big( \sum_{g=1}^G \sum_{h=1}^{M_g} s_{gh}^\ast \Big)
=\var^\ast\Big(\sum_{g=1}^G\sum_{h=1}^{M_g} \hat s_{gh}v_{gh}^\ast\Big)
= \sum_{g=1}^G \sum_{h=1}^{M_g} \hat s_{gh}^2
= O_P \big(N \sup_{g,h}N_{gh} \big),
\end{equation*}
where the second equality uses independence of $v_{gh}^\ast$ across $g$ 
and $h$ and the third equality uses \eqref{shat}, \eqref{Opbeta}, 
\eqref{boundX}, and \Cref{LemmaS}. It follows that
\begin{equation}
\label{Opbetaast}
\var^\ast (\hat\beta^\ast -\hat\beta ) = O_P \big( N^{-1}\sup_{g,h}
N_{gh} \big).
\end{equation}
Using \eqref{boundX} and \eqref{Opbetaast}, we find that $\E^\ast 
|\eqref{BSdiffA}| = O_P ( \sup_{g,h}N_{gh} \sup_g N_g )$ as in 
\eqref{boundA}. By the same argument, see also 
\eqref{diffB1}--\eqref{diffB2}, we find that $\var^\ast \eqref{BSdiffB} =
O_P ( \sup_{g,h}N_{gh}^2\sup_g N_g^2 )$.
\end{proof}

\begin{lemma}
\label{Lemma:varhat}
Let \Cref{as:moments,as:X,as:eigen} be satisfied and let 
$\widehat\var ( \hat\btheta )$ be given by \eqref{var2sided}. 
Suppose also that either (i) \Cref{as:cluster} or (ii) 
\Cref{as:clusteralt} and $\lambda \geq 2$ is satisfied. Then, for 
an arbitrary conforming, non-zero vector $\bdelta$ and 
$\bdelta^\top \biH_k = \bib^\top = [ \bib_1^\top \otimes \bib_2^\top ]$,
\begin{align*}
&\widehat\var ( \bdelta^\top \hat\btheta ) - 2\sum_{g=1}^G 
\sum_{h_1=1}^{M_g}\sum_{h_2 \neq h_1}^{M_g} \bib_1^\top \bSigma_{gh_1}
\bib_1 \bib_2^\top \bSigma_{gh_2}\bib_2
= O_P\big( N^{1/\lambda}\sup_g N_g\sup_{g,h}N_{gh}^{3-1/\lambda}\big) \\
&\quad+O_P \big( \sup_{g,h} N_{gh}^2 \sup_g N_g^2 \big)
+O_P \bigg(\! \sup_{g,h} N_{gh} \sup_g N_g \Big( \sum_{g=1}^G 
\sum_{h_1=1}^{M_g} \sum_{h_2=1}^{M_g} \bib_1^\top\bSigma_{gh_1}\bib_1 
\bib_2^\top\bSigma_{gh_2} \bib_2 \Big)^{\!\!1/2}\, \bigg) 
\end{align*}
and
\begin{equation*}
\sum_{g=1}^G \sum_{h_1=1}^{M_g}\sum_{h_2 \neq h_1}^{M_g} 
\bib_1^\top \bSigma_{gh_1}\bib_1 \bib_2^\top \bSigma_{gh_2}\bib_2
\geq c \sum_{g=1}^G \sum_{h_1=1}^{M_g}\sum_{h_2=1}^{M_g} 
\bib_1^\top \bSigma_{gh_1}\bib_1 \bib_2^\top \bSigma_{gh_2}\bib_2 .
\end{equation*}
\end{lemma}

\begin{proof}
We give the proof of the first result for the univariate case, where 
$\widehat\var (\hat\theta )$ is given by \eqref{varfast}, and we show 
that
\begin{equation}
\label{varhatdiff}
\begin{aligned}
\widehat\var ( \hat\theta ) - 2&\sum_{g=1}^G 
\sum_{h_1=1}^{M_g}\sum_{h_2 \neq h_1}^{M_g} \sigma^2_{gh_1}
\sigma^2_{gh_2} = O_P\big( N^{1/\lambda}\sup_g N_g\sup_{g,h}
N_{gh}^{3-1/\lambda}\big) \\
&+ O_P \big( \sup_{g,h} N_{gh}^2 \sup_g N_g^2 \big)
+O_P \bigg(\! \sup_{g,h} N_{gh} \sup_g N_g \Big( \sum_{g=1}^G \Big(
\sum_{h=1}^{M_g} \sigma^2_{gh} \Big)^{\!2}\, \Big)^{\!\!1/2}\, \bigg) .
\end{aligned}
\end{equation}
The proof for the multivariate case is nearly identical, but with more 
complicated notation.

We decompose the left-hand side of \eqref{varhatdiff} as
\begin{align}
\label{proof2b}
&2\sum_{g=1}^G \sum_{h_1=1}^{M_g} \sum_{h_2\neq h_1}^{M_g} 
(\hat s_{gh_1}^2\hat s_{gh_2}^2 - s_{gh_1}^2 s_{gh_2}^2 )\\
\label{proof2a}
&+2\sum_{g=1}^G \sum_{h_1=1}^{M_g}\sum_{h_2\neq h_1}^{M_g} 
( s_{gh_1}^2 s_{gh_2}^2 -\sigma^2_{gh_1}\sigma^2_{gh_2} ) .
\end{align}

We first prove the result for \eqref{proof2a} under \Cref{as:cluster}. We use 
\eqref{tauvar} and \eqref{varfast} to write
\begin{equation}
\label{proof2c}
\eqref{proof2a}
= 4 \sum_{g=1}^G \sum_{h=1}^{M_g} (q_{1,gh} + q_{3,gh}),
\end{equation}
where $q_{1,gh}=\sigma_{gh}^2 \sum_{j=1}^{h-1}( s_{gj}^2 -
\sigma_{gj}^2 )$ and $q_{3,gh} = ( s_{gh}^2 - \sigma_{gh}^2 ) 
\sum_{j=1}^{h-1} s_{gj}^2$. Under \Cref{as:cluster} we have already 
proven in \eqref{qm} that $\sum_{g=1}^G \sum_{h=1}^{M_g} q_{1,gh} =O_P ( \sup_{g,h} N_{gh}^{3/2}
\sup_g N_g^{1/2} ( \sum_{g=1}^G \sigma_g^4 )^{1/2})$. 
The sequence $q_{3,gh}$ is a martingale difference with 
respect to the filtration $\mathcal F_{gh}$ defined just below 
\eqref{MDS}. When $1< \lambda <2$, we prove convergence in 
$L_\lambda$-norm. By the von Bahr-Esseen inequality, $\E \big| 
\sum_{g=1}^G \sum_{h=1}^{M_g} q_{3,gh} \big|^\lambda \leq 2 \sum_{g=1}^G
\sum_{h=1}^{M_g} \E |q_{3,gh} |^\lambda$, where $\E |q_{3,gh} |^\lambda 
\leq \E | s_{gh}|^{2\lambda} \E \big| \sum_{j=1}^{h-1} s_{gj}^2 
\big|^\lambda$, which was analyzed in \eqref{mom-sum}. The remainder 
of the proof for $q_{3,gh}$ with $1 < \lambda < 2$ is identical to that 
of the Lyapunov condition in \eqref{CLT1}, showing that $\sum_{g=1}^G 
\sum_{h=1}^{M_g} q_{3,gh} = O_P ( N^{1/\lambda}\sup_g N_g \sup_{g,h}
N_{gh}^{3-1/\lambda} )$.

Next, suppose $\lambda \geq 2$. We find that $\sum_{j=1}^{h-1}
s_{gj}^2 \leq \sum_{j=1}^{M_g}s_{gj}^2$ is a non-negative 
random variable, and hence is of order $O_P ( \E \sum_{j=1}^{M_g}
s_{gj}^2 ) = O_P( \sum_{h=1}^{M_g}\sigma_{gh}^2 )$. That is, 
there exists a constant $K < \infty$ and a set $\mathcal{A}$ with 
$P(\mathcal{A})\to 1$ on which $\sum_{j=1}^{h-1} s_{gj}^2 \leq K 
\sum_{h=1}^{M_g}\sigma_{gh}^2$. Then, on the set $\mathcal{A}$,
\begin{equation}
\label{Varq3c}
\E (q_{3,gh}^2 | \mathcal F_{g,h-1} )
= \Var ( s_{gh}^2) \Big( \sum_{j=1}^{h-1} s_{gj}^2 \Big)^{\!2} 
\leq K^2 \Big( \sum_{h=1}^{M_g}\sigma_{gh}^2\Big)^{\!2} \Var ( s_{gh}^2),
\end{equation}
and therefore
\begin{equation}
\label{Varq3}
\Var (q_{3,gh}) \leq cN_{gh}^4
\Big(\sum_{h=1}^{M_g}\sigma_{gh}^2\Big)^{\!2}
\end{equation}
by \Cref{LemmaS}. Using \eqref{Varq3} and the fact that $q_{3,gh}$ is a 
martingale difference sequence, it follows that, on the set $\mathcal A$,
\begin{equation*}
\var \bigg( \sum_{g=1}^G \sum_{h=1}^{M_g} q_{3,gh} \bigg)
= \sum_{g=1}^G \sum_{h=1}^{M_g} \var (q_{3,gh} )
\leq c \sup_{g,h} N_{gh}^3 \sup_g N_g \Big( \sum_{h=1}^{M_g}
\sigma_{gh}^2\Big)^{\!2} .
\end{equation*}
This shows the required result for $q_{3,gh}$ on the set $\mathcal{A}$ 
when $\lambda \geq 2$. Because $P(\mathcal{A})\to 1$, this completes the
proof for \eqref{proof2a} under \Cref{as:cluster}.

We now prove the result for \eqref{proof2a} under \Cref{as:clusteralt} 
and $\lambda \geq 2$. We again apply the decomposition in 
\eqref{proof2c}. Define $Q_{m,g}=\sum_{h=1}^{M_g}q_{m,gh}$ for $m=1,3$, 
which are both independent across $g$ by \Cref{as:clusteralt}. For 
$Q_{1,g}$ we note that $\sum_{j=h+1}^{M_g} \sigma_{gj}^2 \leq 
\sum_{j=1}^{M_g} \sigma_{gj}^2$ and apply the Cauchy-Schwarz inequality 
such that
\begin{equation*}
\E (Q_{1,g}^2 ) \leq \Big( \sum_{h=1}^{M_g} \sigma_{gh}^2 \Big)^{\!2}
\E \Big( \sum_{h=1}^{M_g} | s_{gh}^2 - \sigma_{gh}^2 | \Big)^{\!2}
\leq \Big( \sum_{h=1}^{M_g} \sigma_{gh}^2 \Big)^{\!2}
\Big(\sum_{h=1}^{M_g}(\E ( s_{gh}^2-\sigma_{gh}^2 )^2 )^{1/2}\Big)^{\!2},
\end{equation*}
where last factor on the right-hand side is $O ( \sup_{g,h} N_{gh}^2 
\sup_g N_g^2 )$ by \Cref{LemmaS}. Because $Q_{1,g}$ has mean zero and 
is independent across $g$, it follows that
\begin{equation*}
\var \Big( \sum_{g=1}^G Q_{1,g} \Big) = 
\sum_{g=1}^G \E (Q_{1,g}^2) \leq c\tk \sup_{g,h} N_{gh}^2 \sup_g N_g^2
\sum_{g=1}^G \Big( \sum_{h=1}^{M_g} \sigma_{gh}^2 \Big)^{\!2} ,
\end{equation*}
which proves the result for $Q_{1,g}$. For $Q_{3,g}$ we note that there 
exists a constant $K<\infty$ and a set $\mathcal{A}$ with 
$P(\mathcal{A}) \to 1$ such that, on $\mathcal{A}$, it holds that 
$\sum_{j=1}^{h-1} s_{gj}^2 \leq K \sum_{j=1}^{M_g}\sigma_{gj}^2$.
We can then apply the same proof as for $Q_{1,g}$. This completes the 
proof for \eqref{proof2a} under \Cref{as:clusteralt}.

To prove the result for \eqref{proof2b}, we use \eqref{varfast} and 
\eqref{shat} to write
\begin{align}
\nonumber
\eqref{proof2b}
={}& 2 \sum_{g=1}^G \sum_{h_1=1}^{M_g} \sum_{h_2\neq h_1}^{M_g}
( \hat s_{gh_1} \hat s_{gh_2} + s_{gh_1} s_{gh_2} )
( \hat s_{gh_1} \hat s_{gh_2} - s_{gh_1} s_{gh_2} ) \\
\label{proof3a}
={}&  2 (\hat\beta - \beta_0)^2 \sum_{g=1}^G \sum_{h_1=1}^{M_g} 
\sum_{h_2\neq h_1}^{M_g}( \hat s_{gh_1} \hat s_{gh_2} 
+s_{gh_1} s_{gh_2} )
\sum_{i=1}^{N_{gh_1}} x_{gh_1i}^2 \sum_{j=1}^{N_{gh_2}} x_{gh_2j}^2  \\
\label{proof3b}
&+4 (\hat\beta - \beta_0)\sum_{g=1}^G \sum_{h_1=1}^{M_g} 
\sum_{h_2\neq h_1}^{M_g}( \hat s_{gh_1} \hat s_{gh_2} + 
s_{gh_1} s_{gh_2}) s_{gh_1} \sum_{i=1}^{N_{gh_2}} x_{gh_2i}^2 .
\end{align}
By another application of \eqref{shat} followed by straightforward 
application of \eqref{Opbeta}, \eqref{boundX}, and \Cref{LemmaS}, it 
follows that \eqref{proof3a} is of order $O_P(\sup_{g,h}N_{gh}^2
\sup_g N_g^2 )$.

For \eqref{proof3b}, we apply again \eqref{shat} and write
\begin{align}
\label{proof4a}
\eqref{proof3b} ={}& 8 (\hat\beta - \beta_0)\sum_{g=1}^G 
\sum_{h_1=1}^{M_g} \sum_{h_2\neq h_1}^{M_g} s_{gh_1}^2 
s_{gh_2}\sum_{i=1}^{N_{gh_2}} x_{gh_2i}^2 \\
\label{proof4b}
&-4 (\hat\beta - \beta_0)^2\sum_{g=1}^G \sum_{h_1=1}^{M_g} 
\sum_{h_2\neq h_1}^{M_g} \Big( s_{gh_1}\sum_{i=1}^{N_{gh_2}}x_{gh_2i}^2
+ s_{gh_2}  \sum_{j=1}^{N_{gh_1}} x_{gh_1j}^2 \Big) 
s_{gh_1} \sum_{\ell=1}^{N_{gh_2}} x_{gh_2\ell}^2 \\
\label{proof4c}
&-4 (\hat\beta - \beta_0)^3 \sum_{g=1}^G \sum_{h_1=1}^{M_g} 
\sum_{h_2\neq h_1}^{M_g} s_{gh_1}
\Big( \sum_{j=1}^{N_{gh_1}} x_{gh_1j}^2 \Big)
\Big( \sum_{i=1}^{N_{gh_2}} x_{gh_2i}^2 \Big)^{\!2} .
\end{align}
Direct application of \eqref{Opbeta}, \eqref{boundX}, and 
\Cref{LemmaS} shows that \eqref{proof4b} is $O_P (\sup_{g,h} N_{gh}^2 
\sup_g N_g^2 )$ and that \eqref{proof4c} is $O_P (N^{-1/2}\sup_{g,h} 
N_{gh}^2 \sup_g N_g^{5/2} ) = O_P (\sup_{g,h} N_{gh}^2 \sup_g N_g^2 )$. 
Finally, for the right-hand side of \eqref{proof4a}, we apply the 
Cauchy-Schwarz inequality,
\begin{align}
\nonumber
\eqref{proof4a}^2 &\leq 64 (\hat\beta - \beta_0)^2
\bigg( \sum_{g=1}^G \Big( \sum_{h_1=1}^{M_g} s_{gh_1}^2
\sum_{i=1}^{N_{gh_1}} x_{gh_1i}^2 \Big)^{\!2} \bigg)
\bigg( \sum_{g=1}^G \Big( \sum_{h_2=1}^{M_g}s_{gh_2}\Big)^{\!2}
\bigg) \\
\label{proof5}
&\leq 64 (\hat\beta - \beta_0)^2 
\Big( \sup_{g,h} \sum_{i=1}^{N_{gh}} x_{ghi}^2 \Big)^{\!2}
\bigg( \sum_{g=1}^G \Big( \sum_{h_1=1}^{M_g} s_{gh_1}^2
 \Big)^{\!2} \bigg)
\bigg( \sum_{g=1}^G \Big( \sum_{h_2=1}^{M_g} s_{gh_2}\Big)^{\!2} \bigg) .
\end{align}
As in \eqref{Varq3c}, we find that the penultimate factor on the 
right-hand side of \eqref{proof5} is bounded by a constant times 
$\sum_{g=1}^G ( \sum_{h=1}^{M_g} \sigma_{gh}^2 )^2$ on a set 
$\mathcal{A}$ with $P(\mathcal{A})\to 1$. The last factor on the 
right-hand side of \eqref{proof5} is a non-negative random variable and 
hence is of order $O_P ( \E \sum_{g=1}^G ( \sum_{h=1}^{M_g} 
s_{gh} )^2 ) = O_P ( \sum_{g=1}^G \E s_g^2 ) = O_P ( N \sup_g N_g )$ 
by \Cref{LemmaS}. Combining these results and using \eqref{Opbeta} and 
\eqref{boundX}, we find that
\begin{equation*}
\eqref{proof4a} = O_P \bigg(\! \sup_g N_g \sup_{g,h}N_{gh} \Big( 
\sum_{g=1}^G \big( \sum_{h=1}^{M_g} \sigma_{gh}^2\big)^{\!2} \Big)^{\!\!1/2} 
\bigg) ,
\end{equation*}
which proves the required result for \eqref{proof4a}, and hence for 
\eqref{proof3b} and~\eqref{proof2b}.

To prove the second result of the lemma we write the left-hand side as
\begin{align*}
2\sum_{g=1}^G \sum_{h_1=1}^{M_g}\sum_{h_2=1}^{M_g} 
&\bib_1^\top \bSigma_{gh_1}\bib_1 \bib_2^\top \bSigma_{gh_2}\bib_2
-2 \sum_{g=1}^G \sum_{h=1}^{M_g} \bib_1^\top \bSigma_{gh}\bib_1 
\bib_2^\top \bSigma_{gh}\bib_2 \\
&= 2\sum_{g=1}^G \sum_{h_1=1}^{M_g}\sum_{h_2=1}^{M_g} 
\bib_1^\top \bSigma_{gh_1}\bib_1 \bib_2^\top \bSigma_{gh_2}\bib_2
\bigg(\! 1 - \frac{\sum_{h=1}^{M_g} \bib_1^\top \bSigma_{gh}\bib_1 
\bib_2^\top\bSigma_{gh}\bib_2}%
{ \sum_{h_1=1}^{M_g}\sum_{h_2=1}^{M_g} \bib_1^\top \bSigma_{gh_1}\bib_1 
\bib_2^\top \bSigma_{gh_2}\bib_2} \bigg) \\
&\geq 2\sum_{g=1}^G \sum_{h_1=1}^{M_g}\sum_{h_2=1}^{M_g} 
\bib_1^\top \bSigma_{gh_1}\bib_1 \bib_2^\top \bSigma_{gh_2}\bib_2
\bigg(\! 1 - \frac{\sup_h \bib_2^\top \bSigma_{gh}\bib_2}%
{\sum_{h=1}^{M_g} \bib_2^\top \bSigma_{gh}\bib_2} \bigg) ,
\end{align*}
where the inequality is due to $\sum_{h=1}^{M_g} \bib_1^\top 
\bSigma_{gh}\bib_1 \bib_2^\top\bSigma_{gh}\bib_2 \leq
( \sup_h \bib_2^\top \bSigma_{gh}\bib_2 ) \sum_{h=1}^{M_g} 
\bib_1^\top \bSigma_{gh}\bib_1$. The result follows because 
$\sup_{g,h} \bib_2^\top \bSigma_{gh}\bib_2 / (\bib_2^\top 
\sum_{h=1}^{M_g} \bSigma_{gh} \bib_2 ) \leq \sup_{g,h}\omega_{\max} 
(\bSigma_{gh} (\sum_{h=1}^{M_g} \bSigma_{gh})^{-1} )<~1$ by 
\Cref{as:eigen}.
\end{proof}

\bibliography{mnwc}
\addcontentsline{toc}{section}{\refname}

\end{document}